\documentclass[jcs]{iosart1c}

\usepackage[T1]{fontenc}
\usepackage[utf8]{inputenc}
\usepackage{msc}

\usepackage{pifont}
\usepackage[numbers]{natbib}
\usepackage{xspace}
\usepackage{amssymb}
\usepackage{MnSymbol}         
\usepackage{amsmath}
\usepackage{amsthm}
\usepackage{dcolumn}
\usepackage{graphicx}
\usepackage{color}
\usepackage{url}
\usepackage{cleveref}
\usepackage{listings}
\usepackage{multirow}
 \usepackage{url}

\newcommand{\error}{\mathsf{error}}

\newcommand{\cmark}{\ding{51}}%
\newcommand{\xmark}{\ding{53}}%

\newcommand{\verif}{\cmark}
\newcommand{\holds}{\bf safe}
\newcommand{\nope}{{\xmark}}
\newcommand{\attaque}{{attack}}
\newcommand{\varI}{\mathsf{i}}
\newcommand{\varN}{\mathsf{n}}

\newcommand{\SD}{\mathsf{SD}}
\newcommand{\Oage}{\mathsf{O}_\mathrm{age}}
\newcommand{\Ocheck}{\mathsf{O}_\mathrm{check}}

\usepackage{booktabs}           % better tabular
\usepackage{esvect}
\newcommand{\vect}[1]{\overline{#1}}

\newcommand{\ukano}{\ensuremath{\mathsf{UKano}}\xspace}
\newcommand{\proverif}{\ensuremath{\mathsf{ProVerif}}\xspace}
\newcommand{\tamarin}{\ensuremath{\mathsf{Tamarin}}\xspace}

\newcommand{\Feldhofer}{\mathsf{Fh}}

\newcommand{\ie}{\emph{i.e.}\xspace}
\newcommand{\eg}{\emph{e.g.}\xspace}
\newcommand{\etc}{etc.\xspace}

\newcommand{\N}{\mathcal{N}}
\newcommand{\X}{\mathcal{X}}
\newcommand{\W}{\mathcal{W}}
\newcommand{\T}{\mathcal{T}}

\newcommand{\priv}{\mathsf{priv}}
\newcommand{\mac}{\mathsf{mac}}
\newcommand{\senc}{\mathsf{enc}}
\newcommand{\sdec}{\mathsf{dec}}

\newcommand{\h}{\mathsf{h}}
\newcommand{\pub}{\mathsf{pub}}
\newcommand{\enc}[2]{\mathsf{enc}(#1,#2)}
\newcommand{\dec}[2]{\mathsf{dec}(#1,#2)}

\newcommand{\pk}[1]{\mathsf{pk}(#1)}
\newcommand{\proj}{\mathsf{proj}}

\newcommand{\projl}[1]{\proj_1(#1)}
\newcommand{\projr}[1]{\proj_2(#1)}
\newcommand{\eq}{\mathsf{eq}}

\newcommand{\ok}{\mathsf{ok}}
\newcommand{\ffun}{\mathsf{f}}
\newcommand{\gfun}{\mathsf{g}}
\newcommand{\fneq}{\mathsf{neq}}

\newcommand{\Ch}{\mathcal{C}}
\newcommand{\Out}{\mathsf{out}}
\newcommand{\In} {\mathsf{in}}
\newcommand{\Let}{\mathsf{let}}
\newcommand{\Else}{\mathsf{else}}

\newcommand{\Then}{\mathsf{then}}
\newcommand{\new}{\mathsf{new}\,}
\newcommand{\rep}{!}
\newcommand{\rec}{\mathbin{\rotatebox[origin=c]{180}{$!$}}}

\newcommand{\fv}{\mathit{fv}}
\newcommand{\fn}{\mathit{fn}}
\newcommand{\ta}{\mathsf{ta}}
\newcommand{\annot}[1]{[#1]}

\newcommand{\restrict}[2]{#1|_{#2}}

\newcommand{\fail}{\mathsf{fail}}

\newcommand{\p}{\mathcal{P}}
\newcommand{\q}{\mathcal{Q}}
\newcommand{\refer}{\mapsto} % \triangleright

\newcommand{\trace}{\mathsf{trace}}

\makeatletter
\def\rightarrowfillstar@{\arrowfill@\relbar\relbar{\rightarrow\smash{^*}}}
\newcommand{\xrightarrowstar}[2][]{\ext@arrow
  0{13}{15}8\rightarrowfillstar@{#1}{#2}}
\newcommand{\lrstep}{\@ifstar{\xrightarrowstar}{\xrightarrow}}
\makeatother

\makeatletter
\def\Rightarrowfillstar@{\arrowfill@=={\Rightarrow\smash{^*}}}
\newcommand{\xRightarrowstar}[2][]{\ext@arrow
  0{13}{15}8\Rightarrowfillstar@{#1}{#2}}
\newcommand{\xRightarrow}[2][]{\ext@arrow
  0{13}{15}8\Rightarrowfill@{#1}{#2}}
\newcommand{\LRstep}{\@ifstar{\xRightarrowstar}{\xRightarrow}}
\makeatother

\newcommand{\eint}{\approx}

\newcommand{\tr}{\mathsf{tr}}

\newcommand{\taut}{\tau_\mathsf{then}} 
\newcommand{\taue}{\tau_\mathsf{else}} 
\newcommand{\obs}{\mathsf{obs}}
\newcommand{\upto}{\stackrel{\tau}{=}}

\newcommand{\vars}{\mathit{vars}} 
\newcommand{\dom}{\mathit{dom}}

\newcommand{\E}{\mathsf{E}}

\newcommand{\redc}{\mathrel{\Downarrow}} % relation calcul
\newcommand{\redv}{{\Downarrow}} % à utiliser pour version unaire
\newcommand{\redcb}{\mathrel\nDownarrow} % relation calcul fail

\newcommand{\theo}{=_\E}        %théorie équationnelle
\newcommand{\ini}{\mathcal{I}}
\newcommand{\res}{\mathcal{R}}
\renewcommand{\k}{\ensuremath{\vect k}}
\newcommand{\nI}{\ensuremath{\vect n_I}}
\newcommand{\nR}{\ensuremath{\vect n_R}}

\newcommand{\proto}{\Pi}    % a protocol (tuple (k,n,I,R)
\newcommand{\pM}{\mathcal{M}_\Pi}    % whole system multiple session
\newcommand{\pS}{\mathcal{S}_\Pi}    % whole system, once session per identity
\newcommand{\pMa}{\mathcal{M}_{\Pi,\vect{{id}}}} % whole system multiple
\newcommand{\agent}{\mathcal{A}}
\newcommand{\id}{\mathit{id}}
\newcommand{\idzero}{\mathsf{id}_0}

\newcommand{\monid}{\mathit{id}}
\newcommand{\cred}{\mathit{cred}}
\newcommand{\skex}{\mathsf{sk}}
\newcommand{\skIssuer}{\mathsf{sk_I}}
\renewcommand{\pk}{\mathsf{pk}}
\newcommand{\sign}{\mathsf{sign}}
\newcommand{\checksign}{\mathsf{check}_\sign}
\newcommand{\zk}{\mathsf{zk}}
\newcommand{\checkzk}{\mathsf{check}_\zk}
\newcommand{\publiczk}{\mathsf{public}_\zk}
\newcommand{\tuple}{\mathsf{tuple}}
\newcommand{\DAA}{\mathsf{DAA}}

\newcommand{\aagent}{A} % denotes \aini or \ares
\newcommand{\aini}{I}
\newcommand{\ares}{R}
\newcommand{\agents}{\mathcal{A}} % set of all agents (anotations)

\newcommand{\ideaf}[0]{\Phi_{\mathrm{ideal}}}% idealizations of frames
\newcommand{\ideam}[1]{\mathrm{ideal}(#1)}% idealizations of messages
\newcommand{\ideamstef}{\mathrm{ideal}}% idealizations of messages

\newcommand{\fr}{\mathrm{fr}}
\newcommand{\frr}{\mathrm{fr'}}

\newcommand{\dagR}{\dag\hspace{-0.1em}_R}
\newcommand{\dagI}{\dag\hspace{-0.1em}_I}
\newcommand{\con}{\mathrm{Co}}
\newcommand{\K}{\mathcal{K}}

\newcommand{\choice}[2]{\mathsf{choice}[#1,#2]}
\newcommand{\ideal}{\mathsf{ideal}}

\newtheorem{example}{Example}
\newtheorem{definition}{Definition}

\newtheorem{lemma}{Lemma}
\newtheorem{proposition}{Proposition}

\newtheorem{theorem}{Theorem}

\newcolumntype{d}[1]{D{.}{.}{#1}}

\let\red\undefined
\firstpage{0} \lastpage{0} \volume{0} \pubyear{0000}

\begin{document}
\begin{frontmatter}                           % The preamble begins here.

%
%\pretitle{Pretitle}
\title{A method for unbounded verification of privacy-type
  properties} %\thanks{\stef{Remercier POPSTAR.}}}

\runningtitle{A method for unbounded verification of privacy-type
  properties}
%\subtitle{Subtitle}

\author[A]{\fnms{Lucca} \snm{Hirschi}\thanks{This work
  was conducted when Lucca Hirschi was working at LSV, ENS Paris-Saclay \& Universit\'e Paris-Saclay, France
  and then at ETH Zurich, Switzerland.}}
\author[B]{\fnms{David} \snm{Baelde}}
\author[C]{\fnms{St\'ephanie} \snm{Delaune}\thanks{Corresponding
    author. E-mail: stephanie.delaune@irisa.fr.  
This work has received funding from the European Research Council
(ERC) under the EU’s Horizon 2020 
research and innovation program (grant agreement No 714955-POPSTAR) 
and the ANR project SEQUOIA ANR-14-CE28-0030-01.}}
\runningauthor{L. Hirschi et al.}
\address[A]{Inria \& LORIA, France\\
E-mail: lucca.hirschi@inria.fr}
\address[B]{LSV, ENS Paris-Saclay, CNRS \& Universit\'e Paris-Saclay, France\\
E-mail: david.baelde@lsv.fr}
\address[C]{Univ Rennes, CNRS, IRISA, France\\
E-mail: stephanie.delaune@irisa.fr}

\begin{abstract}
  In this paper, we consider the problem of verifying anonymity and 
  unlinkability in the symbolic model, where protocols are represented
  as processes in a variant of the applied pi calculus, notably used
  in the \proverif tool.
  Existing tools and techniques do not allow to
  verify directly these properties, expressed as behavioral 
  equivalences. We propose a different approach:
  we design two conditions on protocols
  which are sufficient to ensure anonymity and unlinkability, and
  which can then be effectively checked automatically using \proverif.
  Our two conditions correspond to two broad classes of attacks
  on unlinkability, \ie data and control-flow leaks.
  This theoretical result is general enough that it applies to a wide
  class of protocols based on a variety of cryptographic primitives.  In particular, using our tool, \ukano, we
  provide the first formal security proofs of protocols such as BAC and
  PACE (e-passport), Hash-Lock (RFID authentication), etc.
  Our work has also lead to the discovery of new attacks,
  including one on the LAK protocol (RFID authentication)
  which was previously claimed to be unlinkable (in a weak sense).
\end{abstract}

%%% Local Variables:
%%% mode: latex
%%% TeX-master: "../../thesis"
%%% End:


\end{frontmatter}

\section{Introduction}
\label{sec:intro}

\renewcommand{\textsf}[1]{$\mathsf{#1}$}

%%%%% Contexte général

Security protocols aim at securing communications over various types
of insecure networks
(\eg web, wireless devices)
where dishonest users may listen to communications and interfere with them.
A \emph{secure communication} has a different meaning depending on the
underlying application. 
It ranges from the confidentiality of data (medical files, secret
keys, \etc) to, \eg 
verifiability in electronic voting systems. Another example of a security notion is privacy. 
In this paper, we focus on two privacy-related properties, namely
unlinkability (sometimes called untraceability), and anonymity.
These two notions are informally defined in the
ISO/IEC standard 15408~\cite{ISO15408} 
 as follows:
\begin{itemize}
\item Unlinkability aims at
  \emph{ensuring that a user may make multiple uses of a service or resource
  without others being able to link these uses together.}
\item Anonymity aims at
\emph{ensuring that a user may use a service or resource
without disclosing its identity.} 
\end{itemize}
Both are critical for instance for Radio-Frequency Identification
Devices (RFID)  and are thus extensively studied in that context
(see, \eg~\cite{van2008attacks}
for a survey of attacks on this type of protocols), but they
are obviously not limited to it.

\smallskip{}

%%%% Formal methods

One extremely successful approach when designing and analyzing
security protocols is the use of formal verification, \ie
the development of
rigorous frameworks and techniques to analyze protocols.
This approach has notably lead to the discovery of a flaw
 in the Single-Sign-On protocol used \emph{e.g.} by Google Apps. It has
been shown that a malicious application could very easily access to
any other application (\emph{e.g.} Gmail or Google Calendar) of their
users~\cite{DBLP:conf/ccs/ArmandoCCCT08}. This flaw has been found when analyzing the protocol using
formal methods, abstracting messages by a term algebra and using the
\textsf{Avantssar} validation platform. Another example is a flaw on vote-privacy discovered
during the formal and manual analysis of an electronic voting
protocol~\cite{JCS2012-Ben}.
All these results have been obtained  using \emph{formal symbolic
models}, where most of
the cryptographic details are ignored using abstract structures.
The techniques used in symbolic models have become mature
and several tools for protocol
verification are  nowadays available,
\emph{e.g.} the \textsf{Avantssar} platform~\cite{avantssar-tacas12},
the \textsf{Tamarin} prover~\cite{Tamarin},
and the \proverif~tool~\cite{blanchetcsfw01}.

\smallskip{}

%%% Equivalence c'est assez nouveau
Unfortunately, most of these results and tools focus on trace properties, that is,
statements that something bad never occurs on any execution trace of a
protocol. Secrecy and authentication are typical examples of trace
properties: a data remains confidential if, for any execution, the
attacker is not able to produce the data. However, privacy properties like
unlinkability and anonymity are generally not defined as trace properties.
Instead, they are usually defined as the fact that an observer
cannot distinguish between two situations, which requires a notion of
behavioural equivalence.
%\dam{Répétition}
%\david{Roughly, two protocols are equivalent %($P \approx Q$) 
%if
%an attacker  cannot  observe any difference between them.}
Based on such a notion of equivalence, several definitions of
privacy-type properties have been proposed (\eg~\cite{arapinis-csf10,kostas-csf10} for
unlinkability, and \cite{DKR-jcs08,BackesHM08} for vote-privacy).
In this paper, we consider the
well-established definitions of strong unlinkability and anonymity as defined
in~\cite{arapinis-csf10}. 
They have notably been used to establish privacy for various protocols either by hand or
using ad hoc encodings (\eg 
eHealth protocol~\cite{dong2012formal},
mobile telephony~\cite{arapinis2012new,arapinis2014privacy}).
We provide a brief comparison with alternative definitions in Section~\ref{sec:priv:prop:discu}.

\smallskip{}

% Equivalence c'est difficile
Considering an unbounded number of sessions, the problem of deciding
whether {a protocol satisfies an equivalence property} is 
undecidable even for a very limited fragment of {protocols} (see, \eg
\cite{CCD-tocl15}).
Bounding the number of sessions suffices to retrieve decidability for
standard primitives (see, \eg~\cite{baudet-ccs2005,cheval-ccs2011}). 
However, analysing a protocol for a fixed (often low) number of
sessions does not allow to prove security. 
Moreover, in the case of equivalence properties, existing tools scale
badly and can only analyse protocols for a very limited number of 
sessions, typically 2 or 3. 
Another approach consists in implementing a procedure that is not
guaranteed to terminate. 
This is in particular the case of \proverif, a well-established tool
for checking security of protocols. 
\proverif~is able to check a strong notion of equivalence (called 
diff-equivalence) between processes
that share the same structure. 
%\lum{\small Il est possible d'ajouter une sous-section 
%``problème'' qui décrit complètement le problème de la diff (comme ds
%intro de la part C)}
% S.D.: j'ai l'impression que cela n'est pas indispensable. On
% pourrait aussi mettre une ref. pour pointer l'endroit dans ta these
% ou tu expliques cela en detail. Je crois aussi que Vincent avait
% detaille ce probleme sur BAC dans sa these.
Despite recent improvements on diff-equivalence
checking~\cite{ChevalBlanchetPOST13} intended  to prove unlinkability of the BAC protocol (used in e-passport),
\proverif~still cannot be used off-the-shelf to establish unlinkability
properties, and therefore cannot conclude on most
of the case studies presented in Section~\ref{sec:casestudies}.
Recently, similar approaches have been
implemented in two other tools, namely \textsf{Tamarin}~\cite{basin2015automated} and
\textsf{Maude{-}NPA}~\cite{santiago2014formal}. They
are based on a notion of diff-equivalence, and therefore suffer from
the same drawbacks.

%%
%% OUR CONTRIBUTIONS
%%

\paragraph*{\bf Our contribution.}
We believe that looking at trace equivalence of any pair of protocols is a too general problem
and that much progress can be expected when one focuses on a few privacy goals
and a class of protocols only (yet large and generic enough).
We follow this different approach. We aim at proposing
sufficient conditions that can be automatically checked, and that
imply unlinkability and anonymity for a large class of security protocols.
The success of our solution will be measured by confronting it to many real-world case studies.

More precisely, we identify a large class of 2-party protocols
(simple else branches, arbitrary cryptographic primitives) and we devise
two conditions called {\em frame opacity}
and {\em well-authentication}
that imply unlinkability and anonymity for an unbounded number of
sessions. We show how these two conditions can be automatically
checked using \eg the \proverif~tool, and we provide tool support for
that. 
Using our tool \ukano (built on top of \proverif),
we have automatically analysed several protocols, among them the Basic Access Control (BAC)
protocol as well as the Password Authenticated Connection Establishment (PACE) protocol that are both
used in e-passports.
We notably establish the first proof of unlinkability
for ABCDH~\cite{alpar2013secure} and
for the BAC protocol followed by the Passive
Authentication (PA) and Active Authentication (AA) protocols.
We also report on an attack that we found on the PACE protocol, and another
one that we found on the LAK protocol~\cite{LAK'06} whereas
it is claimed untraceable in~\cite{van2008attacks}.
It happens that our conditions are rather tight, and we
believe that the overall methodology and proof method
could be used for other classes of protocols and other
privacy goals.

\paragraph*{\bf Our sufficient conditions.}
%%%% Idée des conditions
We now give an intuitive overview of our two sufficient conditions,
namely \emph{frame opacity} and \emph{well-authentication}.
In order to do this, assume that we 
want to design a mutual authentication protocol between a tag~$T$
and a reader~$R$ based on symmetric encryption, and we want this
protocol to be unlinkable.
We assume that~$k$ is a symmetric key shared between~$T$ and~$R$. 

\smallskip{}

A first attempt to design such a protocol is
presented using Alice \& Bob notation as follows  ($n_R$ is a fresh nonce):
 $$
  \begin{array}{rrl}
    1.& R \to T: & n_R\\[1mm]
    2.& T \to R: & \{n_R\}_k
  \end{array}
  $$
This first attempt based on a challenge-response scheme is actually linkable.
Indeed, an active attacker who systematically intercepts the nonce~$n_R$ and replaces it by a
constant will be able to infer whether the same tag has been used in
different sessions or not by comparing the answers he receives.
Here, the tag is linkable because, for a certain 
behaviour (possibly malicious) of the attacker, some relations between messages
leak information about the agents that are involved in the execution.
Our first condition, namely \emph{frame opacity},  actually checks
that all outputted messages have only relations
that only depend on what is already observable.
Such relations can therefore not be exploited by the attacker to learn
anything new about the involved agents.

\smallskip{}

Our second attempt takes the previous attack into account and randomises
the tag's response %(using a nonce $n_T$) 
and should achieve mutual authentication
by requiring that the reader must answer to the challenge $n_T$. This
protocol can be as follows:
$$
  \begin{array}{rrl}
    1.& R \to T: & n_R\\[1mm]
    2.& T \to R: & \{n_R,n_T\}_k\\[1mm]
    3.& R \to T: & \{n_T\}_k
%4.& T \to R: & \ldots
  \end{array}
$$
Here,
Alice \& Bob notation shows its limit. It does not specify how
the reader and the tag are supposed to check that the 
messages they received are of the expected form, and how they should react when the
messages are not well formed. This has to be precisely defined, since 
unlinkability depends on it.
For instance, assume the tag does not check that the message he
receives at step~$3$ contains the nonce~$n_T$. We assume that it only
checks 
that the received message is an encryption with its own key~$k$, and
it aborts the session otherwise.
In such a flawed implementation, an active attacker can   eavesdrop
a message $\{n_T\}_{k}$ sent by $R$ to a tag~$T$, and try to inject
this message at the third step of another session played by $T'$. The tag~$T'$ will react by either
 aborting or by continuing the execution of this protocol. Depending
 on the reaction of the tag, the attacker will be able to infer if~$T$ and~$T'$ are the same tag or not.

In this example, the attacker adopts a malicious behaviour %(MiM)
that is not detected immediately by the tag who keeps executing the protocol.
The fact that the tag passes successfully a conditional reveals crucial
information about the agents that are involved in the execution. 
Our second condition, namely \emph{well-authentication}, basically
requires that when an execution deviates from the honest one, the
agents that are involved cannot successfully pass a conditional,
thus avoiding the leak of the binary information success\slash failure.

\smallskip{}

Our main theorem states that these two conditions, frame opacity and
well-authentication, are actually sufficient to ensure both unlinkability and anonymity.
This theorem is of interest as our two conditions are fundamentally 
simpler than the targeted properties:
frame opacity can be expressed and established relying on diff-equivalence (without
the aforementioned precision issue) and well-authentication
is only a conjunction of reachability properties.
In fact, they are both in the scope of existing automatic verification
tools like \proverif and \tamarin.

\paragraph*{\bf Some related work.}
The precision issue of diff-equivalence is well-known (acknowledged \eg
in \cite{DRS-ifiptm08,ChevalBlanchetPOST13,2016-verifying-observational-equivalence,delaune2016survey}).
So far, the main approach that has been developed to solve this issue consists
in modifying the notion of diff-equivalence to get closer to trace equivalence.
For instance,
the swapping technique introduced in~\cite{DRS-ifiptm08} and formally 
justified in~\cite{2016-verifying-observational-equivalence} allows
to relax constraints imposed by diff-equivalence in specific situations,
namely in process algebras featuring a notion of phase, often used for 
modelling e-voting protocols.
Besides, the limitation of the diff-equivalence w.r.t.~conditional evaluations
has been partially addressed in~\cite{ChevalBlanchetPOST13}
by pushing away the evaluation of some conditionals into terms.
Nevertheless, the problem remains in general and the limitation described above
is not addressed by those works (incidentally, it is specifically 
acknowledged for the case of the BAC protocol in~\cite{ChevalBlanchetPOST13}).
%
% Since our approach to the problem is non-standard, there are not much related works.
We have chosen to follow a novel approach
in the same spirit as the one presented
in~\cite{kostas-csf10}.
% However, the class of protocols we target is
% quite different.
However, \cite{kostas-csf10} only
considers a very restricted class of protocols (single-step
protocols that only use hash functions), while
% as cryptographic primitives). 
we target more complex protocols.

This paper essentially subsumes the conference paper that has been
published in 2016~\cite{HBD-sp16}.
Compared to that earlier work, we have greatly generalized the scope of our method and improved its mechanization.
First, we consider more protocols, including protocols where 
one party has a single identity (\eg DAA) as well as protocols where 
sessions are executed sequentially instead of concurrently (\eg 
e-passport scenarios).
Second, we consider a much more general notion of frame opacity, which
enables the analysis of more protocols.
As a result of these two improvements, we could apply our method to
more case studies (\eg DAA, ABCDH).

\paragraph*{\bf Outline.}
In Section~\ref{sec:protocol}, we present our model inspired from 
the applied pi calculus as well as the class of protocols we consider.
We then introduce in Section~\ref{sec:properties} the
notion of trace equivalence that we then use to formally define
the two privacy properties we study in this paper: unlinkability and anonymity.
%  Our conditions, namely frame
% opacity and well-authentication, as well as our main theorem.
Our two conditions (frame opacity and well-authentication) and our main theorem are presented in Section~\ref{sec:approach}.
Finally, we discuss how to mechanize the verification of our conditions in 
Section~\ref{sec:mechanization} and present
our case studies in Section~\ref{sec:casestudies}, before concluding
in Section~\ref{sec:conclusion}.
A detailed proof of our main result is provided in Appendix.

\section{Modelling protocols}
\label{sec:protocol}

We model security protocols using a process algebra inspired
from the applied pi calculus~\cite{AbadiFournet2001}.
More specifically, we consider a calculus close to the one which is
used 
in the \proverif~tool~\cite{BlanchetAbadiFournetJLAP08}.
Participants  are modeled as processes, and the
communication between  them is modeled by means of the exchange of
messages that are represented by a term algebra.

%%
%% Term algebra
%%

\subsection{Term algebra}
\label{subsec:term}

We consider an infinite set $\N$ of \emph{names} which are used to represent
keys and nonces, and two infinite and disjoint sets of \emph{variables},
denoted~$\X$ and~$\W$. Variables in~$\X$ will typically be used to
refer to unknown parts of messages expected by participants,
while variables in~$\W$, called \emph{handles}, will be used to store messages learned by the 
attacker.
We assume a \emph{signature}~$\Sigma$, \ie a set of function
symbols together with their arity, split
into \emph{constructor} and \emph{destructor} symbols, \ie $\Sigma =
\Sigma_c \sqcup \Sigma_d$. 

Given a signature $\mathcal{F}$ and a set of initial data
$A$, we denote by $\T(\mathcal{F},A)$ the set of terms built
from elements of $A$ by applying function symbols in $\mathcal{F}$.
Terms of $\T(\Sigma_c, \N \cup \X)$ will be called \emph{constructor
  terms}.
We note $\vars(u)$ the set of variables that occur in a term $u$.
A \emph{message} is a constructor term $u$ that is \emph{ground}, \ie
such that $\vars(u) = \emptyset$.
We denote by $\vect x$, $\vect n$, $\vect u$, $\vect t$ a
(possibly empty) sequence of variables, names, messages, and terms
respectively.
We also sometimes write them $(n_1, n_2, \ldots)$ or simply $n$ (when
the sequence is reduced to one element).
Substitutions are denoted by $\sigma$,
the domain of a substitution is written $\dom(\sigma)$, and
the application of a substitution $\sigma$ to a term $u$ is written
$u\sigma$.
The \emph{positions} of a term are defined as usual.

\begin{example}
\label{ex:signature}
Consider the following signature:
$$\Sigma = \{\senc,\;\sdec,\;\langle \, \rangle, \;
\proj_1, \; \proj_2, \; \oplus, \; 0, \; \eq, \;\ok\}.$$
The symbols $\senc$ and $\sdec$ of
arity 2 represent symmetric encryption and decryption. Pairing is
modeled using $\langle \; \rangle$ of arity 2, and projection
functions $\proj_1$ and $\proj_2$, both of arity 1.
The function symbol $\oplus$ of arity 2 and the constant $0$ are  used to model  the
exclusive or operator. Finally, we consider the symbol $\eq$ of arity~2 to model
equality test, as well as the constant symbol $\ok$. This signature is
split into two parts:
$
\Sigma_c = \{\senc, \langle \; \rangle, \oplus, 0, \ok\}$, and 
$\Sigma_d = \{\sdec, \proj_1, \proj_2, \eq\}$.
\end{example}

As in the process calculus presented in~\cite{BlanchetAbadiFournetJLAP08}, 
constructor terms are subject to an equational theory;
this has proved very useful for modelling algebraic
properties of cryptographic primitives
(see \emph{e.g.}~\cite{CDL05-survey} for a survey).
Formally, we consider a congruence~$=_\E$ on $\T(\Sigma_c,\N \cup \X)$,
generated from a set of equations~$\E$ over~$\T(\Sigma_c,\X)$.
Thus, this congruence relation is closed  under substitution
and renaming. 
We assume that it is not degenerate, \ie
there exist $u$, $v$ such that $u \neq_\E v$.

\begin{example}
\label{ex:xor}
To reflect the algebraic properties of the exclusive or operator, we may
consider the equational theory generated by the following equations:
$$
x \oplus 0 \;=\; x  
\qquad x \oplus x \;=\; 0 
\qquad x \oplus y \;=\; y \oplus x 
\qquad (x \oplus y) \oplus z \;=\; x \oplus  (y \oplus z) 
$$
In such a case, we have that $\senc(a \oplus (b \oplus a), k) =_\E \senc(b,k)$.
\end{example}

We also give a meaning to destructor symbols through the notion of 
\emph{computation relation}.
As explained below, a computation relation may be derived from a
rewriting system but we prefer to not commit to a specific
construction, and therefore we introduce a generic notion of
computation relation.
Instead, we assume an arbitrary relation subject to a number of requirements,
ensuring that the computation relation behaves naturally with respect to
names, constructors, and the equational theory.
%For example, when two messages $u,v$ are equal modulo $\theo$,
%replacing $u$ by $v$ in a term $t$ should not impact the computation of $t$
%since $u$ and $v$ represent the same piece of data.

\begin{definition}
\label{def:computation-rel}
A \emph{computation relation} is a relation over $\T(\Sigma,\N) \times
\T(\Sigma_c,\N)$, denoted $\redc$, that satisfies the following
requirements:
%   \dam{4 plus utile depuis abandon ``transparent''?}
% S.D.: Ok mais on avait dit que 4 etait naturel et que l'on gardait.
\begin{enumerate}
\item if $n \in \N$, then $n\redc n$; \label{def-cr-name}
\item if $\ffun\in\Sigma_c$ is a symbol of arity $k$, and $t_1 \redc
  u_1,\ldots, t_k \redc u_k$, then 
  ${\ffun(t_1,\ldots,t_k) \redc \ffun(u_1,\ldots, u_k)}$; \label{def-cr-congruence}
\item if $t \redc u$  then $t\rho \redc u\rho$ for any bijective
  renaming $\rho$; \label{def:cr-renaming}
\item  if 
$t'$ is a context built from $\Sigma$ and $\N$,
  $t \redc u$, and $t'[u]\redc v$ 
  then $t'[t]\redc v$; \label{def-cr-context-1}
\item if $t'$ is a context built from $\Sigma$ and $\N$,
  and $t_1$, $t_2$ are constructor terms
  such that $t_1 =_\E t_2$ and $t'[t_1] \redc u_1$ for some~$u_1$, then
  $t'[t_2] \redc u_2$ for some $u_2$
    such that $u_1 =_\E u_2$; \label{def-cr-context-2}
\item if $t \redc u_1$ then we have that $t \redc u_2$ if, and only
  if, $u_1 =_\E u_2$. \label{def-cr-context-2}
\end{enumerate}
\end{definition}

The last requirement expresses that the relation $\redc$
associates, to any ground term $t$,
at most one message up to the equational theory $\E$.
When no such message exists, 
we say that the \emph{computation fails}; this is noted  $t\redcb$.
  %As a slight abuse of notation,
  We may sometimes use directly~$t\redv$ as a message,
  when we know that the computation succeeds and the choice
  of representative is irrelevant.

\smallskip{}

A possible way to derive a computation relation is to consider an ordered set of
rules of the form:
${\gfun(u_1,\ldots, u_n) \to u}$ where $\gfun$ is a
destructor, and $u, u_1, \ldots, u_n \in \T(\Sigma_c,\X)$.
A ground term~$t$ can be rewritten into $t'$ if there is  a
position~$p$ in~$t$ 
and a rule $\gfun(u_1,\ldots, u_n) \to u$
 such that $t|_p = \gfun(v_1,\ldots,v_n)$ and $v_1 =_\E u_1\theta, \ldots, v_n=_\E u_n\theta$ for
some substitution $\theta$, and $t' = t[u\theta]_p$ (\emph{i.e.} $t$
in which the subterm at position~$p$ has been replaced by~$u\theta$).
Moreover, we assume that $u_1\theta,
\ldots, u_n\theta$ as well as $u\theta$ are messages. In case
there is more that one rule that can be
applied at a given position $p$, we consider the one occurring first in the ordered set.
We denote~$\to^*$ the reflexive and transitive closure of~$\to$, and 
$\redc$ the relation induced by $\to$, \ie
$t\redc u$ when $t\to^* u'$ and $u' =_\E u$.

%For this paper, we only rely on the notion of computation relation as
%given in Definition~\ref{def:computation-rel}. 
%We do not need to
% S.D.: c'est un peu exagere puisque tous nos exemples sont donnes a
% travers des rewriting systems, on voulait parler du developpement theorique.
%know whether this computation relation has been obtained through a 
%rewriting system.
%Therefore, 
Proving that an ordered set rewriting rules as
  defined above
induces a computation relation is beyond the scope of
this paper but the interested reader will find such a proof
in~\cite{PhD-Hirshi}.

\begin{example}
\label{ex:rewriting}
The properties of 
symbols in~$\Sigma_d$ (Example~\ref{ex:signature})
are reflected through the
following rules:
$$
\sdec(\senc(x,y),y) \to x
\qquad
\eq(x,x) \to \ok
\qquad
\proj_i(\langle x_1,x_2\rangle) \to x_i \;\;\mbox{ for $i \in
  \{1,2\}$}
$$
This rewriting system
induces a computation relation.
For instance, we have that:
$$
\sdec(\senc(c, a \oplus b), b \oplus a) \redc c,
\qquad
\sdec(\senc(c, a \oplus b), b) \redcb, \mbox{and}
\qquad
\sdec(a,b) \oplus\sdec(a,b) \redcb.
$$
\end{example}

\begin{example}
\label{ex:neq}
Ordered rewriting rules are expressive enough to
define a destructor symbol $\fneq$ such that $\fneq(u,v)\redc
\mathsf{yes}$ if, and only if, $u$ and $v$ can be reduced to messages that are
not equal modulo~$\E$.
It suffices to consider 
$\fneq(x,x) \to \mathsf{no}$ and $\fneq(x,y)\to\mathsf{yes}$ (in this order)
  with
$\mathsf{yes}, \mathsf{no}\in\Sigma_c$.
\end{example}

For modelling purposes, we split the signature $\Sigma$ into two
parts, namely $\Sigma_\pub$ and $\Sigma_\priv$. 
An attacker builds his own messages by applying public function symbols to
terms he already knows and that are available through variables
in~$\W$. Formally, a computation done by the attacker is a
\emph{recipe}, \ie a term in $\T(\Sigma_\pub,\W)$.
Recipes will be denoted by $R$, $M$, $N$.
Note that, although we do not give the attacker the ability to generate
fresh names to use in recipes, we obtain essentially the same capability
by assuming an infinite supply of public constants in
$\Sigma_c \cap \Sigma_\pub$.

%%
%% Process algebra
%%

\subsection{Process algebra}

We now define the syntax and semantics of the process algebra 
we use to model security protocols. 
We consider a calculus close to the one which is
used in the \proverif~tool~\cite{BlanchetAbadiFournetJLAP08}.
An important difference is that we only consider public channels.
Our calculus also features, in addition to the usual replication (where an 
unbounded number of copies of a process are ran concurrently),
a simple form of sequential composition and the associated \emph{repetition}
operation (where an unbounded number of copies of a process
are ran sequentially, one after the other).

\medskip{}

We consider a set $\Ch$ of channel names that are assumed to be
public. Protocols are modeled through processes using the following
grammar:
% \lu{\footnotesize
% Concernant sequential replication: dans la litérature, recursive process veut dire $\mathrm{rec}X.P$ avec $X\in P$.
% Et on a une règle d'unfold. Ce n'est pas le comportement que l'on veut ici, sauf si on remplace TOUS les $0$ de $P$ par des $X$
% pour que ``quand on a fini d'exécuter une copie de $P$ par n'importe quel moyen, une nouvelle copie de $P$ apparait.''. Quand je
% googlise ``sequential replication'', je trouve un article TGC'2012 de Lili Xu qui introduit une notion très similaire.
% \normalsize}
 %\begin{figure}[t]
$$
\begin{array}{rclclrlclrlcl}
    P,Q &:=&  0 & \;\;\;& \mbox{null} & \hspace{1cm} \mid  & (P \mid Q)&&\mbox{parallel}&\hspace{1cm}
                                               \mid& \rep  P &&
                                                                \mbox{replication}
  \\[0.7mm]
 &\mid & \In(c, x).P && \mbox{input} &\mid& \new \vect n. P &&
                                                               \mbox{restriction} & \mid& \rec  P &&                     \mbox{repetition} \\[0.7mm]
&\mid&\Out(c, u).P &&\mbox{output}  &\mid& \Let \; \vect x = \vect
                                               t \;\In \; P \; \Else
                                                \; Q&&
                                                       \mbox{evaluation} &\mid& P;Q && \mbox{sequence}\\[0.5mm]
\end{array}
$$
\noindent where $c \in \Ch$, $x \in \X$, $n \in \N$,
$u \in \T(\Sigma_c, \N \cup \X)$, % is a constructor term,
and $\vect x$ and $\vect t$ are two sequences of the same
length, respectively over variables ($\X$) and terms ($\T(\Sigma, \N\cup\X)$).
   
\smallskip{}
We write~$\fv(P)$ for %and $\bv(P)$,
the set of \emph{free variables} of~$P$, \ie the set
of variables that are not bound by an input or a $\Let$
construct.
A process $P$ is ground if $\fv(P) = \emptyset$.
Similarly, we write $\fn(P)$ for the set of \emph{free names} of~$P$,
\ie the set of names that are not bound by a $\new$ construct.

\smallskip{}

Most constructs are standard in process calculi.
The process $0$ does
nothing and we sometimes omit it. The process $\In(c,x).P$ expects a
message $m$ on channel $c$ and then behaves like $P\{x \mapsto m\}$,
\ie~$P$ in which the (free) occurrences of~$x$ have
been replaced by~$m$. The
process $\Out(c,u).P$ emits~$u$ on channel~$c$, and then behaves like~$P$.
The process $P \mid Q$ runs $P$ and $Q$ in parallel. The process
$\new \vect n.P$ generates new names, binds it to~$\vect n$,
and continues as~$P$. 

The special construct  $\Let \; \vect x = \vect t \;\In \;
P \; \Else \; Q$ combines several standard constructions,
allowing one to write computations and conditionals compactly.
Such a  process tries to evaluate the sequence of terms~${\vect t}$ and
in case of success, \ie when $\vect t \redc \vect u$ for some messages
$\vect u$, the process $P$ in which~$\vect x$ are replaced by~$\vect u$
is executed;
otherwise the process~$Q$ is executed. The goal of this construct is to avoid nested $\Let$
instructions to be able to define our class of protocols in a simple
way later on. Note also that the $\Let$ instruction together with the
$\eq$ theory as defined in Example~\ref{ex:rewriting} can encode the usual
conditional construction. Indeed, $\Let \;x = \eq(u,v) \;\In
\; P \; \Else \; Q$ will execute~$P$ only if the computation succeeds
on $\eq(u,v)$, that is only if $u\redc u'$, $v\redc v'$, and $u' \theo v'$ for some
messages $u'$ and $v'$. For brevity, we sometimes omit $\Else\;0$.

The
process $\rep P$ executes $P$ an arbitrary number of times (in
parallel).
The last two constructs correspond to sequential compositions.
The process $(P;Q)$ behaves like $P$ at first,
and after the complete execution of $P$ it behaves like $Q$.
The process $\rec P$ executes $P$ an arbitrary number of times in
sequence, intuitively corresponding to $(P;P;P;\ldots)$.
Such constructions are known to be problematic in process calculi.
Our goal here is however quite modest: as will be visible in our operational
semantics, our sequential composition is only meaningful for restricted 
processes. It could in fact be defined using recursion, but there is no point 
here to consider general recursion: our study is going to restrict to
a simple class of protocols that would immediately exclude it.

\newcommand{\open}{\mathsf{open}}
\newcommand{\close}{\mathsf{close}}

\begin{example}
\label{ex:process}
We consider the RFID protocol due to Feldhofer \emph{et al.} as 
described in~\cite{feldhofer2004strong} 
and which can be presented using Alice \& Bob notation  as follows:\\[1mm]
\null\hfill
$\begin{array}{rrl}
1.& I \to R: & n_I\\
2.& R \to I: & \{ n_I, n_R\}_k\\
3.& I \to R: & \{n_R, n_I\}_k
\end{array}$\hfill\null

\smallskip{}

\noindent The protocol is between an initiator $I$ (the reader) and a responder
$R$ (the tag) that share a symmetric key~$k$.
We consider the term algebra introduced in Example~\ref{ex:rewriting}.
The protocol is modelled  by
the parallel composition  of~$P_I$ and~$P_R$, corresponding
respectively  to the roles~$I$ and~$R$.\\[1mm]
\null\hfill
$P_\Feldhofer \;\; := \; \new k.\; (  \new n_I.  P_I \;\mid \;
\new 
n_R.  P_R)$\hfill\null

\smallskip{}

\noindent where $P_I$ and $P_R$ are defined as follows, with $u = \sdec(x_1,k)$:
\[
\begin{array}{rcl}
P_I  &:=&   \Out(c_I,n_I).\\
&&\In(c_I,x_1). \\
&&\Let \; x_2, x_3 = \eq(n_I, \proj_1(u)), \proj_2(u)\; \In\\
&&\Out(c_I,\senc(\langle x_3, n_I\rangle,k))
\end{array}
\;\;\;\;\;\;
\begin{array}{rcl}
P_R &:=& \In(c_R,y_1). \\
&&\Out(c_R,\senc(\langle y_1, n_R\rangle, k)).\\
&&\In(c_R,y_2).\\
&&
\Let \; y_3 = \eq(y_2, \senc(\langle n_R, y_1 \rangle,k)) \; \In \; 0
\end{array}
\]
We may note that there are potentially several ways to implement the
    last reader's test.
    For instance, we may decide to replace the last line of the process $P_R$ by
    $\Let \; y_3 =\linebreak[4]
    {\eq(\langle n_R, y_1\rangle, \sdec(y_2,k)) \; \In \; 0}$. This last
    check can also be simply removed. Alternatively, it could be
    followed by an observable action to make the outcome of the test manifest,
    e.g.\ the output of a public constant $\open$ in case of success and
    $\close$ in case of failure.
      This would be a reasonable model for many use cases, e.g.\ in
      access control scenarios
      a door may either open or remain close after the execution of the
      protocol.
\end{example}

%%
%% SEMANTICS
%%

%\subsubsection{Semantics}

The operational semantics of processes is given by a labelled transition
system over \emph{configurations} (denoted by $K$) which are
pairs~$(\p; \phi)$ where:
\begin{itemize}
\item $\p $ is a multiset of ground processes where null processes are 
  implicitly removed;
\item $\phi = \{w_1 \refer u_1, \ldots, w_n \refer u_n\}$ is a
  \emph{frame}, \ie a substitution where $w_1, \ldots, w_n$ are
  variables in~$\W$, and $u_1,\ldots, u_n$ are messages.
\end{itemize}

We often write $P \cup \p$  instead of $\{P\} \cup \p$. The terms in~$\phi$ 
represent the messages that are known by the attacker.
Given a configuration~$K$, $\phi(K)$ denotes its second 
component. 
Sometimes, we consider processes as configurations: in such cases, the corresponding
frame is the empty set~$\emptyset$.

\begin{figure*}[t]
\[
\begin{array}{lcl}
%IN
\textsc{In} && (\In(c,x).P \cup \p; \phi)\; \lrstep{\In(c,R)} \; (P \{x \mapsto u\}
 \cup \p; \phi) \\
&&\hfill \mbox{
where  $R$ is a recipe  such that $R\phi\redc u$ for
some message $u$}\\[1mm]
 % OUT
\textsc{Out}&& (\Out(c,u).P \cup \p; \phi) \; \lrstep{\Out(c, w)} \; (P
 \cup \p; \phi \cup \{w \refer u\} ) 
 \hfill \text{with $w$ a fresh variable in $\W$}\\[1mm]
% NEW
\textsc{New}&& (\new \vect{n}.P \cup \p; \phi) \;
\lrstep{\tau} \; (P\cup \p;
 \phi) \;\;\;\;\;
  \hfill \mbox{where $\vect n \cap \fn(\p,\phi) = \emptyset$}\\[1mm]
% PAR
\textsc{Par}&& (\{P_1 \mid P_2\} \cup \p; \phi) \; \lrstep{\tau} \;
 (\{P_1, P_2\} \cup \p; \phi) \\[1mm]
% LET-THEN
\textsc{Then}\;\;&&(\Let \; \vect x = \vect t \; \In \; P \; \Else \; Q\cup \p; \phi) \; \lrstep{\taut} \;
(P\{\vect x \mapsto \vect u\} \cup
\p; \phi) \;\;\;\;\;\;
\hfill \mbox{when $\vect t\redc \vect u$ for some $\vect u$}\\[1mm]
% LET-ELSE
\textsc{Else}&&(\Let \; \vect x = \vect t \; \In \; P \; \Else \; Q\cup \p; \phi) \; \lrstep{\taue} \;
(Q \cup
\p; \phi) 
\hfill \mbox{when $t_i\redcb$ for some $t_i \in \vect t$}\\[1mm]
% REPLICATION
  \textsc{Rep}{-}{\rep} && (\rep P \cup \p; \phi) \; \lrstep{\tau} \;
 (P \,\cup\, \rep P \cup \p; \phi) \\[1mm]
% REPETITION
  \textsc{Rep}{-}{\rec} && (\rec P \cup \p; \phi) \; \lrstep{\tau} \;
 (\{P;\rec P\} \cup \p; \phi) \\[1mm]
% SEQUENCE
% \textsc{Seq$_0$}&& (\{0;Q\} \cup \p; \phi) \; \lrstep{\tau} \;
% (Q \cup \p; \phi)   \\[1mm]
\textsc{Seq}&&
  (P \cup \p; \phi) \; \lrstep{\alpha} \;
  (P'\cup \p; \phi') 
  \hfill
  \mbox{if } P \rightsquigarrow Q \text{ and }
  (Q\cup\p;\phi) \lrstep{\alpha}
  (P'\cup\p;\phi')
\end{array}
\]
\caption{Semantics for processes}
\label{fig:semantics}
\end{figure*}
\begin{figure}[t]
  $\begin{array}{rcll}
 %   \multicolumn{4}{c}{
  %  (\Out(c,u).P);Q \rightsquigarrow \Out(c,u).(P;Q)
   % \quad\quad
   % 0;Q \rightsquigarrow Q
   % }
    0;Q &\rightsquigarrow& Q&
   \\[2mm]
   (\Out(c,u).P);Q &\rightsquigarrow& \Out(c,u).(P;Q)&\\[2mm]
    (\In(c,x).P);Q &\rightsquigarrow& \In(c,x).(P;Q)
    & \text{ when } x \not\in \fv(Q)
    \\[2mm]
    (\new\vect{n}.P);Q &\rightsquigarrow& \new\vect{n}.(P;Q)
    & \text{ when } \vect{n}\cap\fn(Q)=\emptyset
    \\[2mm]
    (\Let\; \vect{x} = \vect{u} \;\In\; P' \;\Else\; P'');Q
      &\rightsquigarrow&
      \Let\; \vect{x} = \vect{u} \;\In\; (P';Q) \;\Else\; (P'';Q)
      & \text{ when } \vect{x}\cap\fv(Q)=\emptyset
  \end{array} $
  \caption{Sequence simplification rules}
  \label{fig:seq}
\end{figure}

\smallskip{}

The operational semantics of a process 
is given by the relation
$\lrstep{\alpha}$
defined as the least relation over configurations satisfying the rules in Figure~\ref{fig:semantics}.
The rules are mostly standard and correspond to the intuitive meaning
given previously.
 Rule \textsc{In} 
 allows the attacker to send on channel~$c$ a message as soon as it is
 the result of a computation done by applying public function symbols
 on messages that are in his current knowledge.
 Rule \textsc{Out} corresponds to the output of a term: the
 corresponding term is added to the frame
of the current configuration, which means that the attacker gains
access to it.
Rule \textsc{New} corresponds to the generation of a fresh
name. As is standard, the bound names $\vect{n}$ can be renamed to achieve
freshness so that the rule can always fire.
The \textsc{Par} rule simply splits parallel compositions.
 The \textsc{Then} and \textsc{Else} rules correspond to the evaluation of
  a sequence of terms $\vect t = t_1,\ldots,t_n$; if this succeeds,
  \ie if there exist messages $u_1, \ldots u_n$ such that $t_1
  \redc u_1, \ldots t_n\redc u_n$ then variables $\vect x$ are bound
  to those messages, and $P$ is executed; otherwise the
  process will continue with $Q$.
%   The two next rules allow one
%   to perform some equality or disequality test modulo the equational
%   theory $\theo$.
% For two messages $u,v\in\mess$, we have that $\models u=v$ if $u\theo v$ and $\models u\neq v$ if $u\not=_\E v$.
Rules \textsc{Rep}${-}{\rep}$ and \textsc{Rep}${-}{\rec}$ unfold 
replication and repetition operators. The latter gives rise to a sequential 
composition, whose execution will have to rely, via the \textsc{Seq} rule,
on the simplification rules of Figure~\ref{fig:seq}.
These rules only support a limited set of operators, hence a sequence
$P;Q$ is only executable for a restricted class of processes $P$,
notably excluding parallel compositions. This is not an issue for our
simple needs; our purpose here is not to define a general (and notoriously 
problematic) notion of sequence.

We note that our semantics enjoys some expected properties.
Reduction is stable by bijective renaming,
thanks to Definition~\ref{def:computation-rel}, item 3:
if $K_1 \lrstep{\alpha} K_2$ then $K_1\rho \lrstep{\alpha} K_2\rho$
where $\rho$ is a bijection over $\N$,
applied here to processes and frames.
It is also compatible with our equational theory,
thanks to Definition~\ref{def:computation-rel}, items 5 and 6:
if $K_1 \lrstep{\alpha} K_2$ then $K'_1\lrstep{\alpha} K'_2$
for any $K'_1 =_\E K_1$ and $K'_2 =_\E K_2$.
By Definition~\ref{def:computation-rel}, item 6, we also have that
$$(\p;\phi) \lrstep{\In(c,R)} (\p';\phi) \mbox{ and }
R\phi\redv =_\E R'\phi\redv \mbox{ yield }
(\p;\phi) \lrstep{\In(c,R')} (\p';\phi) \mbox{\;\;(modulo $\E$).}$$

As usual,
the relation $\lrstep{\alpha_1 \ldots \alpha_n}$ between
configurations (where~$\alpha_1 \ldots \alpha_n$ is a \emph{trace},
\ie a sequence of actions) is defined
as the (labelled) reflexive and transitive closure of~$\lrstep{\alpha}$. 

\begin{definition}
  Input and output actions are called \emph{observable}, while
  all other are unobservable. Given a trace $\tr$
  we define $\obs(\tr)$ to be the sub-sequence of observable actions of $\tr$.
\end{definition}

We generally refer to $\tau$, $\taut$ and $\taue$
as {\em unobservable actions}. %\marginpar{redondant}
It will become clear later on why we make a distinction
when a process evolves using $\textsc{Then}$ or $\textsc{Else}$.

\begin{example}
\label{ex:execution}
Continuing Example~\ref{ex:process}. We have that
$P_\Feldhofer \lrstep{\tr} (\emptyset; \phi_0)$ where:
\begin{itemize}
\item $\tr =
  \tau.\tau.\tau.\tau.\Out(c_I,w_1).\In(c_R,w_1).\Out(c_R,w_2).\In(c_I,w_2).\taut.\Out(c_I,w_3). \In(c_R,w_3).\taut$;
\item $\phi_0 = 
\{w_1 \mapsto {n'_I},  \; w_2 \mapsto \senc(\langle {n'_I}, 
{n'_R}\rangle,k'), \;  w_3 \mapsto \senc(\langle n'_R, n'_I
\rangle,k')\}$.
\end{itemize}
The names $k'$, $n'_I$ and~$n'_R$ are fresh names.
Actually, this execution corresponds to a normal execution of one session of the
protocol.
\end{example}

%%
%% PROTOCOLS
%%
\subsection{A generic class of two-party protocols}
\label{subsec:proto}

We aim to propose sufficient conditions to ensure unlinkability and
anonymity for a generic class of two-party protocols.
In this section, we define formally the class of protocols we are
interested in.

\paragraph{Roles.}
We consider two-party protocols that are therefore
made of two roles called the initiator and responder role
respectively. We assume a set $\mathcal{L}$ of labels
that will be used to name output actions in these roles,
allowing us to identify
outputs that are performed by a same syntactic output action.
These labels have no effect on the semantics.

\begin{definition}%[role]
\label{def:priv:role:grammar}
An \emph{initiator role} is a process that is obtained using the following
grammar:\\[1mm]
\null\hfill
$P_I :=  0\ \mid\  \ell: \Out(c,u). P_R$\hfill\null

\smallskip{}

\noindent where $c\in\Ch$, $u  \in \T(\Sigma_c,\N \cup \X)$, $\ell \in \mathcal{L}$,
and $P_R$ is obtained from the grammar of \emph{responder roles}:\\[2mm]
%\null\hfill
%$\begin{array}{c}
%P_R \;:=\;  0 \quad
 % \mid \; \In(c,y).    \Let\; \vect x =\vect t\; \In\; 
  %                           P_I\;
   %                       \Else\; 0 \quad  \mid \;  \In(c,y).    \Let\; \vect x =\vect t\; \In\; 
    %                         P_I\;
     %                     \Else\; \ell: \Out(c',u')
%\end{array}$\hfill\null
\null\hfill
$
P_R \;:=\;  0 \quad
  \mid \; \In(c,y).    \Let\; \vect x =\vect t\; \In\; 
                             P_I\;
                          \Else\; P_{\fail}
\qquad \mbox{ with } P_{\fail} = 0 \;\mid \; \ell: \Out(c',u')$\hfill\null

\smallskip{}

\noindent where $c,c'\in\Ch$, $y \in \X$,  $\vect x$
  (resp.\ $\vect t$) is a sequence of variables in~$\X$ 
  (resp.\ terms in
  $\T(\Sigma, \N \cup \X)$),   $u' \in
\T({\Sigma_c},\N \cup \X)$, and $\ell \in \mathcal{L}$.

Moreover, an initiator (resp. responder) role is assumed to be
  ground, i.e., contain no free variable, though it may
  contain free names.
\end{definition}

Intuitively, a role describes the actions performed by an agent. 
A responder role consists of waiting for an input and,
depending on the outcome of a number of tests,
the process will continue by sending a message
and possibly waiting for another input, or
stop possibly outputting an error message. 
An initiator behaves similarly but begins with an output.
The grammar forces to add a conditional after each input.
This is not a real restriction as it is always possible to add trivial
conditionals with empty $\vect{x}$, and~$\vect{t}$.

\begin{example}
Continuing our running example, $P_I$ (resp.~$P_R$) as defined in Example~\ref{ex:process}
 is an initiator (resp.\ responder) role, up to the addition of a trivial
 conditional in role~$P_R$ and distinct labels
 $\ell_1$, $\ell_2$, and $\ell_3$ to decorate 
 output actions. 
\end{example}

Then, a protocol 
notably consists
of an initiator role and a responder role that can interact together
producing an {\em honest trace}. 
Intuitively, an honest trace is a trace
in which the attacker does not really
interfere,
and that allows the execution to
progress without going into an $\Else$ branch,
which would intuitively correspond to a way to abort the protocol.

\begin{definition}
\label{def:honest}
  A trace $\tr$ (\ie a sequence of actions)  is {\em honest}
  for  a frame $\phi$ if $\taue\notin \tr$ and
  $\obs(\tr)$ is of the form
$\Out(\_,w_0).\In(\_,R_0).\Out(\_,w_1).\In(\_,R_1).\ldots$
for arbitrary channel names, and such that
$R_i\phi\redv w_i\phi$
for any action $\In(\_,R_i)$ occurring in $\tr$.
\end{definition}

\paragraph{Identities and sessions.}
In addition to the pair of initiator and responder roles,
more information is needed in order to meaningfully define a protocol.
Among the names that occur in these two roles, we need to distinguish those
that correspond to identity-specific,
long-term data (\eg $k$ from Example~\ref{ex:process}),
called {\em identity parameters}
and denoted $\vect{k}$ below,
and those which shall be freshly generated at each
session (\eg $n_I,n_R$ from Example~\ref{ex:process}), called
{\em session parameters} and denoted $\nI$ and $\nR$ below.
We will require that any free name of roles 
must be either a session or an identity parameter.
  When necessary, we model long-term data that is not identity-specific (\ie uniform
  for all agents) as private constants (\ie terms in
  $\Sigma_c \cap \Sigma_\priv$).% (\ie initially unknown to the attacker).
Depending on the protocol to be modelled,
  we shall see that either both the initiator and the responder
  or only one of those roles have identity parameters.
  The former case arises for protocols that involves different identities
  for each party while the latter concerns protocols whose
  only one party can be instantiated by different agents
  (see Examples~\ref{ex:DAA},~\ref{ex:DAA-s} below for a more detailed discussion).

We also need to know whether 
sessions (with the same identity parameters) can be executed
\emph{concurrently} or only \emph{sequentially}. 
For instance, let us assume that the Feldhofer protocol is used in an access
control scenario where all tags that are distributed to users have
pairwise distinct identities. Assuming that tags cannot be cloned,
it is probably more realistic to consider that a tag can be 
involved in at most one session at a particular time, \ie a tag may
run different sessions but only in sequence. Such a situation will
also occur in the e-passport application where a 
%Another concrete example where sessions are executed sequentially is
%the case of the e-passport application where a 
same passport cannot
be involved in two different sessions of the BAC protocol (resp. PACE protocol)
concurrently.
This is the purpose
of the components~$\dag_I$ and~$\dag_R$ in the following definition.
When one role has no identity parameter, we also consider both cases:
  whether the only identity instantiating this role may have concurrent sessions
  or only sequential sessions.
Moreover, we require that the process~$P_\Pi$ which models a single session of the protocol can produce an honest trace.

\begin{definition}
  \label{def:proto}
A protocol $\Pi$ is a tuple $(\vect{k}, \nI, \nR, \dag_I,
\dag_R, \ini, \res)$
where $\vect k$,
$\nI$, $\nR$ are three disjoint sets of names, $\ini$ 
(resp.~$\res$) is an initiator (resp.\ responder) role such that
$\fn(\ini) \subseteq \vect k \sqcup \nI$, $\fn(\res) \subseteq
\vect k \sqcup \nR$, and $\dag_I, \dag_R \in \{\rep, \rec\}$.
Labels of $\ini$ and $\res$ must be pairwise distinct.
Names $\vect k$ (resp.\ $\nI\sqcup\nR$) 
are called {\em identity parameters} (resp. {\em session parameters}).

Given a protocol $\Pi$,
we define $P_\Pi := \new \vect k. (\new \nI. \ini \mid \new \nR. \res)$
and we assume that ${P_\Pi \lrstep{\tr_h} (\emptyset; \phi_h)}$
for some frame~$\phi_h$ 
and some trace $\tr_h$ that is
honest for $\phi_h$.
\end{definition}

Given a protocol $\Pi$, we also associate another process
  $\pM$ that represents the situation where the protocol
  can be executed by an arbitrary number of identities, with the possibility
  of executing an arbitrary number of sessions for a given identity.
  The formal definition differs slightly
  depending on
  whether identity parameters occur in both roles or only in the role $\ini$
  (resp.\ $\res$).

\begin{definition}
\label{def:systemM}
Given a protocol $\Pi =(\k, \nI, \nR, \dag_I, \dag_R, \ini, \res)$, the process $\mathcal{M}_\Pi$ is defined
as follows:
\begin{itemize}
\item If $\vect k \cap \fn(\ini) \neq \emptyset$
and $\vect k \cap \fn(\res) \neq \emptyset$, then 
$\pM = \; \rep \; \new \vect k. (\dag_I\; \new \nI. \ini \mid
\dag_R\; \new \nR.\res)$;
\item  If  $\vect k \cap \fn(\ini) = \emptyset$
and $ \vect k \cap \fn(\res) \neq \emptyset$, then
$\pM = \dag_I\; \new \nI. \ini \;\mid \; \rep \; \new \vect k. \; \dag_R\; \new \nR.\res$.
\end{itemize}
\end{definition}

For the sake of simplicity, in case identity parameters only occur in one role,
we assume that this role is the responder role. The omitted case where
identity parameters occur only in the initiator role is very much similar.
In fact, swapping the initiator and responder roles can also be formally achieved
by adding an exchange of a fresh nonce at the beginning of
the protocol under consideration.
Note that the case where both $\vect k \cap \fn(\ini)$ and $\vect k
\cap \fn(\res)$ are empty means that no identity parameters are involved and
therefore there is no issue regarding privacy. As expected, in such a
  situation, our definitions of unlinkability and anonymity (see
  Definition~\ref{def:anonymity} and Definition~\ref{def:un}) will be trivially satisfied. 
%\lum{doit-on dire pourquoi les deux cas collapse pas ?}
% S.D.: Cela me semble clair 

\begin{example} 
\label{ex:protocol-FH}
Let $\Pi_\Feldhofer = (k, n_I, n_R, \rep, \rep,  P_I, P_R)$ with~$P_I$ and~$P_R$  as defined in
Example~\ref{ex:process} (up to the addition of a trivial conditional).
%We have
%seen that $P_I$ is an initiator role whereas $P_R$ is a
%responder role. 
Let $P_\Feldhofer \; = \; \new k. (\new n_I. P_I \mid \new n_R. P_R)$, 
$\tr_h = \tr$ (up to the addition of an action $\taut$), and $\phi_h = \phi_0$
as defined in Example~\ref{ex:execution}. They satisfy the requirements stated in
Definition~\ref{def:proto}, and therefore $\Pi_\Feldhofer$ is a protocol
according to our definition. For this protocol, the identity parameter~$k$ occurs
both in the role $P_I$ and $P_R$, and therefore we have that
$\mathcal{M}_{\Pi_\Feldhofer} = \; !\,\new k.(! \;\new n_I.P_I
\;\mid\; ! \; \new n_R.P_R)$.
\end{example}

%\paragraph*{\stef{Description de MiniPACE}}

\newcommand{\kdf}{\mathsf{kdf}}
\newcommand{\dihe}{\mathsf{dh}}
\newcommand{\gconst}{\mathsf{g}}
\newcommand{\Toy}{\mathsf{Toy}}

\begin{example}
\label{ex:toyseq}
In order to illustrate our method and the use of the repetition
  operator, we introduce a toy protocol which, as we shall see later,
  satisfies unlinkability only when sessions of the initiator role are
  executed sequentially.
Using Alice \& Bob notation, this
protocol can be described as follows:
$$
  \begin{array}{rrl}
    1.& T\to R : & n_T\\
    2. &R  \to  T: & n_R\\
    3. &T  \to  R: & \mac(\langle n_R,n_T\rangle,k)\\
   4. &R  \to  T: & \mac(\langle n_T, n_R\rangle,k)
  \end{array}
$$

The protocol is between a tag~$T$ (the initiator) and a reader~$R$
(the responder) which share a symmetric key~$k$. To avoid an obvious
reflection attack (where $n_T=n_R$), we assume that the tag systematically checks that
the first message it receives is not the one he sent initially.

To formalise such a protocol, we consider 
$$
\Sigma_c = \{\mac, \langle \;\rangle, \ok, \mathsf{yes}, \mathsf{no}\},
\mbox{and } \Sigma_d = \{\proj_1, \proj_2, \mathsf{eq}, \mathsf{neq}\}.
$$
All symbols of the signature are public:
$\ok$, $\mathsf{yes}$, and $\mathsf{no}$ have arity $0$;
$\proj_1$ and $\proj_2$ have arity~$1$;
other function symbols have arity~$2$. 
The destructors $\proj_1,\proj_2,\mathsf{eq}$ and $\mathsf{neq}$
are defined as in \Cref{ex:rewriting,ex:neq}.
The symbol $\mac$ will be used to model
message authentication code.

The processes modelling the initiator and the responder roles are as
follows:
$$ 
\begin{array}{ll}
    P'_T  := & \Out(c_T, n_T).\\
    &\In(c_T, x_1).\; \\
    & \Let \; x_\mathsf{test} = \mathsf{eq}(\mathsf{yes},\mathsf{neq}(x_1,n_T)) \; \In\\
    & \Out(c_T, \mac(\langle x_1, n_T\rangle, k)). \In(c_T,x_2).\\
    & \Let  \; x'_\mathsf{test} = \eq(x_2,\mac(\langle n_T, x_1\rangle,k)) \; \In \; 0\\
  \end{array}
\;\;
 \begin{array}{ll}
    P'_R  := & \In(c_R, y_1). \\
    &\Out(c_R,  n_R).\\
   & \In(c_R, y_2). \\
& \Let  \; y_\mathsf{test} = \eq(y_2,\mac(\langle n_R, y_1\rangle, k))
                        \; \In\\
& \Out(c_R, \mac(\langle y_1, n_R\rangle,k)). \; 0
  \end{array}
$$

We may note that different choices can be made. For instance, we may
decide to replace the process~$0$ occurring in~$P'_T$ with
$\Out(c_T,\ok)$ to make the outcome of the test manifest.
  The tuple $\Pi^{\rec}_\Toy = (k, (n_T), (n_R), \rec,
\rep, P'_T, P'_R)$ is a protocol according to
Definition~\ref{def:proto}. For this protocol, we have again that
$k$ occurs both in the roles $P'_T$ and $P'_R$, and therefore:
$$\mathcal{M}_{\Pi^{\rec}_\Toy} =\; !\,\new k.\big(
\rec \new n_T. P'_T
\;\mid\;
\rep \;\new n_R. P'_R\big).$$
\end{example}

As a last example,
we will consider one for which identity parameters only 
occur in one role. This example can be seen as a simplified
version of the Direct Anonymous Attestation (DAA) sign protocol that will be detailed in
Section~\ref{sec:casestudies}.

\begin{example}
  \label{ex:DAA}
  We consider a simplified version of the protocol DAA sign (adapted from~\cite{smyth2015formal}).
  Note that a comprehensive analysis of the protocol DAA sign (as well
  as the protocol DAA join) will be  conducted in
  \Cref{sec:casestudies}. Before describing the protocol itself,
  we introduce the term algebra that will allow us to model the signature
  and zero knowledge proofs used in that protocol. We consider:
  \begin{itemize}
  \item 
    $\Sigma_c=\{\sign,\zk,\pk,\langle \; \rangle, \tuple,\ok, \skIssuer, \error\}$, and
  \item 
    $\Sigma_d=\{\checksign, \checkzk,\publiczk, \proj_1, \proj_2, \proj^4_1, \proj^4_2, \proj^4_3,
    \proj^4_4\}$.
  \end{itemize}
  We consider the computation relation induced by the empty set of
  equations, and the rules:
  $$  \begin{array}{lcr}
        \checksign(\sign(x,y),\pk(y)) \;\to\; x &\hspace{1cm}\;\;\;\;\;\;\;\;\;\;&\proj_i(\langle y_1,y_2\rangle) \;\to \; y_i \mbox{\;\;\;\;\;$i \in \{1,2\}$}\\[1mm]
        \multicolumn{3}{l}{\checkzk(\zk(
        \sign(\langle x_k,x_\monid\rangle, z_\skex), x_k,
        \tuple(y_1,y_2,y_3,\pk(z_\skex))
        )   ) \; \to\; \ok} \\[1mm]
        \publiczk(\zk(x,y,z)) \;\to\; z &&\proj^4_i(\tuple(y_1,y_2,y_3,y_4)) \; \to \; y_i \mbox{\;\;\;\;\;$i \in
                                                   \{1, 2, 3, 4\}$}
\end{array}
      $$

      The protocol is between a client~$C$ (the responder) and a
      verifier~$V$ (the initiator).
      The client
      is willing to sign a 
      message~$m$ using a credential issued by some issuer
      and then he has to convince~$V$ that the latter signature is genuine.
      The client~$C$ has a long-term secret key~$k_C$, an
      identity~$\monid_C$, and some 
      credential $\cred_C=\sign(\langle k_C,\monid_C \rangle,\skIssuer)$ issued by some issuer~$I$
      having $\skIssuer$ as a long-term signature key.
      Such a credential would be typically obtained once and for all through a protocol similar to DAA join.
      We give below an Alice \& Bob description of the protocol: % (remember that $\cred_C=\sign((k_C,\id_C),\sk_I)$):
      $$\begin{array}{rrl}
         1.& V \to C: & n_V\\
         2.& C \to V: & \zk(
                        \cred_C, k_C,
                        \tuple(n_V,n_C,m,\pk(\skIssuer))
                        ) \\
                        % 3.& V \to C: & \ok \\
       \end{array}$$
The verifier starts by challenging the client with a fresh nonce,
  the latter then sends a complex zero-knowledge proof bound to this challenge proving that he knows
  a credential from the expected issuer bound to the secret~$k_C$ he knows.
Before accepting this zero-knowledge proof, the verifier~$V$ (i) checks  the validity of the zero-knowledge proof using the 
$\checkzk$ operator, and (ii) verifies that this proof is bound to the challenge $n_V$ and to the public key
of~$I$ using the $\publiczk$ operator.
The processes $P_C$ and $P_V$  are defined as follows:

\smallskip{}

$\begin{array}{rcl}
P_V  &:=&   \Out(c_V,n_V).\\
&&\In(c_V,x_1). \\
&&\Let \; x_2, \;x_3, \; x_4 = \\
&&\hfill \eq(\checkzk(x_1),\ok), \;
   \eq(\proj^4_1(\publiczk(x_1)),n_V), \;
   \eq(\proj^4_4(\publiczk(x_1)),\pk(\skIssuer))\; \In \; 0\\
&&\Else \; \Out(c_V, \error)
\end{array}$

\smallskip{}

$\begin{array}{rcl}
P_C &:=& \In(c_R,y_1). \\
&&\Out(c_R,\zk(\sign(\langle k_C, \monid_C\rangle, \skIssuer),k_C,\tuple(y_1, n_C, m, \pk(\skIssuer)))).\\
\end{array}$

\smallskip{}
  
This protocol falls in our class, the two parties being the verifier $P_V$ and the client $P_C$.
The protocol DAA sign, and the simplified version we consider here,
has been designed to provide privacy (\ie unlinkability and anonymity as defined in Section~\ref{sec:properties})
to users (\ie clients) inside a group associated to a single issuer.
In other words, the privacy set~\cite{pfitzmann2001anonymity} that is typically considered is the set of users
who obtained a credential from a single, given issuer.
Therefore, as we are interested in modelling different clients having credentials signed by the same issuer,
we model $\skIssuer$ as a private constant in $\Sigma_c\cap\Sigma_\priv$ rather than as an identity parameter
(we explore this different modelling choice in Example~\ref{ex:DAA-s}).

The tuple $\Pi_\DAA = ((k_C,\monid_C), n_V, (n_C, m), \rep,
\rep, P_V, P_C)$ is a protocol according to our
Definition~\ref{def:proto}. We have that $k_C$ or $\monid_C$ only occur in $P_C$,
and therefore following Definition~\ref{def:systemM}, we have that:
$$\mathcal{M}_{\Pi_\DAA}= \big({!}\; \new n_V. P_V \big) \;\mid \;  \big(!\;\new (k_C,\monid_C). \; !\; \new (n_C,m). P_C\big)$$
  This models infinitely many different clients who obtained credentials from
  a single issuer having the signature key $\skIssuer$.
  Any of those clients may take part to
  infinitely many sessions of the protocol
  with any verifier associated to that issuer,
  which executes always the same role (he has no proper identity).
We consider here a scenario where sessions can be executed
concurrently (both for clients and verifiers).
% We are interested in establishing ulinkability of the client and its anonymity w.r.t. the name $\id_C$.
% 
%
  We shall see that our verification methods allows one to automatically prove that privacy is preserved
  in this scenario.
\end{example}

\begin{example}[Continuing Example~\ref{ex:DAA}]
  \label{ex:DAA-s}
   The flexibility of our notion of protocol allows
    for a subtly different scenario to be analyzed by considering
    $\skIssuer$ as being identity-specific (\ie as an identity parameter)
    instead of being uniform for all identities (\ie private constant).
    Privacy would then be considered between users associated to
    different issuers (the privacy set being all users);
    this is a stronger property that the protocol is not expected to meet.
    Indeed, a verifier sends ZK proofs whose public parts contain the public
    key of its credential issuer.
We still consider the two parties $P_V$ and $P_C$ but
    we now model $\skIssuer$ as an identity parameter. Therefore, we remove
    $\skIssuer$ from $\Sigma_c$ defined in Example~\ref{ex:DAA}.
    The tuple $\Pi_\DAA = ((\skIssuer,k_C,\monid_C), n_V, (n_C, m), \rep,
    \rep, P_V, P_C)$ is a protocol according to our
    Definition~\ref{def:proto}. We have that $\skIssuer$ occurs both in $P_C$
    and $P_V$,
    and therefore following Definition~\ref{def:systemM}, we have that:
    $$\mathcal{M}_{\Pi_\DAA^{\skIssuer}}= {!}\;\new (\skIssuer,k_C,\monid_C). \; \big(
    {!}\; \new n_V. P_V
    \;\mid \;
    {!}\; \new (n_C,m). P_C\big)$$
    This models
    (i) infinitely many different clients who obtained pairwise different
    credentials from infinitely many issuers having pairwise different signature
    keys $\skIssuer$, and,
    (ii) infinitely many different verifiers who check for credentials that have been signed by pairwise different issuers.
    Any of those clients and verifiers may take part to
    infinitely many sessions of the protocol.
    We consider here a scenario where users and verifiers sessions can be executed
    concurrently.
    Note however that only one user per group (associated to a single issuer) is considered. Relaxing this constraint would require
    a more generic notion of protocols with 3 parties (this limitation will be discussed in Section~\ref{subsec:class}).
  % with any verifier associated to that issuer,
  % which executes always the same role (he has no proper identity).
  % 
    Even for such a weaker scenario,
    we shall see that our verification methods allows one to find that privacy (unlinkability and anonymity) is
    already broken (see Figure~\ref{fig:summary-running} and details in Section~\ref{sec:DAA-join-sign:sign}).
\end{example}

\paragraph{Discussion about shared and non-shared protocols.}
As mentioned earlier and shown in \Cref{def:systemM}, we distinguish two cases depending on whether \emph{(i)} both roles use
identity parameters (\ie when $\fn(\ini)\cap\vect{k}\neq\emptyset$ and $\fn(\res)\cap\vect{k}\neq\emptyset$)
or \emph{(ii)} only one role uses identity parameters (\ie when
$\fn(\ini)\cap\vect{k}=\emptyset$ and
$\fn(\res)\cap\vect{k}\neq\emptyset$, the other case being symmetrical).
The case \emph{(i)} corresponds to the case where we should consider
an arbitrary number of users
for each role, whereas regarding case \emph{(ii)} it is sufficient to
consider an arbitrary number of users for role~$\res$ only.
In addition to this distinction, note that there are two 
different kinds of protocols that
lie in class \emph{(i)}:
\begin{enumerate}
\item[\emph{(i-a)}] The \emph{shared case} when
  $\fn(\ini)\cap\fn(\res)\neq\emptyset$. In such a situation,
roles $\ini$ and $\res$ share names in $\fn(\ini)\cap\fn(\res)$. 
In practice, this shared knowledge may have been established in various ways
such as by using
prior protocols, using another communication channel (\eg optical scan of
a password as it is done with e-passports, use of PIN codes)
or by retrieving the identity from a database that matches the first received message
as it is often done with RFID protocols.
For such protocols, it is expected that an initiator user and a responder user
can communicate successfully producing an honest execution {\em only if}
they have the same identity (\ie they share the same names $\vect{k}$).
\item[\emph{(i-b)}] The \emph{non-shared case} when
$\fn(\ini)\cap\fn(\res)=\emptyset$.
In such a case, both roles do not share any specific prior knowledge, and it is
therefore expected that an initiator and a responder 
can communicate successfully producing an honest execution whatever
their identities.
\end{enumerate}
Unlinkability and anonymity will be uniformly expressed for the cases \emph{(i-a)} and \emph{(i-b)}
but our sufficient conditions will slightly differ depending on the
case under study.

\newcommand{\defe}{:=}
\newcommand{\sint}{\lrstep}

\section{Modelling security properties}
\label{sec:properties}

This section is dedicated to the definition of the security properties we seek 
to verify on protocols: unlinkability and anonymity.
Those properties are defined using the notion of {\em trace equivalence} which relates
indistinguishable processes.

%%
%% Trace equivalence
%%

\subsection{Trace equivalence}
\label{subsec:trace-equiv}

 Intuitively, two configurations are trace equivalent if an attacker
 cannot tell whether he is interacting with one or the other.
 Before formally defining this notion, we first introduce a notion of equivalence 
 between  frames, called \emph{static equivalence}.

\smallskip{}

\begin{definition}%[Static equivalence]
A frame $\phi$ is \emph{statically included} in $\phi'$
when $\dom(\phi) = \dom(\phi')$, and
\begin{itemize}
\item for any recipe $R$
  such that $R\phi \redc u$ for some $u$,
 we have that $R\phi'\redc u'$ for some $u'$;
\item for any recipes $R_1,R_2$ such that
  $R_1\phi\redc u_1$, $R_2\phi\redc u_2$, and $u_1 \theo
  u_2$, we have that $R_1\phi'\redv =_\E R_2\phi'\redv$, \ie there exist $v_1, v_2$ such that
  $R_1\phi'\redc v_1$, $R_2\phi'\redc v_2$, and $v_1 =_\E v_2$.
\end{itemize}
Two frames $\phi$ and $\phi'$ are in \emph{static equivalence}, written $\phi
\sim \phi'$, if the two static inclusions hold.
\end{definition}

\smallskip{}

Intuitively, an attacker can distinguish two frames if he is able to
perform some computation (or a test) that succeeds in~$\phi$ and
fails in~$\phi'$ (or the converse).

\begin{example}
 \label{ex:static}
Let $\phi_0$ be the frame given in Example~\ref{ex:execution},
we have that 
$\phi_0 \sqcup \{w_4\mapsto k'\} \not\sim \phi_0 \sqcup
\{w_4 \mapsto k''\}$.
An attacker may observe a difference relying on the computation $R =
\sdec(w_2,w_4)$.
% succeeds on the left but fails on the right.
\end{example}

Then, \emph{trace equivalence} is the active counterpart of static
equivalence, taking into account the fact that the attacker may
interfere during the execution of the process.
% in order to distinguish
%between the two situations.
In order to define this, we first introduce $\trace(K)$
for a configuration $K = (\p;\phi)$:
$$
\trace(K) = \{(\tr,\phi') ~|~ (\p,\phi) \lrstep{\tr} (\p'; \phi') 
 \mbox{ for some configuration $(\p'; \phi')$}\}.
$$

\begin{definition}%[Trace equivalence]
  Let $K$ and $K'$ be two configurations. We say that $K$ is \emph{trace
  included} in $K'$, written $K \sqsubseteq K'$, when,
  for any $(\tr,\phi) \in \trace(K)$ 
  there exists $(\tr', \phi') \in \trace(K')$ such
  that $\obs(\tr') = \obs(\tr)$
  and $\phi \sim \phi'$.
  They are in \emph{trace equivalence}, written $K \approx K'$, 
  when $K \sqsubseteq K'$ and $K' \sqsubseteq K$.
\end{definition}

\begin{example}
\label{ex:trace-equiv}
Resuming \Cref{ex:protocol-FH}, 
we may be interested in checking whether the configurations
$K = (! P_{\Pi_\Feldhofer}; \emptyset)$ and 
$K'= (\mathcal{M}_{\Pi_\Feldhofer};  \emptyset)$
are in trace equivalence.
This equivalence  models the fact that
$\Pi_\Feldhofer$ is unlinkable: each session of the protocol appears to an attacker as if it has been
initiated by a different tag, since a given tag can perform at most one 
session in the idealised scenario $K$. This equivalence actually
holds.
It is non-trivial, and cannot be established using
existing verification tools such as \proverif~or \tamarin. The
technique developed in this paper will notably allow one to establish it
automatically.
\end{example}

%%
%% Review of the definitions as proposed in Myrto's paper
%%

\subsection{Security properties under study}

In this paper, we focus on two privacy-related properties, namely
\emph{unlinkability} and \emph{anonymity}.

%%
%% Unlinkability 
%%

\subsubsection{Unlinkability}
According to the ISO/IEC standard 15408~\cite{ISO15408}, unlinkability aims at
ensuring that a user may make multiple uses of a service or a resource
without others being able to link these uses together. In terms of our
modelling, a protocol preserves unlinkability if any two sessions of a
same role look to an outsider as if they have been executed with
different identity names. In other words, an ideal version of the
protocol with respect to unlinkability, allows the roles $\ini$ and
$\res$ to be executed at most once for each identity names. An outside
observer should then not be able to tell the difference between the
original protocol and the ideal version of this protocol.

In order to precisely define this notion, we have to formally define
this ideal version of a protocol $\Pi$. This ideal version, denoted
$\mathcal{S}_\Pi$, represents an arbitrary number of agents that can
at most execute one session each. Such a process is obtained from
$\mathcal{M}_\Pi$ by simply removing the symbols $\rep$ and $\rec$
that are in the scope of identity names. Indeed, those constructs
enable each identity to execute an arbitrary number of sessions (respectively
concurrently and sequentially).
Formally, depending on whether identity
names occur in both roles, or only in the responder role, 
this leads to slightly different definitions.

\begin{definition}
Given a protocol $\Pi = (\vect k, \vect n_I, \vect n_R, \dag_I,
\dag_R, \ini, \res)$, the process $\mathcal{S}_\Pi$ is defined as
follows:
\begin{itemize}
\item If $\vect k \cap \fn(\ini) \neq \emptyset$ and $\vect k \cap
  \fn(\res) \neq \emptyset$, then
  $\pS \defe\; !\; \new \vect k.(\new \vect n_I.\ini \; \mid \; \new \vect n_R. \res)$;
\item  If $\vect k \cap \fn(\ini) = \emptyset$ and $\vect k \cap
  \fn(\res) \neq \emptyset$, then
  $\pS \defe \dag_I\; \new \vect n_I. \ini \; \mid \; \rep \;\new \vect k. \new \vect n_R. \res$.
\end{itemize}
\end{definition}

Unlinkability is defined as a trace equivalence between $\pS$ (where each identity can
execute at most one session) and $\pM$ (where each identity can execute an arbitrary number
of sessions).

\begin{definition}
\label{def:un}
A protocol $\Pi = (\vect k, \vect n_I, \vect n_R, \dag_I, \dag_R, \ini, \res)$
  ensures \emph{unlinkability} if
$\mathcal{M}_\Pi \eint \mathcal{S}_\Pi$.
\end{definition}

\begin{example}
\label{ex:un-FH}
Going back to our running example (\Cref{ex:protocol-FH}), unlinkability is expressed through
the equivalence given in Example~\ref{ex:trace-equiv} and recalled below:
$$
!\; \new k. (!\; \new n_I. P_I \; \mid \; !\; \new n_R. P_R)
\; \eint\;
!\; \new k. ( \new n_I. P_I \; \mid \; \new n_R. P_R).
$$

This intuitively represents the fact that the real situation where a
tag and a reader may execute many sessions in parallel is
indistinguishable from an idealized one where a given tag and a given
reader can execute at most one session for each identity.
\end{example}

Although unlinkability of only one role (\eg the tag
for RFID protocols) is often considered in the literature
(including~\cite{arapinis-csf10}), we consider a stronger notion here since both
roles are treated symmetrically. As illustrated through the case
studies developed in Section~\ref{sec:casestudies} (see
Sections~\ref{subsec:lak} and~\ref{subsec:pace}), this is actually needed to not miss
some practical attacks.

\begin{example}
\label{ex:toyseq-att}
We consider the variant of the toy protocol described in
\Cref{ex:toyseq} where concurrent sessions are authorised for the
initiator:
  $\Pi^{!}_\Toy \defe (k, (n_T), (n_R), \rep, \rep, P'_T, P'_R)$.
We may be interested in checking unlinkability as in
Example~\ref{ex:un-FH}, i.e. whether the following equivalence
holds or not:
$$
!\,\new k.\big(! \;\new
n_T. P'_T
\;\mid\; ! \; \new n_R. P'_R\big)
\eint\;\; 
!\,\new k.\big(\new n_T. P'_T
\;\mid\; \new n_R. P'_R\big)
$$
Actually, this equivalence does not hold. When concurrent sessions are
authorised, the following scenario is possible: two tags of the
same identity can start a session.
Then, the attacker just forwards messages from one tag to the
other. They can thus complete the protocol. In particular,  the
mac-key verification stage goes well and the attacker observes that
the last conditional of the two tags holds. 
Such a scenario
(which is possible on the left-hand side of the equivalence) cannot be mimicked
on the right-hand side (each tag can
execute only once).
 Therefore we have a trace that can only be executed by the multiple 
  sessions process: the equivalence does not hold.

However, we shall see that the original toy protocol of
\Cref{ex:toyseq} (with sessions of the initiator running sequentially only) can be shown
unlinkable using the technique developed in this paper.
Formally, the following equivalence holds:
$$
 !\,\new k.(\new n_T. P'_T
\;\mid\; \new n_R. P'_R)
\;\eint\;  !\,\new k.(\rec \;\new
n_T. P'_T
\;\mid\; \rec \; \new n_R. P'_R)
$$

\end{example}

\subsubsection{Anonymity}

According to the ISO/IEC standard 15408~\cite{ISO15408}, anonymity aims at
ensuring that a user may use a service or a resource without
disclosing its identity. In terms of our modelling, a protocol
preserves anonymity of some identities $\vect \id \subseteq \vect k$,
if a session executed with some particular
(public) identities $\vect\idzero$ looks to
an outsider as if it has been executed with different identity
names. In other words, an outside observer should not be able to tell
the difference between the original protocol and  a version of the
protocol where the attacker knows that specific roles $\ini$ and $\res$ with
identities $\idzero$ (known by the attacker) are present.

\begin{definition}
Given a protocol $\Pi = (\vect k, \nI, \nR, \dag_I,
  \dag_R, \ini, \res)$, and $\vect \id \subseteq \vect k$, the process 
  $\pMa$ is
defined as follows:
\begin{itemize}
\item  If $\vect k \cap \fn(\ini) \neq\emptyset$ and $\vect k \cap
  \fn(\res) \neq \emptyset$, then ${\mathcal{M}_{\Pi,\vect \id} \defe
  \mathcal{M}_\Pi \; \mid \; \new \vect k. (\dag_I\, \new
  \nI. \ini_0 \; \mid \; \dag_R\, \new \nR. \res_0)}$.
\item If $\vect k \cap \fn(\ini) = \emptyset$ and $\vect k \cap
  \fn(\res) \neq \emptyset$, then $\mathcal{M}_{\Pi,\vect \id} \defe
  \mathcal{M}_\Pi \; \mid \; \new \vect k.\,\dag_R\, \new \nR. \res_0$.
\end{itemize}
where $\ini_0  =\ini\{\vect \id \mapsto \vect \idzero\}$ and $\res_0 =
\res\{\vect \id \mapsto \vect \idzero\}$ for some fresh public constants
$\vect\idzero$.
%\lum{fresh: ok?} S.D.: oui
\end{definition}

\begin{definition}
\label{def:anonymity}
Let $\Pi = (\vect k,\nI, \nR, \dag_I, \dag_R, \ini,
\res)$, and $\vect \id \subseteq \vect k$. We say that $\Pi$ ensures
\emph{anonymity w.r.t. $\vect \id$} if $\mathcal{M}_{\Pi, \vect \id} \eint \mathcal{M}_{\Pi}$. 
\end{definition}

\begin{example}
\label{ex:DAA-anonymity}
Going back to Example~\ref{ex:DAA}, anonymity w.r.t.\ identity of the client (\ie $\monid_C$)
is expressed through the following equivalence:
% (where $\vect{m_C}$
%represents the sequence $n_C$, $m$):
$$
\begin{array}{l}
  !\; \new n_V. P_V \, \mid \, (! \;
\new (k_C,\monid_C). \; !\; \new (n_C, m). P_C) 
\, \mid \,
  (\new (k_C,\monid_C). \rep \; \new (n_C, m). P_C\{\monid_C \mapsto \mathsf{id_0}\}) \\
\; \eint \; 
  !\; \new n_V. P_V \;\mid \; (! \;
\new (k_C,\monid_C). \; !\; \new (n_C, m). P_C) 
\end{array}
$$
This intuitively represents the fact that the situation in which a
specific client with some known identity~$\mathsf{id_0}$ may execute
some sessions is indistinguishable from a situation in which this
client is not present at all. Therefore, if these two situations are
indeed indistinguishable from the point of view of the attacker, it
would mean that there is no way for the attacker to deduce whether a
client with a specific identity is present or not. 
\end{example}
%%
%% Discussion
%%

\subsection{Discussion}
\label{sec:priv:prop:discu}

% {21/12: Je commence a m'y perdre dans cette section. Si le but
%   est d'illustrer que la game based rate des attaques, est-ce que la
%   variante FH ou l'on aura exhiber une attaque
%   d'unlinkability avec une definition a la Myrto ne pourrait pas faire
%   l'affaire. En plus, ProVerif arrive a conclure tout seul si
%   j'utilise deux fonctions de chiffrement differentes (ou en mettant
%   des tags, cela devrait passer aussi). Mon codage est sous svn dans
%   le fichier: FH-variant-game-based.pv .}

% The definitions we proposed are variations of the ones proposed
% in~\cite{arapinis-csf10}. In particular, they all share the same pattern: we
% compare some process modelling real usage of the protocol with another
% process modelling an idealised version of the system. 
%   We may note that we consider here the notion of trace
% equivalence instead of the stronger notion of labeled bisimilarity used
% in~\cite{arapinis-csf10}.

%%%%

{
The notion of strong unlinkability that we consider is inspired
by~\cite{arapinis-csf10}. In this paper, the authors first propose
a definition of \emph{weak unlinkability} that is not expressed
via a process equivalence, then they give a notion of
\emph{strong unlinkability} that implies the former notion and
is expressed via a labelled bisimilarity. The authors argue
that, compared to weak unlinkability, the strong variant is
too constraining but has the advantage of being more amenable to
verification. The first claim is based on an example
protocol~\cite[Theorem 1]{arapinis-csf10} where a reader emits
an observable ``beep'' when it sees the same tag twice, which
breaks strong unlinkability but not weak unlinkability. 
Unlike the authors of~\cite{arapinis-csf10},
we do not consider this to be a spurious attack, but a potentially threatening
linkability issue.
The second claim is only substantiated by the fact that tools exist
for automatically verifying bisimilarities. As discussed before,
this is not sufficient. Moreover, there might be spurious attacks on
strong unlinkability just because bisimilarity is a very restrictive
equivalence:
for this reason
we would also consider that the strong unlinkability of~\cite{arapinis-csf10}
is too strong, and instead advocate for our variant based on trace 
equivalence.}

% In \cite{arapinis-csf10}, a formal notion of \emph{unlinkability}
% is defined, which is considered by the authors as realistic but
% hard to verify. They propose the stronger notion of \emph{strong 
% unlinkability}, based on labelled bisimilarity, as a better candidate
% for automated verification. They consider the notion as being too strong,
% as it can fail for protocols where there is, according to them, no privacy 
% weakness; we do not agree on this kind of example, though it is likely that 
% using bisimilarity rather than trace equivalence is indeed too strong.

We now show formally that, in the setting that we consider,
our notion of unlinkability (\cref{def:un}) indeed corresponds to
the strong unlinkability of~\cite[Definition 12]{arapinis-csf10}
where trace equivalence is required rather than bisimilarity.
In this original formulation, strong unlinkability is a property of one
specific role and not of the whole protocol. As a more technical difference,
protocols in \cite{arapinis-csf10} may involve more than two roles,
and agents may use private channels. In practice, this is used to
communicate honest identities in setup phases, as is the case
in their BAC case study. If we specialise the setting of
\cite{arapinis-csf10} to two roles $R$ and $T$ (for Reader and Tag) which
fall into the format of \cref{def:priv:role:grammar} and do not use the
distinguished private channel $c$, strong unlinkability
of the tag role $T$ corresponds to the following
labelled bisimilarity:
\begin{equation} \label{eqn:unlink1} \begin{array}{rl}
  &
\new c.~ \bigl(
  (\rep\; \new \vect k.~
   {\rep\;}
     \Out(c,\vect k).
     \new \vect n_T. T)
  \mid
  (\rep\;
     \In(c,\vect k).
     \new \vect n_R. R)
\bigr)
  \\ \approx_\ell\phantom{x} &
\new c.~ \bigl(
  (\rep\; \new \vect k.~
   \phantom{\rep\;}
     \Out(c,\vect k).
     \new \vect n_T. T)
  \mid
  (\rep\;
     \In(c,\vect k).
     \new \vect n_R. R)
\bigr)
\end{array} \end{equation}
As communications on channel $c$ are private, \cref{eqn:unlink1}
is equivalent to:
\begin{equation} \label{eqn:unlink2}
\rep\;
  \new \vect k.~ \bigl(
  (\rep\;
   \new \vect n_T. T)
  \mid
  (\rep\;
   \new \vect n_R. R)
\bigr)
  \quad \approx_\ell \quad
\rep\;
  \new \vect k.~ \bigl(
  (%\phantom{\rep\;}
   \new \vect n_T. T)
  \mid
  (%\phantom{\rep\;}
   \new \vect n_R. R)
\bigr)
\end{equation}
The key observation here is that, even though we had only removed replication for the tag role (on the 
right of \cref{eqn:unlink1}), replications are
removed for both tags and readers in \cref{eqn:unlink2} because communications 
on $c$ are linear (\ie can be triggered only once).
Thus, in this particular case, the only difference between
strong unlinkability and our unlinkability is that we rely on
trace equivalence rather than labelled bisimilarity.
% (\david{REMOVE?}
% At a high level, our argument is simply saying that there is no point
% constraining the reader to a single session per identity when a reader's
% role is not tied to a particular identity, as is the case in 
% \cref{eqn:unlink1}.)

% \New{ok pour remover}
%%%%

\medskip

Several other definitions of unlinkability have been proposed in the 
literature (see, \eg~\cite{bruso2012linking,bruso2014dissecting} for a 
comparison). In particular, % Among the strongest ones,
various game-based formulations have been considered, both in
the computational and symbolic models.
% \paragraph*{Game-based definition.}
We first discuss the most common kind of games,
called \emph{two-agents games} in \cite{bruso2012linking} and seen
e.g.\ in~\cite{backes2008zero,juels2009defining,ck10csf}. As we
shall see, these games can be accurately verified through diff-equivalence,
but systematically miss some linkability attacks.
We will not need any formal definition, but simply rely on the general idea
behind these games, which run in two phases:
\begin{enumerate}
\item \emph{Learning phase:}
During this phase, the attacker can trigger an arbitrary number of 
sessions of the two roles (namely tag and reader) with the identity of
his choice. This allows him to gain some knowledge.
Eventually, the attacker chooses to end the learning
phase and enter the second phase.~
\item \emph{Guessing phase:}
The challenger
chooses an identity~$x$ among two distinguished identities
$\id_1$ and~$\id_2$.
The attacker is allowed to interact again (an arbitrary number of times)
with roles of $x$, or of identities other than $\id_1$ and~$\id_2$.
\end{enumerate}
The attacker wins the game if he can infer whether $x$ is $\id_1$ or
$\id_2$, \emph{i.e.} if he is able to  
distinguish between these two scenarios.
The following example shows that these two-agent games miss some
linkability attacks, and do not imply unlinkability in our sense for 
this reason.

\begin{example} \label{ex:concurrent-attack}
We consider a protocol between a tag $T$ and a reader $R$ sharing 
a
symmetric key~$k$. We consider that sessions can be executed in
parallel, and
we assume that $T$ aborts in case the
nonce $n_R$ he receives is equal to the nonce $n_T$ he sent
previously (in the same session).
$$  \begin{array}{rrll}
    1.& T \to R: & \{n_T\}_k & \\
    2.& R \to T: & \{n_R\}_k & \\ %T \text{ aborts if }n_R = n_T\\
    3.& T \to R: & \{n_R\oplus n_T\}_k &
  \end{array}$$

We consider the term algebra introduced in
Example~\ref{ex:signature}, and the equational theory introduced in
Example~\ref{ex:xor} with in addition the equation
$\sdec(\senc(x,y),y) \; = \; x$.
To show that the property formally stated in Definition~\ref{def:un}
does not hold, consider the following scenario.
$$
\begin{array}{l}
 1. \; T \to R: \; \{n_T\}_k   \\[-1mm]
\hspace{4.5cm} 1'. \; T' \to R: \; \{n'_T\}_k  \\[-1mm]
    2.\; I(R) \to T: \; \{n'_T\}_k  \\[-1mm]
    \hspace{4.5cm} 2'. \; I(R) \to T': \; \{n_T\}_k \\[-1mm]
3. \; T \to R: \; \{n'_T \oplus n_T\}_k \\[-1mm]
\hspace{4.5cm} 3'. \; T' \to R: \; \{n_T \oplus n'_T\}_k
\end{array}$$

A same tag starts two sessions\footnote{This is possible if different {\em physical} tags
share the same identity, as may be the case e.g.\ in access control scenarios. In such cases,
two different {\em physical} tags may run sessions concurrently.}
and therefore generates two nonces
$n_T$ and $n'_T$. The attacker answers to these requests by sending
back the two encrypted messages to the tag who will accept both of
them, and sends on the network two messages that are actually equal
(the exclusive or operator is commutative).
Therefore the attacker observes a test, namely the equality between the
last two messages, which has no counterpart in the single session
scenario. Therefore, this protocol does not ensure unlinkability.
In practice, this can be very harmful. Suppose, for example,
that tags are distributed among distinct groups (\eg for access control 
policies) sharing each the same key~$k$.
By interacting with two tags, the attacker would then be able to know
if they belong to the same group and thus be able to trace groups.
\end{example}

The previous example illustrates a general phenomenon: two-agent games do not 
capture concurrent attacks. This is also seen with the protocol of
\Cref{ex:toyseq}, which suffers from the attack shown in
\Cref{ex:toyseq-att}, but is secure in the sense of
two-agent games --- this can actually be proved in \proverif because the 
protocol does not involve the exclusive-or primitive.
Due to this general weakness, two-agent games do not adequately express 
unlinkability. They are however convenient for automation, as they
can be directly and accurately expressed using the notions of
diff-equivalence available in \proverif or \textsf{Tamarin}.
For instance, two-agent games have been used in~\cite{backes2008zero}
for unbounded sessions of the DAA protocols ---
although in this work the security property expressed in this way
is called pseudonymity rather than unlinkability.

  As pointed out in \cite{bruso2012linking}, \emph{three-agent games} have also been
considered where the challenge phase is
changed as follows: the attacker chooses three tags $(a,a_1,a_2)$ and must
distinguish interactions with several tags including
tags $x$ and $y$ with the same identity as $a$, and interactions
with $x$ and $y$ having the respective identities of $a_1$ and $a_2$.
This allows to capture the attacks described above which the two-agent games 
missed.
Three-agent games have successfully been used in \cite{kostas-csf10}
for automated verification of unlinkability, though only for a restrictive 
class of protocols and for bounded sessions only.

We suspect that three-agent games still miss some linkability attacks, though
counter-example protocols are likely to be artificial.
Further generalisations of these games could then be considered to obtain
stronger security properties, and get closer to our notion of unlinkability.
In any case, it is important to remark that this line of thought fundamentally 
relies on having a centralised reader since the attacker must distinguish 
between scenarios that differ only in the identities of some tags.
This contrasts with our notion of unlinkability, which does not assume a 
centralised reader but treats symmetrically the tag and reader role, more
generally called initiator and responder. Such a symmetric treatment is 
required to model unlinkability when the two parties share a dedicated channel
or have an initial shared knowledge, \eg in secure messaging protocols.
We also argue that our definition has some value even when analysing protocols
featuring centralised readers. In such cases, having reader roles expecting a 
specific identity seems artificial, but it can actually be seen as a way to 
model the successive states of a reader (e.g. in LAK, where the tags and 
reader evolve a common state almost in synchronisation) or a pre-established 
communication (e.g. in BAC or PACE, where an optical scan is performed to 
securely exchange a first secret). In any case, our analysis of the 
aforementioned protocols using our notion of unlinkability has revealed actual 
attacks that were previously unknown (see \Cref{sec:casestudies}).

\section{Our approach}
\label{sec:approach}

We now define our two conditions, namely frame opacity and 
well-authentication, and our result which states that these
conditions are sufficient to ensure unlinkability and anonymity as
defined in Section~\ref{sec:properties}.
Before doing that, we shall introduce annotations in the semantics
of our processes, in order to ease their analysis.

%%
%% Annotations
%%

\subsection{Annotations}
\label{subsec:annotations}

We shall now define an annotated semantics whose transitions 
are equipped with more informative actions. The annotated actions will
feature labels identifying  which concurrent process has performed the
action. This will allow us to identify which specific agent (with some specific identity
and session names) performed some action. 

Given a protocol $\Pi =  (\vect{k}, \vect n_I, \vect n_R,
\dag_I, \dag_R, \ini, \res)$ and $\vect{id} \subseteq \vect{k}$, consider any execution of $\pMa$, $\pM$ or $\pS$.
In such an execution, $\tau$
actions are solely used to create new 
agents (\ie instantiations of~$\ini$ and~$\res$ with new names
or constants from~$\idzero$)
by unfolding replications (\ie~$\rep$) or repetitions (\ie~$\rec$),
breaking parallel compositions or choosing
fresh session and identity parameters. 
Actions other than $\tau$
(that is, input, output and conditionals) are then 
only performed by the created agents.
Formally,
we say that an agent is either an instantiation of one of the two roles with
some identity and session parameters, or its continuation after the execution
of some actions.
When $\dag_A = \,!$, agents of role $A$ are simply found at toplevel in the
multiset of processes.
When $\dag_A = \,\rec$, they may be followed by another process.
For instance, in traces of $\pM$ when $\dag_I = \rec$, newly created
initiator agents occur on the left of the sequence in processes of the
form:
$$\ini\{\vect k \mapsto \vect l, \vect n_I \mapsto \vect n\};
 \rec \, \new \vect m.~ \ini\{\vect k\mapsto \vect l, \vect n_I \mapsto \vect m\}.$$

The previous remark allows us to define an \emph{annotated semantics} for our processes of
interest.
We consider \emph{annotations} of the form $\aagent(\vect k, \vect n)$ where
$\aagent \in \{\aini,\ares\}$ and $\vect k$, $\vect n$ are sequences
of names, or constants from~$\vect{\idzero}$.
Annotations are noted with the letter~$a$,
and the set of annotations is noted~$\agents$.
We can then define an annotated semantics, where agents are
decorated by such annotations, indicating their identity and session parameters.
An agent $P$ decorated with the annotation $a$ is written $P[a]$,
and the actions it performs are also decorated with $a$, written $\alpha[a]$.
Note that this includes $\taut$ and $\taue$ actions;
in the annotated semantics, the only non-annotated action is $\tau$.
For instance, let us consider $\pM$ in the annotated semantics
when $\dag_I = \rec$. Newly created initiator agents now appear as
$\ini\{\vect k \mapsto \vect l, \vect n_I \mapsto \vect n\}[I(\vect l, \vect n)];
 \rec \, \new \vect m.~ \ini\{\vect k\mapsto \vect l, \vect n_I \mapsto \vect m\}$;
they execute actions of the form $\alpha[I(\vect l, \vect n)]$
with $\alpha\neq\tau$;
upon termination of the agent, unannotated $\tau$ actions can be executed
to create a new agent annotated $I(\vect l, \vect m)$ for fresh names $\vect m$.
We stress that agents having constants $\vect{\idzero}$ as identity parameters shall be annotated
with some $\aagent(\vect k, \vect n)$ where $\vect{\idzero}\subseteq \vect k$. 
Intuitively, in such a case, we keep in the annotation the information that 
the identity parameters $\vect{\idzero}$ of that agent has been disclosed to the attacker.

Traces of the annotated semantics will be denoted by $\ta$.
We assume\footnote{
  This assumption only serves the purpose of uniquely identifying agents.
  The assumed session nonces do not have to occur in the corresponding roles,
  so this does not require to change the protocol under study.
} that $\vect{n}_I \neq \emptyset$ and $\vect{n}_R \neq \emptyset$,
so that at any point in the execution of an annotated trace, an annotation $a$ 
may not decorate more than one agent in the configuration. Thus, 
an annotated action may be uniquely traced back to the annotated process that performed 
it.
We also assume that labels used to decorate output actions (\ie elements of $\mathcal{L}$)
are added to the produced output actions
so that we can refer to them when needed:
output actions are thus of the form $\ell:\Out(c,w)\annot{a}$.

In annotated traces, $\tau$ actions
are not really important. We sometimes need
  to reason
up to these~$\tau$ actions.
Given two annotated trace $\ta$ and~$\ta'$, we write $\ta \upto \ta'$
when both traces together with their annotations are equal up to some~$\tau$ actions (but not $\taut$ and $\taue$). 
We write $ K \LRstep{\ta} K'$ when $K \lrstep{\ta'} K'$ for some $\ta'$
such that ${\ta \upto \ta'}$.

\begin{example}
\label{ex:annotation}
Considering the protocol $\Pi_\Feldhofer$ defined in Example~\ref{ex:protocol-FH},
process $\mathcal{S}_{\Pi_\Feldhofer}$ can notably  perform the execution seen in 
Example~\ref{ex:execution}.
The annotated execution has the trace $\ta$ given below (up to some $\tau$),
where $k'$, $n'_I$ and $n'_R$ are fresh names, 
$a_I = \aini(k',n'_I)$ and $a_R =
\ares(k',n'_R)$:
 $$
\begin{array}{ll}
  \ta = &\ell_1:
          \Out(c_I,w_1)[a_I].\In(c_R,w_1)[a_R]. \taut[a_R].\\
& \ell_2:\Out(c_R,w_2)[a_R].\In(c_I,w_2)[a_I].\taut[a_I]. \\
&\ell_3:\Out(c_I,w_3)[a_I].\In(c_R,w_3)[a_R].\taut[a_R]
\end{array}
  $$
After the initial $\tau$ actions, the annotated configuration
is
$
  (\{
    \ini\sigma_I[a_I],\,
    \res\sigma_R[a_R],\,\pS
  \}; \emptyset)
$
where $\sigma_I = \{k \mapsto k', n_I \mapsto n'_I\}$, and $\sigma_R = \{k \mapsto k', n_R \mapsto n'_R\}$.
The structure is preserved for the rest of the execution 
with three processes in the multiset (until they become null).
%two of which remaining annotated with $a_I$ and $a_R$. 
After $\ta$, the annotated configuration is $
  (\{
   {\mathcal{S}_{\Pi_{\Feldhofer}}}
   \}; \phi_0)
$
where~$\phi_0$ has been defined in \Cref{ex:execution}.
\end{example}

\begin{example}
\label{ex:annotation-bis}
Going back to Example~\ref{ex:DAA-anonymity}
and starting with $\pMa$,
a possible annotated configuration obtained
after some $\tau$ actions
can be $K = (\p; \emptyset)$
where $\p$ is a multiset containing:
 %the following processes:
\begin{itemize}
\item $P_C\{k_C \mapsto k^0_C, \id_C \mapsto \idzero, n_C \mapsto
  n^0_C, m \mapsto m^0_C\}[a^0_C]$;
\item $\rep \new \, (n_C, m). P_C\{k_C \mapsto k^0_C, \id_C \mapsto
  \idzero\}$; 
\item $P_V\{n_V \mapsto  n^1_V\}[a^1_V]$; and
\item $\mathcal{M}_{\Pi_\DAA}$.
\end{itemize}
\noindent where $a^1_V = \ini(\epsilon, n^1_V)$ and $a^0_C =
\res((k^0_C, \idzero), (n^0_C, m^0))$.
We may note that the annotation $a^1_V$ contains the empty sequence
$\epsilon$ since the initiator role does not rely on identity names;
and the annotation $a^0_C$ contains~$\idzero$.
\end{example}

%%
%% Frame opacity
%%

\subsection{Frame opacity}
\label{subsec:frame-opacity}

In light of attacks based on leakage from messages where non-trivial 
relations between outputted messages are exploited by the attacker to trace an agent,
our first condition will express that all relations
the attacker can establish on output messages only depend on what is
already observable by him
and never depend on a priori hidden information such as identity names of specific agents.
Therefore, such relations cannot be exploited by the attacker to learn anything new about 
the agents involved in the execution.
We achieve this by requiring that any reachable frame must be 
indistinguishable  from an \emph{idealised frame} that only depends on data 
already observed in the execution, and not on the specific agents (and their 
names) of that execution.

As a first approximation, one might take the idealisation of a frame
$\{w_1 \mapsto u_1, \ldots, w_l \mapsto u_n\}$
to be $\{w_1 \mapsto n_1, \ldots, w_l \mapsto n_l\}$ where the
$n_1, \ldots, n_l$ are
distinct fresh names. It would then be very strong to require that frames 
obtained in arbitrary protocol executions are statically equivalent to their 
idealisation defined in this way. Although this would allow us to carry out
our theoretical development, it would not be realistic since any protocol 
using, \eg a pair, would fail to satisfy this condition.
We thus need a notion of idealisation that retains part of the shape of 
messages, which a priori does not reveal anything sensitive to the attacker. 
We also want to allow outputs to depend
on session names or previous inputs in ways that are observable, \eg
to cover the output of the signature of a previously inputted message.

Our idealised frames will be obtained
by replacing each message, produced by an output of label $\ell$,
by a context that only depends on $\ell$,
whose holes are filled with fresh session names and (idealisations of) previously 
inputted messages.
Intuitively, this is still enough to ensure that the attacker does not learn
anything that is identity-specific.
In order to formalise this notion,
we assume two disjoint and countable subsets of variables:
\emph{input variables} $\X^\varI=\{x^\varI_1,x^\varI_2,\ldots\}\subseteq\X$, and
%session
 \emph{name variables} $\X^\varN
 =\{x^\varN_1,x^\varN_2,\ldots\}\subseteq\X$. We also consider 
a fixed but arbitrary \emph{idealisation operator}
$\ideam{\cdot} : \mathcal{L} \to \T(\Sigma, \X^\varI\cup\X^\varN)$.
Variables $x^\mathsf{i}_j$ intuitively refers to the $j$-nth variable received by the agent of interest.
Therefore, we assume that our  idealisation operator satisfies the following:
for all $\ell\in\mathcal{L}$, we have that
$\ideam{\ell}\cap\X^\mathsf{i}\subseteq\{x^\mathsf{i}_1,\ldots,x^\mathsf{i}_k\}$ where $k$
is the number of inputs preceding the output labelled $\ell$.

\begin{definition}%[Idealized frames]
  \label{def:idealphi}
  Let $\fr : \agents\times\X^\varN\to \N$ be an injective function
  assigning names to each agent and name variable.
  We define the idealised frame associated to $\ta$, denoted $\ideaf^\fr(\ta)$, inductively on the annotated trace $\ta$:
  \begin{itemize}
  \item $\ideaf^\fr(\epsilon)=\emptyset$ and
    $\ideaf^\fr(\ta.\alpha)=\ideaf^\fr(\ta)$ if
    $\alpha$ is not an output;
  \item
    $\ideaf^\fr\bigl(\ta.(\ell:\Out(c,w)\annot{a})\bigr)
    =\ideaf^\fr(\ta)\cup\{w\mapsto \ideam{\ell}\sigma^\varI\sigma^\varN
      \redv      \}$
    where
      \begin{itemize}
        \item
          $\sigma^\varN(x^\varN_j)=\fr(a,x^\varN_j)$ when $x^\varN_j
          \in \X^\varN$, and
        \item
          $\sigma^\varI(x^\varI_j)=R_j\ideaf^\fr(\ta)$
          when $x^\varI_j \in \X^\varI$ and $R_j$ is the recipe corresponding to
          the $j$-th input of agent $a$ in $\ta$.
   %       \dam{attn: glissement agent/annotation OK pour moi}
      \end{itemize}
  \end{itemize}
\end{definition}

We may note this notion is not necessarily well-defined,
  as $\ideam{\ell}\sigma^\varI\sigma^\varN$
  may not compute to a message. Note also that  well-definedness does not depend on the
  choice of the function~$\fr$.
Remark also that, by definition, $\ideaf^\fr(\ta)$ never depends on the specific identity names occurring
in $\ta$.
In particular, idealised frames do not depend on whether agents rely
on the specific constants $\vect{\idzero}$ or not.

\begin{example}
\label{ex:idealization-FH}
Continuing Example~\ref{ex:annotation}, we consider the
idealisation operator defined as follows:
$\ell_1 \mapsto x^\varN_1, \, \ell_2 \mapsto x^\varN_2, \, \ell_3 \mapsto
x^\varN_3$.
Let $\fr$ be an injective function such that
$\fr(a_I,x^\varN_j) = n^I_j$ and  ${\fr(a_R,x^\varN_j) =
n^R_j}$. We have that
$\ideaf^\fr(\ta) = \{w_1 \mapsto n^I_1 , w_2 \mapsto n^R_2, w_3 \mapsto  n^I_3\}$.

\end{example}

On the latter simple example, such an idealisation will be
sufficient to establish that any reachable frame obtained through an
execution of $\mathcal{M}_{\Pi_\Feldhofer}$ is indistinguishable
from its idealisation.
However, as illustrated by the following two examples, we sometimes need to consider more complex idealisation
operators.

\begin{example}
\label{ex:idealization-toyseq-variant}
Continuing Example~\ref{ex:toyseq}, to establish our
indistinguishability property, namely \emph{frame opacity} defined
below, we  will consider:
%can consider the 
%need to define an idealisation operator that retains part of the shape of
%outputted messages. Typically, assuming that the three outputs are
%labelled with $\ell_1$, $\ell_2$, $\ell_3$ and $\ell_4$ respectively, we will consider:
$$
\ell_1 \mapsto x^\varN_1, \; \ell_2 \mapsto x^\varN_2, \; \ell_3
\mapsto x^\varN_3,
\; \ell_4 \mapsto x^\varN_4
$$
assuming that the four outputs are
labelled with $\ell_1$, $\ell_2$, $\ell_3$, and $\ell_4$ respectively.
\end{example}

\begin{example}
\label{ex:idealization-DAA}
Regarding Example~\ref{ex:DAA}, we also need to define an
idealisation that retains the shape of the second outputted
message. Moreover, the idealisation of the second outputted message
will depend on the nonce previously received. 
Assuming that the outputs are
labelled with $\ell_1$ and $\ell_2$ respectively, we consider:
$\ell_1 \mapsto x^\varN_1, \;\; \ell_2 \mapsto
\zk(\sign(\langle x^\varN_2, x^\varN_3\rangle, \skIssuer),
   x^\varN_2, \tuple(x^\varI_1, x^\varN_4, x^\varN_5,
   \pk(\skIssuer)))$.
Note that such an idealisation would not work for
Example~\ref{ex:DAA-s}; there $\skIssuer$ is an identity-parameter
which, following Definition~\ref{def:idealphi}, cannot occur in $\ideal(\ell_2)$. 
 % Note that such an idealisation would not be conform
  %   w.r.t.~Definition~\ref{def:idealphi}
   %  for Example~\ref{ex:DAA-s}
    % since $\skIssuer$ is an identity-parameter for that different protocol.
     It turns out that frame opacity
     cannot be established (for any heuristics considered by our tool
     UKano) for a good reason: unlinkability fails to hold for
     this variant. Essentially, this is because verifiers send ZK proofs whose
     the public parts contain the public key of their issuers.
\end{example}

The following proposition establishes that
the particular choice of $\fr$ in $\ideaf^\fr(\ta)$
is irrelevant with respect to static equivalence.
We can thus note $\ideaf(\ta)\sim\phi$ instead of there exists
$\fr$ such that $\ideaf^\fr(\ta)\sim\phi$.

\begin{proposition}
  \label{prop:id-represent}
Let $\ideaf^\fr(\ta)$ (resp. $\ideaf^\frr(\ta)$) be the idealised frame
associated to $\ta$ relying on~$\fr$ (resp.~$\frr$). We have
that
$\ideaf^\fr(\ta) \sim \ideaf^\frr(\ta)$.
\end{proposition}

\begin{proof}
It is sufficient to observe that $\ideaf^\fr(\ta)$  and $\ideaf^\frr(\ta)$
are equal up to a bijective renaming of names.
\end{proof}

We can now formalise the notion of frame opacity as announced:
it requires that all reachable frames must be statically equivalent to
idealised frames.

\begin{definition}%[Frame opacity]
\label{def:frame-opacity}
  The protocol $\proto$ ensures \emph{frame opacity} w.r.t.~$\ideamstef$ if
  for any execution 
  $(\pMa;\emptyset) \lrstep{\ta} (Q;\phi)$ we have that 
$\ideaf(\ta)$ is defined and
  $\ideaf(\ta)\sim\phi$.
\end{definition}

There are many ways to choose the idealisation operator $\ideam{\cdot}$.
We present below a {\em syntactical construction} that is sufficient to deal
with almost all our case studies.
This construction has been implemented as a heuristic to automatically build
idealisation operators in the tool \ukano.
The tool \ukano also provides  other heuristics that generally lead to better performance
but are less tight (\ie they cannot always be used to establish frame opacity).
We explain how \ukano verifies frame opacity and compare the different heuristics it can
leverage in \Cref{sec:casestudies}.

At first reading, it is possible to skip the rest of the section and directly go to
\Cref{sec:priv:wa} since proposed canonical constructions
are just instantiations of our generic notion of idealisation.

%%% IDEALIZATION SYNTAXIQUE

\subsubsection{Syntactical idealisation} \label{subsec:syntactical-id}
Intuitively, this construction builds the idealisation operator
by examining the initiator and responder roles as syntactically given in
the protocol definition. The main idea is to consider (syntactical) outputted terms
one by one, and to replace identity parameters, as well as variables bound by a let
construct by pairwise distinct names variables, \ie variables
  in $\X^\varN$.

\begin{definition}
Let $\Pi = (\vect{k}, \nI, \nR, \dag_I,
\dag_R, \ini, \res)$ be a protocol that uses input variables \linebreak[4]
${\{x^\varI_1,
x^\varI_2, \ldots\} \subseteq  \X^\varI}$ (in this order) for its two
roles, and distinct variables from $\X_\Let$ in let constructions. Let 
$\sigma : \vect{k} \cup \vect{n_R} \cup \vect{n_I} \cup \X_\Let \to \X^\varN$
be  an injective renaming.
The syntactical idealisation
operator %associated to~$\Pi$ 
maps any $\ell \in \mathcal{L}$ occurring in an output action
$\ell: \Out(c,u)$ in~$\ini$ or~$\res$ (for some $c$ and some $u$) 
to $u\sigma$.
\end{definition}

\begin{example}
\label{ex:idea-synt-FH} 
Continuing Example~\ref{ex:protocol-FH},
we first perform some renaming to
satisfy the conditions imposed by the previous definition. 
We therefore replace $x_1$ by $x^\varI_1$ in role $\ini$, and $y_1,
y_2$ by $x^\varI_1, x^\varI_2$ in role $\res$. We assume that $x_2,
x_3$, and $y_3$ are elements of~$\X_\Let$. We consider a renaming
$\sigma$ that maps $k, n_I, n_R,  x_2, x_3, y_3$ to $x^\varN_1, \ldots,
x^\varN_6$. We obtain the following idealisation operator:
$$\ell_1 \mapsto x^\varN_2 ;\;\;\;
\ell_2 \mapsto  \senc(\langle x^\varI_1,x^\varN_3\rangle,x^\varN_1); \; \;\;
\ell_3 \mapsto \senc(\langle x^\varN_5,
  x^\varN_2\rangle, x^\varN_1).$$

Considering $\fr$ as defined in
Example~\ref{ex:idealization-FH}, \ie such that 
$\fr(a_I,x^\varN_j) = n^I_j$ and  $\fr(a_R,x^\varN_j) =
n^R_j$, and relying on the idealisation operator defined above, and
$\ta$ as given 
in Example~\ref{ex:annotation}, we have that: 
$\ideaf^\fr(\ta) = \{w_1 \mapsto n^I_2 , \;w_2 \mapsto
\senc(\langle n^I_2, n^R_3\rangle,n^R_1), \;w_3 \mapsto  \senc(\langle
  n^I_5, n^I_2\rangle, n^I_1)\}$.
This idealisation is different
from the one described in~Example~\ref{ex:idealization-FH},  but it also
allows us to establish frame opacity.
\end{example}

% SD: je supprime
% \begin{example}
% \label{ex:idea-synt-toyseq}
% Continuing Example~\ref{ex:toyseq}, we 
% consider a renaming $\sigma$  that maps $k$, $n_R$, $n_T$, $x_2$, $x_3$, $y_3$, $r_I$, $r_R$ to $x^\varN_1, \ldots,
% x^\varN_8$. We obtain the following idealisation operator which is
% also suitable to establish frame opacity:
% $$\ell_1 \mapsto x^\varN_2; \;\;\;
% \ell_2 \mapsto \langle x^\varN_3, \renc(x^\varI_1,x^\varN_1,x^\varN_8)\rangle ;\;\;\;
% \ell_3 \mapsto  \renc(x^\varN_5,x^\varN_1,x^\varN_7)
% $$
% \end{example}

\begin{example}
Continuing Example~\ref{ex:DAA}, we consider a renaming $\sigma$ that
maps: $k_C$, $\id_C$, $n_V$, $n_C$, $m$, $x_2$, $x_3$, $x_4$, $y_3$ to
$x^\varN_1, x^\varN_2,\ldots, x^\varN_9$. We obtain the following idealisation
operator:
$$
\ell_1 \mapsto x^\varN_3;\;\;\;
\ell_2 \mapsto \zk(\sign(\langle x^\varN_1, x^\varN_2\rangle, 
\skIssuer),x^\varN_1,\tuple(x^\varI_1, x^\varN_4, x^\varN_5, \pk(\skIssuer)));\;\;\;
\ell_3 \mapsto \error.
$$

\noindent Such an idealisation operator is also suitable to establish frame
opacity. 
\end{example}

As illustrated by the previous examples, the syntactical
idealisation is sufficient to conclude on most
examples. Actually, using this canonical construction, we
  automatically build the idealisation operator and check frame
  opacity for all the examples we have introduced in the previous
  sections and for most of the case studies presented in
  \Cref{sec:casestudies}.

\subsubsection{Semantical idealisation}
The previous construction is clearly purely syntactic and therefore closely
connected to the way the roles of the protocol are written.
Its main weakness lies in the way variables are bound by let constructions.
Since there is no way to statically
guess the shape of messages that will be instantiated for those variables,
the previous technique replaces them by fresh session names.
False negatives may result from such over-approximations.
We may therefore prefer to build an idealisation operator looking at the
messages outputted during a concrete execution. In such a case, we may
simply retain part of the shape of messages, which a priori does not
reveal anything sensitive to the attacker (\eg pairs, lists). This can
be formalised as follows:

\begin{definition}
\label{def:transparent}
A symbol $\ffun$ (of arity~$n$) in $\Sigma$ is \emph{transparent} if it is a public constructor
symbol that does not occur in $\E$ and such that:
for all $1 \leq i \leq n$, there exists a recipe $R_i \in
\T(\Sigma_\pub,\{w\})$ such that for any message $u =
\ffun(u_1,\ldots, u_n)$, we have that $R_i\{w \mapsto u\} \redc v_i$
for some $v_i$ such that $v_i =_\E u_i$.
\end{definition}

\begin{example}
\label{ex:transparent}
Considering the signature and the equational theory introduced in
Example~\ref{ex:signature} and Example~\ref{ex:xor}, the symbols
$\langle \; \rangle$ and $\ok$ are the only ones that are transparent.
Regarding the pairing operator, the recipes $R_1 = \proj_1(w)$ and
$R_2 = \proj_2(w)$ satisfy the requirements.
\end{example}

Once the set $\Sigma_t$ of transparent functions is fixed, the idealisation associated to a
label $\ell$ occurring in $\Pi$ will be computed relying on a particular 
(but arbitrary)
message $u$ that has been outputted with this label $\ell$ during a
concrete execution of $\mathcal{M}_\Pi$. The main idea is to go
through transparent functions until getting stuck, and then replacing
the remaining sub-terms using distinct name variables from $\X^\varN$.

\begin{example}
\label{ex:frame-opa:semantical-ieali}
Considering the protocol given in Example~\ref{ex:protocol-FH}, 
the resulting
idealisation associated to~$\Pi$ (considering messages 
in~$\phi_0$ as defined in Example~\ref{ex:execution}) is:
$
\ell_1 \mapsto x^\varN_1;\;
\ell_2 \mapsto x^\varN_2;\;
\ell_3 \mapsto x^\varN_3$.
%
%The protocol given in Example~\ref{ex:protocol-variant-FH} will give
%us:
%$
%\ell_1 \mapsto x^\varN_1;\;
%\ell_2 \mapsto \langle x^\varN_2, x^\varN_3\rangle;\;
%\ell_3 \mapsto x^\varN_4
%$.
Even if this idealisation operators is quite different from the one
presented in Example~\ref{ex:idea-synt-FH}. it
is also suitable
to
establish frame opacity.
\end{example}

In~\cite{HBD-sp16}, the idealisation operator associated to $\Pi$ was
exclusively computed using this method. The technique is
implemented in the tool \ukano, and yields simple idealisations
for which frame opacity often holds and can be established
quickly.
However, it happens to be insufficient to
establish frame opacity in presence of function symbols that are
neither
transparent nor totally opaque such as signatures. Indeed, a signature
function symbol is not transparent according to our definition: an
attacker can make the difference between 
a signature $\sign(m,sk(A))$ and a random nonce. Therefore, replacing
such a term by a fresh session name will never allow one to establish frame opacity.
  That is why we also defined other types of idealisations that produce
  more complex idealised messages but allow for a much better level of
  precision.
In practice, our tool UKano
  has three different built-in heuristics for computing idealisations
  which span the range between precision (syntactical idealisation)
  and efficiency (semantical idealisation).

%%
%% WELL-AUTHENTICATION
%%
\subsection{Well-authentication}
\label{sec:priv:wa}
Our second condition will prevent the attacker from obtaining some information
about agents through the outcome of conditionals. To do so, we will
essentially require that conditionals of $\ini$ and $\res$ can only be
executed successfully in honest, intended interactions.
However, it is unnecessary
to impose such a condition on conditionals 
that never leak any information, which are found in several security 
protocols. We characterise below a simple class of such conditionals, for 
which the attacker will always know the outcome of the conditional based on 
the past interaction.

\begin{definition}
  \label{def:safe}
 For a protocol $\Pi$,
  a conditional
  $\Let \; \vect z = \vect t \;\In \; P \; \Else \; Q$
  occurring in $\agent\in\{\ini,\res\}$ is \emph{safe} if
  $\vect t \in
  \T(\Sigma_\pub, \{x_1,\ldots,x_n\}\cup\{u_1,\ldots,u_m\})$,
  where the $x_i$ are the variables bound by the previous inputs
  of that role, and $u_i$ are the messages used in the previous
  outputs of that role.
\end{definition}

\begin{example}
Consider the process 
$\Out(c,u).\In(c,x).\Let \; z = \mathsf{neq}(x,u)\; \In\;
P\;\Else\;Q$.
The conditional is used to ensure that the
agent will not accept as input the message he sent at the previous
step. Such a conditional is safe according to our definition. 
\end{example}

Note that trivial conditionals required by the grammar of protocols
(\Cref{def:priv:role:grammar})
are safe and will thus not get in the way of our analysis.
We can now formalise the notion of association, which expresses that
two agents are having an honest, intended interaction, \ie the attacker
essentially did not interfere in their communications.
For an annotated trace $\ta$ and annotations $a$ and $a'$,
  we denote by $\restrict{\ta}{a, a'}$
  the subsequence of $\ta$ that consists of actions of the form
  $\alpha[a]$ or $\alpha[a']$.

\begin{definition}
    Given a protocol $\Pi$,  
    two annotations $a_1=\aagent_1(\vect k_1,\vect n_1)$ and $a_2=\aagent_2(\vect k_2,\vect n_2)$ are {\em associated} in $(\ta,\phi)$ if:
  \begin{itemize}
  \item they are \emph{dual}, \ie $\aagent_1\neq\aagent_2$, and
    $\vect{k_1}=\vect{k_2}$ when
    $\fn(\res)\cap\fn(\ini)\neq\emptyset$ (the shared case);
  \item the interaction  $\restrict{\ta}{a_1,a_2}$ is honest for
    $\phi$ (see Definition~\ref{def:honest}).
  \end{itemize}
\end{definition}

\begin{example} Continuing Example~\ref{ex:annotation}, $\aini(k',n'_I)$ and
  $\ares(k',n'_R)$ are associated in $(\ta,\phi_0)$.
\end{example}

Finally, we can state our second condition.

\begin{definition}
  \label{condi:auth}  
  The protocol $\proto$ is \emph{well-authenticating} if,
for any 
$ 
  (\pMa;\emptyset)
  \lrstep{\ta.\taut\annot{a}}
  (\p;\phi)
$,
  either the last action corresponds to a safe conditional of $\ini$ or $\res$, or
  there exists $a'$ such that:
  \begin{enumerate}
    \item[(i)]
      The annotations $a$ and $a'$ are associated in $(\ta, \phi)$;
    \item[(ii)] Moreover, when $\fn(\res)\cap\fn(\ini)\neq\emptyset$
      (the shared case), $a'$ (resp. $a$) is only associated with~$a$ (resp.~$a'$) in $(\ta, \phi)$.
  \end{enumerate}
\end{definition}

Intuitively, this condition does not require anything for safe conditionals as we already
know that they cannot leak new information to the attacker (he already knows their outcome).
For unsafe conditionals, condition \emph{(i)} requires that whenever an agent~$a$ evaluates them positively 
(\ie he does not abort the protocol), it must be the case that this agent~$a$
is so far having an honest interaction with a dual agent~$a'$.
Indeed, as discussed in introduction,
it is crucial to avoid such unsafe conditionals to be evaluated positively when the attacker
is interfering because this could leak crucial information.
In the rest of the paper, when considering a protocol $\Pi$, we will say
  that a conditional in a process resulting from $\pMa$ or $\pS$ is {\em safe} when
  it corresponds to a safe conditional in $\ini$ or $\res$.

As illustrated in the following example, condition \emph{(ii)} is needed to prevent from having
executions where an annotation is associated to several annotations, which would
break unlinkability in the shared case (\ie when $\fn(\res)\cap\fn(\ini)\neq\emptyset$).

\begin{example}
\label{ex:uniquely-associated}
We consider a protocol between an initiator and a responder that share
a symmetric key~$k$. The protocol can be described informally as follows:\\[1mm]
\null\hfill
$\begin{array}{rrl}
1.& I \to R: & \{n_I\}_k\\
2.& R \to I: & n_R
\end{array}$\hfill\null

\smallskip{}

Assuming that the two outputs are labelled with $\ell_1$ and $\ell_2$
respectively, the idealisation operator $\ell_1 \mapsto
x^\varN_1$, $\ell_2 \mapsto x^\varN_2$ is suitable to establish frame
opacity.
We may note that the only conditional is the one performed by the
responder role when receiving the ciphertext. He will check whether it
is indeed an encryption with the expected key $k$.
When an action $\taut[R(k, n_R)]$ occurs, it means that a
ciphertext encrypted with $k$ has been received by $R(k,n_R)$ and
since the key $k$ is unknown by the attacker, such a ciphertext has
been sent by a participant: this is necessarily a participant
executing the initiator role with key $k$. Hence condition $(i)$ of
well-authentication holds (and can actually be formally
proved). However, condition $(ii)$ fails to hold since two responder roles
may accept a same ciphertext $\{n_I\}_k$ and therefore be associated
to the same agent acting as an initiator. This corresponds to an attack
scenario w.r.t. our formal definition of unlinkability since such a
trace will have no counterpart in $\mathcal{S}_\Pi$. More formally,
the trace $\tr =
\Out(c_I,w_0).\In(c_R,w_0).\taut.\Out(c_R,w_1).\In(c_R,w_0).\taut.\Out(c_R,w_2)$
will be executable starting from $\mathcal{M}_\Pi$ and will allow one
to reach
$\phi = \{w_0 \mapsto \senc(n_I,k); \; w_1 \mapsto n_R; \; w_2 \mapsto n'_R\}$.
Starting from $\mathcal{S}_\Pi$ the second action $\taut$ will
not be possible, and more importantly this will prevent the observable
action $\Out(c_R,w_2)$ to be triggered.
\end{example}

While the condition (i) of well-authentication is verifiable
quite easily by expressing it as simple reachability properties (as explained in Section~\ref{sec:mecha-wa}),
the required condition (ii) for the shared-case is actually
harder to express in existing tools.
We therefore shall prove that, for the shared case,
once condition (i) of well-authentication is known to 
hold, condition (ii) is a consequence of two simpler conditions
that are easier to verify (as shown in Section~\ref{sec:mecha-wa-ii}).
First, the first conditional of the responder role should be safe
--- remark that if this does not hold, similar attacks as the one discussed above
may break unlinkability.
Second, messages labelled by some~$\ell$
outputted in honest interactions by different agents should
always be different.

\begin{lemma} \label{lem:wa-i}
  Let
  $\Pi = (\vect{k}, \vect n_I, \vect n_R, \dagger_I, \dagger_R, \ini, \res)$
  be a protocol such that  $\fn(\ini)\cap\fn(\res)\neq\emptyset$
  (shared case)
  that satisfies condition {(i)} of well-authentication. Then
  well-authentication holds
%  Condition (ii) of well-authentication holds 
provided that:
  \begin{itemize}
    \item[(a)] the first conditional that occurs in $\res$ is safe;
    \item[(b)] for any execution $(\pMa;\emptyset)\sint{\ta}(\p;\phi)$,
      if  $\ta_1 = \restrict{\ta}{a_1,b_1}$
      and $\ta_2 = \restrict{\ta}{a_2,b_2}$ are
      honest with $a_1 \neq a_2$
      then for any 
      $\ell:\Out(c,w_1)\annot{a_1}\in\ta_1$
      and
      $\ell:\Out(c,w_2)\annot{a_2}\in\ta_2$
      then $\phi(w_1)\not =_\mathsf{E} \phi(w_2)$.
% aany message outputted by $a_1$ in $\ta_1$
%       is different (modulo $\E$)
%       from any message outputted by $a_2$ in $\ta_2$.
  \end{itemize}
\end{lemma}

\begin{proof}
Consider an execution
$\pMa\lrstep{\ta.\taut\annot{a'}}(\p;\phi)$
where two agents~$a$ and $a'$
are associated and~$a'$ has performed the last {$\taut$}.
If this test corresponds to  a safe conditional, there is nothing to
prove. Otherwise, 
we shall prove that $a$ is only associated to~$a'$,
and vice versa.

\smallskip{}

\noindent \emph{Agent $a'$ is only associated to $a$.}
  Consider the last input of $a'$ (the one just before
  $\taut\annot{a'}$) and the output of~$a$ that occurs before
  this input of $a'$:
  $$\pMa\lrstep{
    \ta.\Out(c,w_\ell)\annot{a}.
      \ta'.\In(c',R)\annot{a'}.\ta''.\taut\annot{a'}
    }(\p;\phi)$$
  We have $R\phi\redc \phi(w_\ell)$ where $w_\ell$ is labelled $\ell$.
  Assume, for the sake of contradiction,
  that $a'$ is associated to another agent~$b \neq a$.
  Then, we have $R\phi\redc =_\E \phi(w_{\ell}')$ for some handle,
  and thus thanks to Item~\ref{def-cr-context-2} of Definition~\ref{def:computation-rel}, we have that
  $\phi(w_\ell) =_\E \phi(w_{\ell}')$,
  for a handle~$w_{\ell}'$ corresponding to some output of $b$ labelled $\ell$
  in the honest trace $\restrict{\ta}{a',b}$.
  This contradicts assumption~$(b)$.

\smallskip{}

\noindent\emph{Agent $a$ is only associated to $a'$.}
  Agent $a$ must have performed an input in $\ta$:
  this is obvious if~$a$ is a responder, and follows from assumption $(a)$
  otherwise.
Let $\ell: \Out(c,w_\ell)\annot{a'}$ be the output label (with
annotation $a'$) occurring in
$\ta$ just before the input of $a$ mentioned above.
  The considered execution is thus of the following form:
  $$\pMa\lrstep{
      \ta.\Out(c,w_\ell)\annot{a'}.
      \ta'.\In(c,R)\annot{a}.\ta''.\taut\annot{a'}
    }(\p;\phi)$$
  We know that the message $m$, satisfying $R\phi\redc m$, which is inputted by $a$ is equal
  (modulo $\E$)
  to the previous output of~$a'$, that is $\phi(w_\ell)$. As for the previous case,
  condition (b) implies that it cannot be equal to the output of another
  agent having an honest interaction in $\ta$,
  thus $a$ is only associated to~$a'$.
\end{proof}

%%
%% OUR MAIN RESULT
%%

\subsection{Main result}
\label{subsec:result}

Our main theorem establishes that the previous two conditions are
sufficient to ensure unlinkability and anonymity.

\begin{theorem}
\label{theo:main}
Consider a protocol $\Pi = (\vect{k}, \vect n_I, \vect n_R, \dag_I,
\dag_R, \ini, \res)$ and some identity names $\vect{id}
\subseteq \vect k$. If the protocol ensures both well-authentication
and frame opacity w.r.t.\ $\vect{id}$, then $\Pi$ ensures unlinkability and anonymity
w.r.t.\ $\vect{id}$.
\end{theorem}

Note that, when $\vect{id} = \emptyset$, we have that
$\mathcal{M}_{\Pi, \vect{id}} \approx \mathcal{M}_{\Pi}$ and our two
conditions coincide on $\mathcal{M}_{\Pi, \vect{id}}$ and
$\mathcal{M}_{\Pi}$. We thus have as a corollary that if 
$\mathcal{M}_{\Pi}$ ensures well-authentication and frame opacity,
then $\Pi$ is unlinkable.

The proof of this theorem is detailed in Appendix~A, and we
explain in Section~\ref{sec:mechanization} how to check these two conditions in
practice relying on existing verification tools.
We apply our method on various case studies that are detailed 
 in Section~\ref{sec:casestudies}. Below, we only briefly
summarize the result of the confrontation of our method to our various
running examples, focusing on unlinkability.

\begin{figure}[h]
  \centering
  
  \begin{tabular}{l|lll}
 \multirow{2}{*}{Protocol \rule[-5mm]{0pt}{1cm} }     & Frame% Frame Opacity
                  & Well- % Well-authentication &
                  &   \multirow{2}{*}{Unlinkability \rule[-5mm]{0pt}{1cm} } \\
& opacity & \; authentication \;& \\
 \hline
        Feldhofer (Example~\ref{ex:protocol-FH}) & \verif & \verif & \holds \\
%        Feldhofer variant with $\rep$ (Example~\ref{ex:att-var-FH})& \verif & \nope & \attaque \\
%      Feldhofer variant with $\rec$
 %   (Example~\ref{ex:protocol-variant-FH}) & \verif & \verif (with \tamarin) & \holds \\
        Toy protocol with $\rep$ (Example~\ref{ex:toyseq-att})& \verif & \nope & \attaque \\
      Toy protocol with $\rec$
    (Example~\ref{ex:toyseq}) & \verif & \verif (with \tamarin) & \holds \\
DAA-like with one issuer (Example~\ref{ex:DAA})& \verif & \verif & \holds\\
DAA-like with many issuers (Example~\ref{ex:DAA-s})& \nope & \verif & \attaque
  \end{tabular}
  \caption{Summary of our running examples.
         }
  \label{fig:summary-running}
\end{figure}

We note  \verif~for a condition automatically checked using \ukano~and \nope~when the condition does not hold.
For the analysis of the Toy protocol with $\rec$, we do not rely on
\ukano (which is based on \proverif)  since \proverif does not support
the $\rec$ operator. We establish the well-authentication property
using \tamarin
and frame opacity using \proverif
by allowing sessions to run concurrently 
and thus doing a sound over-approximation of the protocol's behaviors.
%Except for the latter example, all positive results were automatically
%established using our tool \ukano.
%\marginpar{\stef{est-ce vrai pour Toy?}}
Frame opacity has been established relying on the syntactical idealisation as well as the semantical one,
except for \Cref{ex:DAA}. Indeed, as explained at the end of
Section~\ref{subsec:frame-opacity}, the
semantical idealisation is not suitable in this case.

\newcommand{\fst}[1]{\mathsf{fst}(#1)}
\newcommand{\snd}[1]{\mathsf{snd}(#1)}
\newcommand{\bi}{P}
\newcommand{\red}{\rightarrow}

\newcommand{\att}[1]{\mathsf{att}(#1)}
\newcommand{\biatt}[1]{\mathsf{att}'(#1)}

\newcommand{\ideap}[1]{P^\mathsf{ideal}}

\section{Mechanization}
\label{sec:mechanization}

We now discuss how to verify unlinkability and anonymity in practice,
through the verification of our two conditions.
More specifically, we describe how appropriate encodings allow one to
verify frame opacity (\Cref{sec:mecha-fo}) and
well-authentication~(\Cref{sec:mecha-wa}),
respectively through diff-equivalence and correspondence properties
in \proverif.

We additionally provide a tool, called \ukano~\cite{depotANO} (Section~\ref{sec:mecha-ukano}), which
mechanises the encodings described in this section.
Our tool takes as input a specification of a protocol in our class,
computes encodings, and calls \proverif to automatically check our two 
conditions, and thus unlinkability and anonymity.
As briefly mentioned in Section~\ref{sec:approach} and detailed in \Cref{sec:casestudies},
\ukano concludes on many interesting case studies.

%%
%% Frame Opacity
%%

\subsection{Frame opacity}
\label{sec:mecha-fo}

  We shall describe how to encode frame opacity using
  the diff-equivalence of \proverif~\cite{BlanchetAbadiFournetJLAP08}.
  In a nutshell, we will check this strong notion of equivalence
  between $\pMa$ and
  a modified version of it that produces idealised outputs
  instead of real ones, in order to check static equivalence between
  all pairs of frames of the form
  $(\Phi,\ideaf(\ta))$ where $\ta$ is executable by $\pMa$ and $\Phi$
  is the resulting frame.
  The main issue in implementing this idea arises from diff-equivalence being 
  too
  strong regarding tests and computations of idealized terms.
  In~\cite{HBD-sp16}, we had proposed a solution that avoided this
  problem by largely over-approximating the set of executable traces,
  which is sound but very imprecise.
  Moreover, this first solution is only adequate for the notion of
  idealization considered in \cite{HBD-sp16}, and not for the generalization
  proposed in the present paper.
  We describe below a simpler solution, that is much more precise and 
  efficient, and can accommodate our generalized notion of idealization,
  at the cost of a slight extension of \proverif's diff-equivalence.
We shall start with a brief reminder on diff-equivalence in \proverif,
in order to describe how we extend it, before showing how this extension
allows us to naturally encode frame opacity.

\paragraph{Diff-equivalence.}
Intuitively, diff-equivalence is obtained from trace equivalence by forcing
the two processes (or configurations) being compared to follow the same execution.
It has been introduced in \cite{BlanchetAbadiFournetJLAP08} as a means
to automatically verify observational equivalence in \proverif.
This paper deals with the full process algebra supported by \proverif, which
is more general but compatible with the process algebra of the present paper,
the main difference being that we do not account for private channels.
Processes of \cite{BlanchetAbadiFournetJLAP08} are equipped with a reduction
semantics, noted $P \red P'$. This allows to define observational equivalence,
which implies trace equivalence in our sense.
The key notion in \cite{BlanchetAbadiFournetJLAP08} is that of a 
\emph{bi-process}, that is a process in which some terms are replaced by 
\emph{bi-terms} of the form $\choice{u_1}{u_2}$. Given a bi-process $\bi$,
its first \emph{projection} $\fst{\bi}$ is defined by taking the first 
component of all choice operators occurring in it. The second projection
$\snd{\bi}$ is defined analogously. Bi-processes are given a reduction 
semantics by taking the same rules as those for processes (which we do not
recall) with three modified rules\footnote{
  We use our own notations here, taking advantage of the fact that
  our calculus is a simplification of the one used 
  in~\cite{BlanchetAbadiFournetJLAP08}:
  in particular, channels are public constants and choice operators cannot
  be used in channel positions.
  We also use the notation $\choice{\cdot}{\cdot}$ as in the tool \proverif,
  rather than $\mathsf{diff}[\cdot,\cdot]$ as in 
  \cite{BlanchetAbadiFournetJLAP08}.
}:
$$ \begin{array}{llr}
  \Out(c,u).Q \mid \In(c,x).P
  \;\red\; Q \mid P\{x\mapsto u\}
  & & (\text{Red I/O})
  \\
  \Let\; \vect{x} = \vect{t} \;\In\; P \;\Else\; Q
  \;\red\;
  P\{\vect{x}\mapsto\choice{\vect u_1}{\vect u_2}\}
  & \text{ if } \fst{\vect{t}}\redc \vect u_1
  \text{ and } \snd{\vect{t}}\redc\vect u_2
  & \quad (\text{Red Fun 1})
  \\
  \Let\; \vect{x} = \vect{t} \;\In\; P \;\Else\; Q
  \;\red\;
  Q
  & \text{ if } \fst{\vect{t}}\not\redc
  \text{ and } \snd{\vect{t}}\not\redc
  & \quad (\text{Red Fun 2})
\end{array} $$
A bi-process execution step thus consists of strongly synchronized execution 
steps of its two projections. In particular, a conditional in a bi-process
succeeds (resp.\ fails) if it succeeds (resp.\ fails) for both of its projections.
Formally, we have that $\bi\red\bi'$ implies $\fst{\bi}\red\fst{\bi'}$ and
$\snd{\bi}\red\snd{\bi'}$.
However, it could be that some execution step of $\fst{\bi}$
(resp.\ $\snd{\bi}$) cannot be obtained in this way from an execution step
of $\bi$. In fact,
\cite[Theorem~1]{BlanchetAbadiFournetJLAP08} shows that the two projections
of a bi-process $\bi$ are observationally equivalent if, for all $C$,
all reducts of $C[\bi]$ are uniform in the following sense:
\begin{definition}[\cite{BlanchetAbadiFournetJLAP08}]
  A bi-process $\bi$ is \emph{uniform} if for all reductions
  $\fst{\bi}\red P_1$, there exists $\bi'$ such that $\bi \red \bi'$ and 
  $\fst{\bi'}=P_1$, and symmetrically for $\snd{\bi}$.
\end{definition}
From now on, we say that a bi-process is \emph{diff-equivalent} when it 
satisfies the condition of \cite[Theorem~1]{BlanchetAbadiFournetJLAP08}.
By extension, we say that the two projections of a bi-process are 
diff-equivalent when the bi-process is.

\paragraph{Diff-equivalence verification in \proverif.}
The next contribution of \cite{BlanchetAbadiFournetJLAP08} is to show
that diff-equivalence can be automatically verified in \proverif by adapting 
its Horn clause encoding and resolution algorithm to bi-processes.
We will not recall the encoding and its modifications in detail here, but
only present some key ideas at a high level.
In the single-process case, a unary predicate $\att{\cdot}$ is used
to encode that the attacker knows some message.
The attacker's capabilities are expressed as Horn clauses involving
this predicate, \eg the ability to encrypt is translated as
$$\forall x \forall y.~ \att{x} \wedge \att{y} \Rightarrow \att{\enc{x}{y}}.$$
Then, a process fragment of the form
$\In(c,x). \Let\; y = \dec{x}{k} \;\In\; \Out(c',t)$,
where $t$ is a constructor term with free variable $y$,
is encoded as
$$\forall y.~ \att{\enc{y}{k}} \Rightarrow \att{t}.$$
Note that the variable $x$ does not appear in the clause,
but has been refined into $\enc{y}{k}$ as part of the translation.

Consider now the analogue bi-process fragment
$$\In(c,x). \Let\; y = \dec{x}{k} \;\In\; \Out(c',\choice{t_1}{t_2}).$$
It corresponds to two processes, each of which may receive a message and
attempt to decrypt it using $k$. Upon success, the processes output $t_1$
and $t_2$ respectively,
relying on the value $y$ obtained from the respective decryptions.
In \proverif, this bi-process would translate into the Horn clause
$$\forall y_1 \forall y_2.~
\biatt{\enc{y_1}{k},\enc{y_2}{k}} \Rightarrow \biatt{t'_1,t'_2}$$
{where $t'_1$ (resp.~$t'_2$) is $t_1$ (resp.~$t_2$) in which
all the occurrences of~$y$ have been replaced by $y_1$ (resp.~$y_2$).}
This time, a binary predicate $\biatt{\cdot,\cdot}$ is used to encode the 
attacker's knowledge on each side of the bi-process run:
the clause roughly says that if, at some point of the execution
of the bi-process, the attacker can deduce (using the same
derivation) a term of the form $\enc{y_1}{k}$ from the left frame
and a term $\enc{y_2}{k}$ from the right frame,
then he will learn $t_1$ on the left and~$t_2$ on the right.
The attacker's capabilities are also modified to encode the effect of the
attacker's capabilities on each side of the bi-process run, \eg for encryption:
$$\forall x_1 \forall x_2 \forall y_1 \forall y_2.~
\biatt{x_1,x_2} \wedge \biatt{y_1,y_2} \Rightarrow 
\biatt{\enc{x_1}{y_1},\enc{x_2}{y_2}}.$$

\paragraph{An extension of bi-processes.}
In the original notion of bi-processes~\cite{BlanchetAbadiFournetJLAP08},
the two sides of a bi-process are isolated and can execute independently.
However, the Horn clause encoding of bi-processes that is used for 
verification in \proverif makes it easy to lift this restriction in a way
that enables interesting new applications of diff-equivalence.
Specifically, we introduce the possibility of binding two variables at once
in a bi-process input, which we write $\In(c,\choice{x_1}{x_2}).P$.
For simplicity, we can consider that all inputs feature such choice variables,
as the usual form $\In(c,x).P$ can be replaced by 
$\In(c,\choice{x_1}{x_2}).P\{x\mapsto\choice{x_1}{x_2}\}$.
The intuitive semantics of such a construct is that $x_1$ is bound to the 
message received on the left side of the bi-process run, while $x_2$ is bound
to the message received on the right.
Formally, we change the (Red I/O) rule as follows:
$$ \Out(c,u).Q \mid \In(c,\choice{x_1}{x_2}).P
\quad\red\quad Q \mid P\{x_1\mapsto\mathsf{fst}(u), x_2\mapsto\mathsf{snd}(u)\}.$$
Crucially, each side of the bi-process will then have access to both $x_1$ and 
$x_2$, allowing a form of communication between the two sides.

With this modification, \cite[Theorem~1]{BlanchetAbadiFournetJLAP08} does
not hold anymore. In fact, $\fst{\bi}$ and $\snd{\bi}$ may not be 
ground
processes when $\bi$ uses choice variables in inputs, so that it does not even
make sense to compare them for observational equivalence. More generally, we cannot talk 
in general of the projection of a bi-process execution: the fact that
$\bi\red\bi'$ implies $\fst{\bi}\red\fst{\bi'}$ becomes not only false but
also ill-defined in general.
However, the notion of uniformity is still meaningful, if properly adapted to 
be mathematically well-defined: a bi-process $\bi$ is uniform if, \emph{whenever
$\fst{\bi}$ is ground}, it is the case that for all reductions $\fst{\bi}\red 
P_1$ there exists $\bi'$ such that $\bi\red\bi'$ and $\fst{\bi'}=P_1$, and
symmetrically for $\snd{\bi}$.
Furthermore, we shall see that it can be useful in cases where at least one 
projection of the executions of a bi-process is well-defined, as will be the 
case with our encoding of the notion of frame opacity through (extended) 
diff-equivalence.

Besides the problem of the new meaning of diff-equivalence, an 
important question is whether extended diff-equivalence can be verified 
automatically, and how. We claim that it is straightforward, as the Horn 
clause encoding of bi-processes already features what is needed for adequately 
encoding extended bi-processes, namely the duplicated input variables seen in 
the above example --- accordingly, we only had to modify a few tenth of lines of 
the \proverif tool to implement our extension. A formal justification of this 
claim would require to adapt the long technical development of 
\cite{BlanchetAbadiFournetJLAP08}, and is thus out of the scope of the present 
paper. We simply illustrate the idea here by getting back to our running example illustrating the various
Horn clause encodings, considering now the process fragment
$$\In(c,\choice{x_1}{x_2}).
\Let\; y = \dec{x_1}{k} \;\In\; \Out(c',\choice{t_1}{t_2})$$
where $t_1$ and $t_2$ are constructor terms with free variable $y$.
This would be encoded as
$$\forall {y_1} \forall x_2.~
\biatt{\enc{{y_1}}{k},x_2} \Rightarrow \biatt{t'_1,t'_2}$$
{where $t'_1$ (resp.~$t'_2$) is $t_1$ (resp. $t_2$) in which all the 
occurrences of~$y$ have been replaced by~$y_1$.}
This time, $x_2$ is not refined since the bi-process does not attempt
to deconstruct it, on either side. The clause expresses that if
the attacker can derive a term $\enc{y_1}{k}$ on the left (regardless of
what would be the corresponding term on the right) then he will learn
$t'_1$ on the left and $t'_2$ on the right.

\paragraph{Encoding frame opacity through extended diff-equivalence.}
Using this extended notion of bi-process, we can now directly express
frame opacity as the diff-equivalence of a bi-process. This bi-process
will have $\pMa$ as its first projection. Its second projection
should replace each message output with its idealization, so that
diff-equivalence of $\ideap{P}$ implies
$\Phi\sim\ideaf(\ta)$ for any $\ta$ that is executable by $\pMa$
with $\Phi$ as the resulting frame. In itself, this can be achieved
easily, as it suffices to create new names to use as values for
$\X^n$ variables and use appropriate input variables for the~$\X^i$ variables.
The difficulty lies with tests (and computation failures) which need
to be carefully used to obtain a correct encoding of frame opacity.

First, we need to ensure that
any reduction of $\pMa$ (in any context) can be obtained as the projection 
of a reduction of $\ideap{P}$ (in the same context).
In other words, the second projection of any test should agree with its first
projection, which corresponds to a normal execution of $\pMa$.
This would not hold in general if we performed the same test on the
idealizations which are computed in the second projection. We obtain the desired
behaviour using our extension of bi-processes, by performing
tests using only left-hand side input variables.

Second, we need to ensure that computation failures that occur while
computing idealizations result in a diff-equivalence failure. This is necessary
to obtain a match with frame opacity, which requires that idealizations are
well-defined even when destructors are involved.
Hence, when the bi-process computes idealized values in its second projection,
a failsafe computation should happen in the first projection.
Moreover, idealizations should be computed using the values of the input 
variables from the right side of the bi-process execution, in line with
the definition of idealization, \ie Definition~\ref{def:idealphi}.

\begin{figure}[t]
\footnotesize
\begin{minipage}{0.5\textwidth}
% TODO language Proverif
\begin{lstlisting}
let I (k:bitstring) =
  new nI:bitstring;
  out(ci, nI);
  in(ci, x:bitstring);
  let (=nI, xnr:bitstring) = dec(x, k) in
  out(ci, enc((xnr,nI),k)).   
\end{lstlisting}
\end{minipage}
 \begin{minipage}{0.45\textwidth}
 \begin{lstlisting}
let R (k:bitstring) =
  new nR:bitstring;
  in(cr, ynI:bitstring);
  out(cr, enc((ynI, nR), k));
  in(cr, y:bitstring);
  let (=nR,=ynI) = dec(y,k) in
  out(cr, ok).
 \end{lstlisting}
 \end{minipage}
 \begin{lstlisting}
let FH = ! new k:bitstring; ! (I(k) | R(k)).
 \end{lstlisting}
 \caption{Our running example (Feldhofer) using \proverif's syntax}
\label{fig:syntax-proverif}
\end{figure}

Before proving more formally that our translation is adequate,
let us illustrate it on our running example, \ie
Example~\ref{ex:protocol-FH}.
We first give in Figure~\ref{fig:syntax-proverif} the description
of the protocol in \proverif syntax. This syntax is actually very 
close to the one we introduced in Section~\ref{sec:protocol}. The main
difference is the fact that \proverif relies on types, and here any
name or variable is given the generic type bitstring. Next, we show
in Figure~\ref{fig:check-fo} the bi-process expressing frame opacity
as described above, using the syntaxic idealisation
(Section~\ref{subsec:syntactical-id}).
Note that, when decrypting the first input of the initiator role,
the variable \texttt{x} is used, corresponding to the left side of
the bi-process execution.
The variable \texttt{xid}, correspond to the right (idealized) side,
is not used in that case because this input is not used in idealizations.
However, the variable \texttt{ynIid} corresponding to an idealized
input in the responder role is used in the first output.
In this example, the idealisation operator does not contain any destructor,
hence the computation of the idealisation can never fail. If destructors
were present, they would be computed using $\Let$ constructs inside
the right component of the $\choice{\cdot}{\cdot}$ operator\footnote{
  This is not possible with the theoretical notion of bi-process,
  where $\choice{\cdot}{\cdot}$ operators can only contain terms, which cannot
  contain $\Let$ constructs. However, it is available in (vanilla)
  \proverif as syntactic sugar: it is equivalent to performing the
  $\Let$ outside the $\choice{\cdot}{\cdot}$ with a dummy first projection,
  which is exactly what we need.
} in output, so that
their failure would result in a non-equivalence.

We now conclude with a formal correctness argument.

\begin{proposition}
  Let $\ideamstef$ be an idealisation operator, and note $\ideap{P}$ the 
  corresponding encoding of a process $P$ into a bi-process.
  Assume that, for all $C$ and $B$ such that $C[\ideap{P}]\red^* B$, $B$
  is uniform. Then $P$ satisfies frame opacity wrt.\ $\ideamstef$.
\end{proposition}
\begin{proof}
  Assume, by contradiction, that there exists an execution
  $(P;\emptyset)\lrstep{\ta} (Q;\Phi)$ with $\Phi \not\sim \ideaf(\ta)$.
  Then there exists a test
  $T = \bigl( \Let\; \vect{y} = \vect{u} \;\In\; \Out(o,\ok) \bigr)$
  with $\fv(T)=\fv(\vect{u})\subseteq\dom(\Phi)$ such that
  $T\Phi$ can perform an output (after the successful evaluation of its
  $\Let$) while $T\ideaf(\ta)$ cannot.
  Further, we can construct in a standard way\footnote{
    Define $C_\ta[\bullet]=(\bullet \mid T_\ta)$
    with $T_\epsilon = T$, $T_{\tau.\ta} = T_\ta$,
    $T_{\Out(c,w).\ta} = \In(c,w).T_\ta$ and
    $T_{\In(c,R).\ta} = \Out(c,R).T_\ta$.
    In full details, the last case should include the creation of names
    that are used in $R$ but not previously in the trace.
  } from $\ta$ a context $C_\ta$
  such that $C_\ta[P] \red^* T\Phi$, with only
  communications and internal reductions of $P$ in that reduction.
  Because of this, the reduction can be lifted to the idealized bi-process
  as follows, by definition of $\ideap{P}$:
  $$C_\ta[\ideap{P}]\red^* 
  T'=T\{w\mapsto\choice{\Phi(w)}{\ideaf(\ta)(w)}\}_{w\in\dom(\Phi)}$$
  We have obtained our contradiction, since $T'$ is not uniform:
  indeed, the reduction step $\fst{T'}\red\Out(c,\ok)$ cannot be obtained
  as a projection of a reduction of $T'$, which in fact cannot perform
  any reduction at all.
\end{proof}

\begin{figure}[t]
\footnotesize
\begin{lstlisting}
let SYSTEM =
( !
  new k : bitstring;
     !
      ((
        new nI: bitstring;
        out(ci, nI);
        in(ci, choice[x,xid]: bitstring);
        let ((=nI,xnr: bitstring)) = dec(x,k) in
        new hole__xnr_I_0: bitstring;
        new hole__k_I_1: bitstring;
        out(ci, choice[enc((xnr,nI),k),enc((hole__xnr_I_0,nI),hole__k_I_1)])
      )|(
        in(cr, choice[ynI,ynIid]: bitstring);
        new nR: bitstring;
        new hole__k_R_2: bitstring;
        out(cr, choice[enc((ynI,nR),k),enc((ynIid,nR),hole__k_R_2)]);
        in(cr, choice[y,yid]: bitstring);
        let ((=nR,=ynI)) = dec(y,k) in
        out(cr, ok)
      ))
).
\end{lstlisting}
\caption{\proverif file checking frame opacity generated by \ukano (Feldhofer)}
\label{fig:check-fo}
\end{figure}

\paragraph{Practical application.}
Our tool \ukano automatically constructs the bi-process described
above from a description of the protocol, and calls the extension of
\proverif  in order to check frame opacity. Until this extension is 
integrated in the next release of \proverif, the source files of this
slight extension of \proverif are distributed with \ukano~\cite{depotANO}.
Out tool does not require the user to input the idealisation function.
Instead, a default idealisation is extracted from the protocol's outputs.
The user is informed about this idealisation, and if he wants to,
he can bypass it using annotations or choose another heuristic
to build idealisation operators. In practice, this is rarely necessary;
we provide more details about this in \Cref{sec:mecha-ukano}.
Also note that, although \proverif does not support the repetition
operator $\rec$, we can over-approximate the behaviours of protocols using
it by replacing occurrences
of~$\rec$ with~$\rep$ before checking frame opacity.

%%
%% WELL-AUTHENTICATION
%%

\subsection{Well-authentication} 
\label{sec:mecha-wa}

We explain below how to check condition \emph{(i)} of
well-authentication (see Definition~\ref{condi:auth}).
Once that condition is established, together with frame opacity,
we shall see that condition \emph{(ii)} is actually a consequence of a simple
assumption on the choice of idealisation, which is always guaranteed
when using \ukano. This result is established relying on the
sub-conditions that have been proved to be sufficient in Lemma~\ref{lem:wa-i}.

\subsubsection{Condition (i)}

Condition \emph{(i)} of well-authentication is basically a conjunction of reachability properties,
which can be checked in \proverif using correspondence
properties \cite{abadi2003computer}.
To each syntactical output $\Out(c,m_0)$ of the initiator role, we associate an event, namely
$\mathtt{Iout_i}(\vect k_I, \vect n_I, \vect m_I)$
 which uniquely identifies the action. We have that:
\begin{itemize}
\item $\vect k_I$ are the identity parameters used in the intiator role;
\item $\vect n_I$ are the sessions parameters; and
\item $\vect m_I$ are the messages inputted and outputted so far in this
  role 
  including $m_0$. 
\end{itemize}
Such an event is placed just before the action $\Out(c,m_0)$. We proceed
similarly for each syntactical input $\In(c,m_0)$ putting the event
$\mathtt{Iin_i}(\vect k_I, \vect n_I, \vect m_I)$ just after the
corresponding input. 
Lastly, we also apply this transformation on the responder role using
events of the form $\mathtt{Rout_i}(\vect k_R, \vect n_R, \vect m_R)$ and
$\mathtt{Rin_i}(\vect k_R, \vect n_R, \vect m_R)$.
To be able to express condition $(i)$ relying on events, we need to
consider some events that will be triggered when 
conditional are passed successfully. Therefore, we add events of the
form $\mathtt{Itest_i}(\vect k_I, \vect n_I, \vect m_I)$ (resp. $\mathtt{Rtest_i}(\vect k_R, \vect n_R, \vect m_R)$) at the beginning of
each $\Then$ branch of the initiator (resp. responder) role.

\smallskip{}

For each conditional of the protocol, we first check
if the simple syntactical definition of {\em safe} conditionals
holds (see \Cref{def:safe}).
If it is the case we do nothing for this conditional.
Otherwise, we need to check condition {(i)} of well-authentication.
This condition can be easily expressed as a correspondence
property relying on the 
events we have introduced. 
Let $\vect k_I = (k^1_I, \ldots, k^p_I)$ and $\vect k_R = (k^1_R,
\ldots, k^q_R)$. We denote $\vect x_I = (x_{k^1_I}, \ldots, x_{k^p_I})$
and $\vect x_R = (x_{k^1_R}, \ldots, x_{k^q_R})$. Note that when
$\vect k_I \cap \vect k_R \neq \emptyset$ (shared case), we have also  that
$\vect x_I \cap \vect x_R \neq\emptyset$ and the correspondence
property (see below) will therefore allow us to
express duality of the two underlying agents.

For instance,
given a conditional of the initiator role tagged with
event  $\mathtt{Itest_i}(\vect k_I,\vect n_I, \vect m_I)$,
we express as a correspondence property the fact that
if the conditional is positively evaluated, then the involved 
agent must be associated to a dual agent as follows:
\begin{enumerate}
\item when the event $\mathtt{Itest_i}(\vect x_I,\vect y_I, (z_1,
  \ldots, z_\ell))$ is fired,
\item   there must be a previous event $\mathtt{Iin_i}(\vect x_I,\vect
  y_I, (z_1,\ldots, z_\ell))$
  (the one just before the
  conditional),
\item  and a previous event $\mathtt{Rout_j}(\vect x_R,\vect
  y_R, (z_1,\ldots, z_\ell))$ (the one corresponding to the output that fed the
  input $\mathtt{Iin_i}$ in an honest execution),
\item and a previous event $\mathtt{Rin_j}(\vect x_R,\vect y_R, (z_1,
  \ldots, z_{\ell-1}))$
  (the one just before the output $\mathtt{Rout_j}$), 
 \etc 
\end{enumerate}
%Note that when $\fn(\ini)\cap\fn(\res)\neq\emptyset$ (shared case), we
%have that $\vect k_I$ and $\vect k_R$ share some elements, and
%therefore the query aboveto reflect 
%that duality also requires that identity parameters
%should be the same, we replace $\vect k'$ by $\vect k$.
Note that by using
the same variables ($z_1, \ldots, z_\ell$) in both the intiator and
responder roles, we express that the
messages that are outputted and inputted
are equal modulo the equational theory $\mathsf{E}$.
We provide in Figure~\ref{fig:FH-with-events} the process obtained by
applying the transformation on the Feldhofer protocol
(Example~\ref{ex:protocol-FH}). In Figure~\ref{fig:query}, we show the
\proverif queries we
have to consider to check
condition~{(i)} on the two conditionals.

\begin{figure}
\footnotesize
\begin{lstlisting}
let SYSTEM = ( !  new k : bitstring; !((
       new nI:bitstring;
       event Iout_1(k,nI,nI); out(ci, nI);
       in(ci, x:bitstring); event Iin_1(k,nI,nI,x)
       let ((=nI,xnr:bitstring)) = dec(x,k) in event Itest_1(k,nI,nI,x);
       event Iout_2(k,nI,nI,x,enc((xnr,nI),k)); out(ci, enc((xnr,nI),k))
      )|(
       new nR: bitstring;   
       in(cr, ynI: bitstring); event Rin_1(k,nR,ynI);
       event Rout_1(k,nR,ynI,enc((ynI,nR),k)); out(cr, enc((ynI,nR),k));
       in(cr, y:bitstring); event Rin_2(k,nR,ynI,enc((ynI,nR),k),y);
       let ((=nR,=ynI)) = dec(y,k) in event Rtest_1(k,nR,ynI,enc((ynI,nR),k),y);
       event Rout_2(k,nR,ynI,enc((ynI,nR),k),y,ok); out(cr, ok)
      ))).
\end{lstlisting}
\caption{Process modelling the Feldhofer protocol with events} 
\label{fig:FH-with-events}
\end{figure}

\begin{figure}
\footnotesize
\begin{minipage}{0.45\textwidth}
\begin{lstlisting}
query x:bitstring, 
      y1:bitstring, y2:bitstring,
      z1:bitstring, z2:bitstring;
   (event(Itest_1(x,y1,z1,z2))  ==>
   (event(Iin_1(x,y1,z1,z2))  ==>
   (event(Rout_1(x,y2,z1,z2))  ==>
   (event(Rin_1(x,y2,z1))  ==>
   (event(Iout_1(x,y1,z1))))))).
\end{lstlisting}
\end{minipage}
 \begin{minipage}{0.50\textwidth}
 \begin{lstlisting}
query x:bitstring, y1:bitstring, 
      y2:bitstring, z1:bitstring, 
      z2:bitstring, z3:bitstring;
   (event(Rtest_1(x,y2,z1,z2,z3))  ==>
   (event(Rin_2(x,y2,z1,z2,z3))  ==>
   (event(Iout_2(x,y1,z1,z2,z3))  ==>
   (event(Iin_1(x,y1,z1,z2))  ==>
   (event(Rout_1(x,y2,z1,z2))  ==>
   (event(Rin_1(x,y2,z1))  ==>
   (event(Iout_1(x,y1,z1))))))))).
\end{lstlisting}
 \end{minipage}
\caption{\proverif queries for checking condition {(i)} on the
  Feldhofer protocol}
\label{fig:query}
\end{figure}

\paragraph{Some practical considerations.}
In our tool, safe conditionals are not automatically identified.
Actually, the tool lists all conditionals and tells which ones
satisfy condition (i) of well-authentication.
The user can thus easily get rid of the conditionals that he identifies as safe.
Furthermore, the structure of the \proverif file produced by \ukano
makes it easy for the user to remove the proof obligations corresponding
to safe conditionals. To obtain more precise encodings once the
translation in Horn clauses is performed by \proverif, we sometimes
push the creation of session parameters (i.e. instructions of the form
$\new\, \mathtt{nI}$).
Therefore, in order to ensure the existence of at
least one session parameter in each event, we systematically introduce
a fresh
session parameter \texttt{sessI} (resp. \texttt{sessR}) which is
is generated at the beginning
of the initiator (resp. responder) role. Such parameters are
systematically added in the events, and since they do not occur in the
messages exchanged during the protocol execution, there is no need to
push them.

Note that, for some examples,
we also verified condition (i) of well-authentication using \tamarin by encoding 
the queries described above as simple lemmas. % (models are available at~\cite{UKANO-ex}). 
In our case, one of the most important advantage of \tamarin over \proverif is
its capability to model the repetition operator~$\rec$ and thus protocols
for which a role executes its sessions in sequence.
Relying on \tamarin, we were thus able to verify
condition (i) for protocols that ensure unlinkability when
sessions are running sequentially but not when they are running
concurrently, \eg we automatically verified the toy example
described in \Cref{ex:toyseq}.

\subsubsection{Condition (ii) - shared case}
\label{sec:mecha-wa-ii}
To verify Condition (ii) of well-authentication,
we rely on Lemma~\ref{lem:wa-i} which provides two sufficient sub-conditions.
Condition $(a)$ of Lemma~\ref{lem:wa-i} can be checked
manually; \ukano leaves it to the user.
Condition $(b)$ may in general be very difficult to verify.
While it is surely possible to reduce the verification of this sub-condition
to classical reachability properties verifiable in \proverif, we
prefer to give a more direct verification technique.

Indeed, once frame opacity is known to hold, condition $(b)$ actually follows immediately
from simple properties of the idealisation function, since checking that 
honest outputs cannot be confused in executions of $\pMa$ is equivalent to 
checking that they cannot be confused in idealised executions.
Often, the idealisation function uses only function symbols that do not 
occur in $\E$ and such that at least one session variable $x^\varN \in \X^\varN$ 
occurs in $\ideam{\ell}$ for each honest output label~$\ell$.
Checking that the idealisation function enjoys these properties is 
straightforward. Let us now show that it implies
condition $(b)$ of Lemma~\ref{lem:wa-i}.

\begin{proposition}
  Let $\Pi = (\vect{k}, \vect n_I, \vect n_R, \dagger_I, \dagger_R, \ini, \res)$
  be a protocol such that  $\fn(\ini)\cap\fn(\res)\neq\emptyset$
  (shared case).
  Consider an idealisation operator
  $\ideam{\cdot}$ such that, for any label $\ell\in\mathcal{L}$
  occurring in the honest execution of $\Pi$,
    some name variable $x \in \X^\varN$ appears in
    $\ideam{\ell}$ in a position only under symbols $\ffun\in\Sigma_c$
    that do not occur in equations of $\E$.
  If $\Pi$ satisfies frame opacity for the idealised operator
  $\ideam{\cdot}$ then condition (b) of Lemma~\ref{lem:wa-i} holds.
\label{prop:mecha:fo:default}
\end{proposition}

\begin{proof}
  Consider an execution $\ta$ of $\pMa$ where agent $a_1$ performs an
  output with label $\ell$ and handle $w_1$, and agent $a_2\neq a_1$
  performs another output with label $\ell$ and handle $w_2$.
  We assume that $\ell$ occurs in the honest execution of $\Pi$
  and we note $\phi$ the resulting frame from the above execution.
  Assume, for the sake of contradiction, that
  $\phi(w_1) =_\E \phi(w_2)$. Since the protocol ensures frame opacity
  for the idealised operator $\ideam{\cdot}$, we deduce that
  $\ideaf^\fr(\ta)(w_1) =_\E \ideaf^\fr(\ta)(w_2)$.
  By hypothesis,
  some name variable $x_1\in\X^\varN$ occurs in $\ideam{\ell}$
  in a position which (even after a substitution) cannot be
  erased by the equational theory nor the computation relation.
  In other words we have that
  $\fr(a_1,x^\varN)$ occurs in $\ideaf^\fr(\ta)(w_1)$,
  and similarly
  $\fr(a_2,x^\varN)$ occurs in $\ideaf^\fr(\ta)(w_2)$,
  at the same position under non-malleable constructor symbols
  only. Since we have assumed (in \Cref{subsec:term})
  that our equational theory is non-degenerate,
  this implies that $\fr(a_1,x^\varN)=_\E\fr(a_2,x^\varN)$ and
  contradicts the injectivity of $\fr$.
\end{proof}

%%
%% UKANO
%%
\subsection{The tool \ukano}
\label{sec:mecha-ukano}
As mentioned earlier,
the tool \ukano~\cite{depotANO}
automatises the encodings described in this section.
It takes as input a \proverif model specifying the protocol to be verified
(and the identity names~$\vect\id$) and returns:
\begin{enumerate}
\item whether frame opacity could be established or not: in particular, it infers  
an idealisation operator that, when in the shared case, satisfies
the assumptions of \Cref{prop:mecha:fo:default}; % discussed in \Cref{sec:mecha-wa-ii};
\item and the list of conditionals for which condition (i) of well-authentication holds.
\end{enumerate}
If frame opacity holds and condition~(i) of well-authentication holds for all conditionals
--- possibly with some exceptions for conditionals the user can identify
as safe --- then the tool concludes that the protocol given as input ensures unlinkability and anonymity w.r.t.~$\vect{\id}$.
Note that the tool detects whether $\fn(\ini) \cap \fn(\res) =
\emptyset$ or not 
and adapts the queries for verifying item~(i) of well-authentication accordingly.
Our tool uses heuristics to build idealised operators
that always satisfy the assumptions
of \Cref{prop:mecha:fo:default}.
Actually, three different heuristics have been implemented.
\smallskip{}

\begin{description}
\item[Syntaxic heuristic.]
The syntaxic heuristic fully adopts the canonical syntactical construction from \Cref{subsec:frame-opacity}
(and displays a warning message when in the shared case, since all requirements are not met in this case).
It can be enabled using the option {\texttt{--ideal-syntaxic}}.

\item[Semantic heuristic.]
The semantic heuristic (enabled with the option \texttt{--ideal-semantic})
follows the semantical construction from \Cref{subsec:frame-opacity} with only tuples identified
as transparent. Roughly,  idealisation of a tuple is a tuple of idealisations of the corresponding 
sub-terms and
idealisation of any other term is a fresh session variable in $\mathcal{X}^\mathsf{n}$.
Such an idealised operator is much less precise (\ie may lead to more false negatives) but since
idealised messages are much simpler, it allows better performance when it works.

\item[Quasi-syntaxic heuristic.]
This heuristic follows the canonical syntactical construction
described in \Cref{subsec:frame-opacity} except that sub-terms having
a function symbol at top-level that is involved in the equational theory will be replaced by
a fresh session name in order to comply with hypothesis of \Cref{prop:mecha:fo:default}.
This is the default heuristic in UKano.
\end{description}

Finally, the user can also define its own idealisations and  the tool
\ukano will check that assumptions of 
\Cref{prop:mecha:fo:default} are satisfied when in the shared case.

\smallskip{}

At a technical level, we built \ukano on top of \proverif.
We only re-used the lexer, parser and AST of \proverif and build upon those 
a generator and translator of \proverif models implementing our sufficient conditions via the above encodings.
This effort represents about 2k OCaml LoC.
% .ml*: 19 files changed, 1674 insertions(+), 918 deletions(-)
% ---> 1674
% DOC (.md/.sh):  3 files changed, 262 insertions(+)
% --> 262
%% ~
% The tool \ukano, as well as  its manual
% % and the generated \proverif models for the examples given in this paper
% can be found at~\cite{depotANO}.
% % S.D.: pour les codages, j'imagine que tu le diras dans la section suivante.
The official page of the tool \ukano with distributed releases of the tool
can be found at \url{http://projects.lsv.ens-cachan.fr/ukano/}.
We also distribute \proverif v1.97 modified for handling extended
diff-equivalence (see \Cref{sec:mecha-fo}). {The difference
  between our modified version of \proverif v1.97 and the original one is about 60 lines of code.}

\newcommand{\mongen}{\mathsf{gen}}

\newcommand{\hfun}{\mathsf{h}}
\newcommand{\montag}{\mathsf{Tag}}
\newcommand{\reader}{\mathsf{Reader}}
\newcommand{\verifier}{\mathsf{Verifier}}
\newcommand{\issuer}{\mathsf{Issuer}}
\newcommand{\client}{\mathsf{Client}}
\newcommand{\tpm}{\mathsf{TPM}}
\newcommand{\PACE}{\mathsf{PACE}}
\newcommand{\ar}{\cdot}

%%
%% Case studies
%%
\section{Case studies}
\label{sec:casestudies}
In this section we apply our verification method to several case studies.
We rely on our tool \ukano
to check whether the protocol under study satisfies frame opacity
and well-authentication as defined in Section~\ref{sec:approach}.
We also discuss some variations of the protocols
to examine how privacy is affected. 
Remind that if privacy can be established for concurrent sessions (\ie $\dagI=\dagR=!$)
then it implies privacy for all other scenarios as well, \ie when
$\dagI, \dagR \in \{\rec, \rep\}$.
We thus model protocols with concurrent sessions and discuss 
alternative scenarios only when attacks are found.
The source code of our tool and material to reproduce
results can be found  at
\begin{center}
\url{http://projects.lsv.ens-cachan.fr/ukano/}.
\end{center}
All case studies discussed in this section except two
(\ie DAA in \Cref{subsec:pace} and ABCDH in \Cref{sec:irma}) have been automatically
verified using our tool \ukano without any manual effort.
We discuss little manual efforts needed to conclude for DAA and ABCDH in the
dedicated sections.
We used UKano~\texttt{v0.5} based on ProVerif~v1.97
on a computer with following specifications:
\begin{itemize}
\item OS: Linux 3.10-2-amd64 \#1 SMP Debian 3.10.5-1x86\_64 GNU\slash Linux
\item CPU \slash\ RAM: Intel(R) Xeon(R) CPU X5650 @ 2.67GHz  \slash\ 47GO
\end{itemize}

%%
%% HASH LOCK
%%
\subsection{Hash-Lock protocol}
\label{subsec:hl}

We consider the Hash-Lock protocol 
as described in~\cite{juels2009defining}. This is an RFID protocol that has been designed to achieve privacy even if no formal proof is
given. We suppose that,
initially, each tag  has his own key~$k$ and the reader maintains a database containing
those keys. The protocol relies on a hash function, denoted $\hfun$,  and can be informally
described as follows.
$$
\begin{array}{rcll}
\reader & \to & \montag: & n_R\\
\montag & \to & \reader: & n_T, \;\hfun(n_R,n_T,k)\\
\end{array}$$

This protocol falls into our generic class of $2$-party protocols in the shared case, and
frame opacity and well-authentication can be automatically established 
in less than 0.01 second. We can therefore conclude that the protocol preserves
unlinkability (note that anonymity does not make sense here).
Actually, all implemented heuristics %(see those three heuristics in \Cref{sec:mecha-ukano})
were able to successfully establish frame opacity automatically.

figure%%
%% LAK
%% 

\subsection{LAK protocol}
\label{subsec:lak}

We present an RFID protocol first introduced in~\cite{LAK'06}, and we refer
to the description given in~\cite{van2008attacks}.
To avoid traceability attacks, the main idea is to ask the tag to
generate a nonce and to use it to send a different message at 
each session.
We suppose that
initially, each tag  has his own key~$k$ and the reader maintains a database containing
those keys. 
The protocol is informally described below ($\hfun$ models a hash
function). 
In the original version (see \emph{e.g.}~\cite{van2008attacks}),
in case of a successful
execution, both parties update the key $k$ with $\hfun(k)$ (they
always store the last two keys).
Our framework does not allow one to model protocols that rely on a mutable
state. Therefore,
we consider here a version where the key is not updated at the end of
a successful execution allowing the key $k$ to be reused from one
session to another.
This protocol lies in the shared case
since the identity name $k$ is used by the reader and the tag.
$$
\begin{array}{rcll}
\reader & \to & \montag: & r_1\\
\montag & \to & \reader: & r_2, \; \hfun(r_1 \oplus r_2 \oplus k)\\
\reader & \to & \montag: & \hfun(\hfun(r_1 \oplus r_2 \oplus k) \oplus k \oplus r_1)
\end{array}
$$

Actually, this protocol suffers from an authentication attack. The
protocol does not allow the reader to authenticate the  tag. 
This attack can be informally described as follows (and already exists
on the original version of this protocol). By using algebraic properties
of $\oplus$, an attacker can impersonate a tag %($I(\montag)$) 
by injecting
previously eavesdropped messages. Below, $I(\mathsf{A})$ means that the attacker plays
the role $\mathsf{\aagent}$.
$$
\begin{array}{rcll}
I(\reader) & \to & \montag: & r_1 \\
\montag & \to & \reader: & r_2, \; \hfun(r_1 \oplus r_2 \oplus k)\\[2mm]
\reader & \to & \montag: & r'_1 \\
I(\montag) & \to & \reader :& r_2^I, \; \hfun(r_1 \oplus r_2\oplus k)\\
\reader & \to & \montag: & \hfun( \hfun({r_1 \oplus r_2}\oplus k)
                              \oplus k \oplus r'_1)
\end{array}
$$
where $r_2^I = r_1' \oplus r_1 \oplus r_2$, thus % in order to have the equalilty
$\hfun(r_1 \oplus r_2 \oplus k)\theo
\hfun(r_1' \oplus r_2^I \oplus k)$.

\smallskip{}

Due to this, the protocol does not satisfy our well-authentication
requirement {even with sessions in sequence for $\montag$ and $\reader$.}
Indeed, the reader can end a session with a tag whereas
the tag has not really participated to this session. In other words,
the reader passes a test (which does not correspond to a safe
conditional) with success, and therefore performs a $\taut$
action whereas it has not interacted honestly with a tag.
Actually, this trace can be turned into 
an attack against the
unlinkability property (for any combination of $\dagI, \dagR \in \{\rec, \rep\}$). Indeed, by continuing the previous trace,
the reader can send a new request to the tag generating a fresh nonce
$r''_1$. The attacker $I(\montag)$ can again
answer to this new request choosing his nonce $r''_2$ accordingly,
\ie $r''_2  = r_1'' \oplus r_1 \oplus r_2$. This execution, involving
two sessions of the reader talking to the same tag, cannot be
mimicked in the single session scenario, and corresponds to an attack
trace.

More importantly, this scenario can  be seen as a traceability attack 
on the stateful version of the protocol leading to a practical attack.
The attacker will first start a session with
the targeted tag by sending it a nonce~$r_1$ and storing its answer.
Then, later on, he will interact with
the reader as described in the second part of the attack scenario. Two
situations may occur: either the interaction is successful meaning
that the targeted tag has not been used since its last interaction
with the attacker; or the interaction fails meaning that the key has
been updated on the reader's side, and thus the targeted tag has performed a
session with the reader since its last interaction with the attacker.
%\lucca{
This attack shows that the reader may be the source of leaks exploited by the attacker
to trace a tag. This is why we advocate for the strong notion of unlinkability
we used, taking into account the reader and considering it as important as the tag.

We may note that the same protocol was declared untraceable
in~\cite{van2008attacks} due to the fact that they have in mind a weaker
notion of unlinkability. Actually, their notion captures the
  intuitive notion that a tag is untraceable if for any execution in
  which two actions are performed by the same tag, there is another
  execution indistinguishable from the original one in which the
  actions have been performed by two different tags. We may
  note that in the attack scenario described above, 
 the tag in itself does not leak anything but the reader does,
  explaining why this weak notion of untraceability missed this attack.

Now, to avoid the algebraic attack due to the properties of the xor
 operator, we may replace it by the pairing operator. 
The resulting protocol is a 2-party protocol that falls into our
class, and for which frame opacity and well-authentication can be
established {(with concurrent sessions)} using \ukano (any heuristic
is suitable for that).
 Therefore, Theorem~\ref{theo:main} allows
us to conclude that it preserves unlinkability.
%Note also that frame opacity can be automatically checked using any heuristic described in \Cref{sec:mecha-ukano}.

%%
%% BAC
%%

\subsection{BAC protocol and some others}
\label{subsec:bac}

An e-passport is a paper passport with an RFID chip that stores the
critical information 
printed on the passport. The International Civil Aviation Organization
(ICAO) standard~\cite{ICAO-Passport}
specifies several protocols through which this
information can be accessed.
Before executing the Basic Access Control (BAC) protocol,
the reader optically scans a weak secret from which
it derives two keys $k_E$ and $k_M$ that are then shared between the
passport  and the reader.
Then, the BAC protocol establishes a key seed from which two sessions keys
are derived. 
The session keys are then used to 
prevent skimming and eavesdropping on subsequent
communications.

In~\cite{arapinis-csf10}, two variants
of the BAC protocol are
described and analysed. We refer below to these two variants as the French version and the United Kingdom (U.K.) version.
The U.K.~version is claimed unlinkable (with no formal proof)
whereas an attack is reported on
the French version. 
We first give an informal  description of the BAC protocol
using Alice \& Bob notation:
$$
\begin{array}{ll}
\montag \to \reader: & n_T\\
\reader \to \montag: & \{n_R,n_T,k_R\}_{k_E},
                       \mac(\{n_R,n_T,k_R\}_{k_E}, k_M)\\
\montag \to \reader: & \{n_T,n_R,k_T\}_{k_E},
                       \mac(\{n_T,n_R,k_T\}_{k_E}, k_M)\\
\end{array}
$$

Then, to explain the difference between the two
versions, we give a description of the passport's role  in
Figure~\ref{fig:passport}. 
\begin{figure}[t]
$$
\begin{array}{ll}
T({k_E},{k_M}) \;  = &% \In(=get\_chall).  
\new n_T. \new k_T. 
\Out(c_T, n_T).\In(c_T, x).\\
&
\Let \; x_E = \projl{x}, \;x_M = \projr{x},  \;
  z_{\mathsf{test}} = \eq(x_M, \mac(x_E,k_M)) \; \In\\
&\hspace{0.7cm} \Let \;  z'_{\mathsf{test}} =
  \eq(n_T,\projl{\projr{\sdec(x_E,{k_E})}}) \; \In \;
\Out(c_T,\langle m, \mac(m,{k_M})\rangle)\\
\phantom{\Let\; } %\hspace{3cm}
&\hspace{0.7cm} \Else \;\Out(\mathsf{{error}_{Nonce}}) \\
&\Else \; \Out(\mathsf{error_{Mac}}) \\
\end{array}
$$
where $m = \senc(
\langle n_T, \langle \projl{\sdec(x_E,k_E)},k_T\rangle\rangle,{k_E})$.
\caption{Description of the passport's role}
\label{fig:passport}
\end{figure}

 We do not model the
  \textsf{getChallenge} constant message that is used to initiate the
protocol but it is clear this message does not play any role
regarding the security of the protocol.
We consider the  signature given in~Example~\ref{ex:signature}
augmented with a
function symbol $\mac$ of arity $2$. This is a public constructor 
 whose purpose is to model message
authentication code,
taking as arguments the message to authenticate and the mac key.
There is no rewriting rule and no equation
regarding this symbol.
We also assume public constants to model error messages.
The U.K.~version of the protocol does not distinguish the two
cases of failure, \emph{i.e.} $\mathsf{error_{Mac}}$ and
$\mathsf{error_{Nonce}}$ are the same constant, whereas the French
version does.
 The relevant point is the fact that, in
case of failure, the French version sends a different error message
indicating whether the failure occurs due to a problem when checking
the mac, or when checking the nonce. This allows the attacker to
exploit this conditional to learn if the mac key of a tag is the one used
in a given message  $\langle m,\mac(m,k) \rangle$. Using this, he can very easily
trace a tag~$T$ by first eavesdropping an honest interaction between the tag~$T$
and a reader.

The U.K.~version of the BAC protocol is a 2-party protocol
according to our definition. Note that since  the two
error messages are actually identical, we can merge the two
\textsf{let} instructions, and therefore satisfy our definition of
being a responder role.
Then, we automatically proved frame opacity and
well-authentication using \ukano.
It took less than 0.1 second independently of the chosen heuristic
regarding frame opacity.
%\lum{temps pour les différents heuristiques en com}
% Voir plus bas 
Therefore, Theorem~\ref{theo:main} allows us to conclude that
unlinkability is indeed satisfied.

Regarding the French version of this protocol, it happens that the
passport's role is neither an initiator role, nor a responder role
according to our formal definition. Indeed, our definition of a role,
and therefore of a 2-party protocol does not allow to model two
sequences of tests that will output different error messages in case
of failure. As illustrated by the attack on the French version of the
BAC protocol, imposing this syntactic condition is actually a good design
principle w.r.t. unlinkability.

Once the BAC protocol has been successfully executed, the reader
gains access to the information stored in the RFID tag through
the Passive and Active Authentication protocols
(PA and AA).
They are respectively used to prove authenticity of
the stored information and prevent cloning attacks,
and may be executed in any order.
A formal description of these protocols is available in~\cite{ACD-csf12}.
These two protocols also fall into our class and our conditions can be
checked automatically both for unlinkability and anonymity
properties.
We can also use our technique to analyse directly the
three protocols together (\ie the U.K.~version of the BAC together with 
the PA and AA protocols in any order).
We analysed both orders, \ie BAC followed by PA, and then AA, as well
as BAC following by AA, and then PA.
We establish unlinkability and anonymity
w.r.t.\ all private data stored in the RFID chip
(name, picture, \etc). \ukano concludes within 1 second to
establish both well-authentication and frame opacity (independently of
the selected heuristic).

%%%%% BENCHMARKS %%%%%%%%
% | Protocol    | Better time (total) | Time for WA | Time for FO (greedy) | Time for FO (default) | Time for FO (syntax)  | Time for FO (user-defined) |
% |:------------|:-------------:|:-------------------:|:-------------------:|:---------------------:|:--------------------:|:---------------------------|
% | Hash-Lock      | 0.00s  | 0.00s | 0.00s  | 0.00s   | 0.00s   | --    |
% | Fixed LAK      | 0.00s  | 0.00s | 0.00s  | 0.00s   | 0.00s   | --    |
% | BAC            | 8.41s  | 0.02s | 8.39s | 17.24s  | 17.20s  | --    |
% | BAC+AA+PA      | 198.28s| 0.42s |197.86s | 1013.56s    | 998.81s    | --    |
% | BAC+PA+AA      | 183.40s| 0.33s |183.07s|  1068.79s | 1191.04s   | --    |
% | PACE with tags | 169.91 | 62.99s| 106.92s (*) | :curly_loop:   | :curly_loop: |106.92s |
% | DAA simplified [HBD17]| 0.02s |0.01s| :x: | 0.01s  | 0.00s   | --    |
% | DAA sign       | 2.94s  | 0.01s | :x:    | :x:     | 2.76s   | --    |
% | DAA join       | 4.68s  | 1.82s | 2.30s  | 2.30s   | 28.85s  | --    |
% | ABCDH (irma)   | todo   | todo | :x: | :x: |  2389.76s* |  2389.76s |

\subsection{PACE protocol}
\label{subsec:pace}
% \lum{TODO:
% dans le cas concurrent on a l'attaque et le fix comme dans S\&P. Si on est sûr que les sessions
% ne sont pas conucrrentes alors on peut vérifier Pace sans les tags comme ceci:
% FO en proverif en faisant la sur-approx avec sessions concurrentes;
% WA en Tamarin avec sur-approx mais axiome qui interdit les croisements.
% MAIS dans de cas WA sera violée si on modélise DH correcetement à cause de l'attaque sur DH que l'on connaît.
% Que fais-t-on?
% }

The Password Authenticated Connection
Establishment protocol %~\cite{bsiPACE} 
(PACE)
has been proposed by the German Federal Office for
Information Security (BSI) to replace the BAC protocol.
It has been studied in the literature~\cite{bender2009security},
\cite{bender2012pace},
\cite{cheikhrouhou2012merging} 
but to the best of our knowledge, no formal proofs about privacy 
have been given.
Similarly to BAC, its purpose is to establish a secure channel based 
on an optically-scanned key~$k$.
%
%\smallskip{}
%
This is done in four main steps (see Figure~\ref{fig:pace:AB}):
\begin{itemize}
\item The tag chooses a random number $s_T$,
encrypts it with the symmetric key $k$ shared between the tag and the reader %key %password-derived key $k$  $k$
and sends the encrypted random number to the reader (message 1).
\item Both the tag and the reader perform a Diffie-Hellman exchange  (messages 2 \& 3),
  and derive~$G$ from $s_T$ and $g^{n_R n_T}$.
\item  The tag and the reader perform a Diffie-Hellman exchange
  based on the parameter $G$ computed at the previous step  (messages 5 \& 6).
\item The tag and the reader derive a session key $k'$ 
which is confirmed by exchanging and checking the authentication
tokens (messages 8 \& 9).
\end{itemize}
Moreover, at step 6, the reader is not supposed to accept
as input a message which is equal to the previous message that it has
just sent.
\begin{figure}[t]
  \centering
  $
  \begin{array}{lrcll}
    1.&\montag& \to &\reader : & \{s_T\}_{k}\\
    2. &\reader & \to & \montag: & g^{n_R}\\
    3. &\montag & \to & \reader: & g^{n_T}\\
    4. & \multicolumn{4}{l}{\mbox{ Both parties compute $G  = \mongen(s_T,g^{n_R n_T})$.}}\\
    5. & \reader &\to & \montag: & G^{n'_R} \\
    6. & \montag & \to & \reader: & G^{ n'_T}\\
    7. & \multicolumn{4}{l}{\mbox{Both parties compute $k' =G^{n'_R n'_T}$}}\\
    8. &\reader & \to & \montag: & \mac(G^{n'_T}, k')\\
    9.& \montag & \to &\reader: & \mac(G^{n'_R},k')
  \end{array}
  $
  \caption{PACE in Alice \& Bob notation}
\label{fig:pace:AB}
\end{figure}

\smallskip{}

To formalise such a protocol, we consider $\Sigma_c = \{\senc, \; \sdec, \; \dihe, \; \mac, \; \mongen, \; \gconst,
\;\ok\}$, and  $\Sigma_d = \{ \mathsf{neq}\}$.

\noindent Except $\gconst$ and $\ok$ which are public constants, all these
function symbols are public constructor symbols of arity 2.
The destructor $\mathsf{neq}$ has already be defined in \Cref{ex:neq}.
The symbol $\dihe$ is used to
model modular exponentiation whereas $\mac$ will be used to model
message authentication code. We consider the equational theory $\E$
defined by the following equations:
$$\sdec(\senc(x,y),y) \;=\; x\qquad 
\dihe(\dihe(x,y),z) \;=\; \dihe(\dihe(x,z),y)
$$

%\marginpar{\stef{la modelisation avec neq est etrange !}}
\begin{figure}[t]
  \centering
$  \begin{array}{ll}
    \res_\PACE  := & \In(c_R, y_1). \\
    &\Out(c_R,  \dihe(\gconst, n_R)). \In(c_R, y_2). \; \\
    & \Out(c_R, \dihe(G,n'_R)).  \In(c_R, y_3). \\
    & \Let  \; y_\mathsf{test} = \mathsf{eq}(\mathsf{yes},\mathsf{neq}(y_3, \dihe(G,n'_R))) \; \In\\
    &\ \ \Out(c_R, \mac(y_3,k')); \\
    &\ \   \In(c_R, y_4). \\
    &\ \  \Let \; y_5 = \eq(y_4, \mac(\dihe(G,n'_R), k')) \; \In \; R'.
  \end{array}$

\smallskip{}

\noindent where $G =  \mongen(\sdec(y_1,k) , \dihe(y_2, n_R))$ and $k' = \dihe(y_3,n'_R)$.

  \caption{Process $\res_\PACE$}
\label{fig:pace:proc}
\end{figure}

We consider the process $\res_\PACE$ as described in
Figure~\ref{fig:pace:proc}. We do not detail the continuation~$R'$ {and we omit trivial conditionals}.
The process modelling the role $\ini_\PACE$ can be obtained in a similar way. 
Then, we consider $\Pi_\PACE = (k, (s_T, n_T,n'_T), (n_R, n'_R), !, !, 
\ini_\PACE, \res_\PACE)$ which falls into our generic class of
$2$-party protocols.
Unfortunately, \proverif cannot handle the equation above on
  the $\dihe$ operator (due to some termination issues). 
Instead of that single equation,
we consider the following equational theory that is more
suitable for \proverif:
$$
\dihe(\dihe(\mathsf{g},y),z) \;=\; \dihe(\dihe(\mathsf{g},z),y) \qquad
\dihe(\dihe(\mongen(x_1,x_2),y),z) \;=\; \dihe(\dihe(\mongen(x_1,x_2),z),y) 
$$

\noindent This is sufficient for the protocol to work properly but it obviously lacks
equations that the attacker may exploit.

\smallskip{}

First, we would like to highlight an imprecision in the official
specification %~\cite{bsiPACE}
that may lead to practical attacks on unlinkability.
As the specification seems to not forbid it, we could have 
assumed that the decryption operation in $G =  \mongen(\sdec(y_1,k) , \dihe(y_2, n_R))$ is implemented
in such a way that it may fail when the key $k$ does not match with the key of the ciphertext
$y_1$. In that case, an attacker could eavesdrop a first message $c^0=\senc(s_T^0,k^0)$ of a certain tag $T^0$
and then, in a future session, it would let the reader optically scan a tag $T$ but replace its challenge $\senc(s_T,k)$
by $c^0$ and wait for an answer of the reader. If it answers, he learns that the decryption did not fail and
thus $k=k^0$: the tag $T$ is actually~$T^0$.
We discovered this attack using our method since, in our first attempt to model the protocol,
we modelled $\sdec(\cdot,\cdot)$ as a destructor (that may fail) and the computation of $G$ as an evaluation:
$$
\Let \; G = \mongen(\sdec(y_1,k), \dihe(y_2,n_R)) \; \In\; [...]
$$
In order to declare the protocol well-authenticating, this conditional
computing~$G$ which is not safe has to satisfy our requirement (see
Definition~\ref{condi:auth}). However, as witnessed by the attack
scenario described above (the reflection attack),  the condition actually
fails to hold.
Incidentally, the same attack scenario shows that the protocol does not ensure unlinkability
(this scenario cannot be observed when interacting with~$\pS$).
Similarly to the attack on LAK, we highlight here the importance to take the reader
into account and give it as much importance as the tag in the definition of unlinkability.
Indeed, it is actually a leakage from the reader that allows an
attacker to trace a specific tag.

\smallskip{}

Second, we now consider that decryption is a constructor, and thus
cannot fail, an we report on an attack
%\footnote{For that different attack, we 
%obviously consider that decryption is a constructor, and thus cannot fail.}
that we discovered using our method
on some models of PACE found in the literature~\cite{bender2009security},\cite{bender2012pace},\cite{cheikhrouhou2012merging}.
Indeed, in all those papers,
the first conditional of the reader
%$$\Let  \; y_\mathsf{test} = \mathsf{neq}(y_3, \dihe(G,n'_R)) \; \In$$
%
$$ \Let  \; y_\mathsf{test} = \mathsf{eq}(\mathsf{yes},\mathsf{neq}(y_3, \dihe(G,n'_R))) \; \In$$
\noindent is omitted.
Then the resulting protocol is not well-authenticating.
To see this, we
simply have to consider a scenario where the attacker will send to the
reader the message it has outputted at the previous step.
Such an execution will allow the reader to execute its role until the
end, and therefore execute $\taut$, but the resulting trace is
not an honest one. Again, this scenario can be turned into an attack
against unlinkability as explained next.
As before, an attacker could eavesdrop a first message $c^0=\senc(s_T^0,k^0)$ of a certain tag $T^0$.
Then, in a future session, it would let the reader optically scan a tag $T$ but replace its challenge $\senc(s_T,k)$
by $c^0$. Independently of whether~$k$ is equal to~$k^0$ or not, the reader answers $g^{n_R}$. The attacker then plays the two rounds of
Diffie-Hellman by reusing messages from the reader (he actually
performs a reflection attack).
More precisely, he replies with $g^{n_T}=g^{n_R}$, $G^{n'_T}=G^{n'_R}$ and $\mathsf{mac}(G^{n'_R}, k')=\mathsf{mac}(G^{n'_T}, k')$.
The crucial point is that the attacker did not prove he knows~$k$ (whereas he is supposed to do so to generate $G$ at step 4)
thanks to the reflection attack that is not detected.
Now, the attacker waits for the reader's answer. If it is positive (the process $R'$ is executed),
he learns that $k=k^0$: the tag $T$ is actually
the same as $T^0$.

\smallskip{}

Third, we turn to PACE as properly understood from the official specification:
when the latter test is present and the decryption may not fail.
In that case, we report on a new attack.
\ukano~found that the last test of the reader violates well-authentication.
This is the case for the following scenario: the message $\senc(s_T,k)$ sent by a tag $T(k,n_T)$
is fed to two readers $R(k,n_R^1),R(k,n_R^2)$ of same identity name.
Then, the attacker just forwards messages from one reader to the other. They can thus complete
the two rounds of Diffie-Hellman (note that the test avoiding reflection attacks holds).
More importantly, the mac-key verification phase (messages~8 and~9 from Figure~\ref{fig:pace:AB})
goes well and the attacker observes that the last conditional of the two readers holds.
This violates well-authentication but also unlinkability because the latter scenario cannot
be observed at all in $\pS$: if the attacker makes two readers talk to each other in $\pS$ they cannot
complete a session because they must have different identity names.
In practice, this flaw seems hard to exploit but it could be a real privacy concern:
if a tag initiates multiple readers, an attacker
may learn which ones it had initiated by forwarding messages from one to another.
It does not seem to be realistic in the e-passport scenario, but could be harmful in other contexts.
It seems that, in the e-passport context,  a modelling with sequential sessions would be more realistic.
We come back to such a modelling at the end of this section.

\smallskip{}

Further, we propose a simple fix to the above attack by adding tags avoiding confusions between reader's
messages and tag's messages.
It suffices to replace messages~8 and~9 from Figure~\ref{fig:pace:AB} by respectively
$\mac(\langle \mathsf{c}_r, G^{n'_T}\rangle, k')$
and $ \mac(\langle \mathsf{c}_t, G^{n'_R}\rangle, k')$
where $\mathsf{c}_r,\mathsf{c}_t$ are public constants,
and adding the corresponding checks.
Well-authentication can be automatically established using \ukano in around 1 minute.
Frame opacity can be automatically established using any heuristic described in \Cref{sec:mecha-ukano}. 
Heuristics producing more complex idealisations (\ie the syntaxic one)
are less efficient. Nevertheless, the tool concludes in at most 16
seconds. 
We thus conclude that PACE with tags preserves unlinkability in the model considered here.

\smallskip{}

\subsection{Attributed-based authentication scenario using ABCDH protocol}
\label{sec:irma}

Most authentication protocols are identity-based: the user needs to provide his identity
and prove to the service provider he is not trying to impersonate somebody else.
However, in many contexts, the service provider just needs to know that the user has some non-identifying
attributes (\eg age, gender, country, membership).
For instance, a liquor shop just needs to have the proof that the user has the right to buy liquors (\ie that he is old enough)
and does not need to know the full identity of the user (\eg as it is currently done when showing ID cards).
Attribute-based authentication protocols solve this problem and
allow a user to prove to another user, within a secure channel,
that he has some attributes without disclosing its identity.

We used our method to automatically establish unlinkability of a typical use case of such a protocol
taking part to the IRMA
project\footnote{For more information about IRMA (``I Reveal My Attributes''), see \url{https://www.irmacard.org.}}.
We analysed a use case of the protocol ABCDH as defined in~\cite{alpar2013secure}. This protocol
allows a smartcard $C$ to prove to some Verifier $V$
that he has the required
attributes. The protocol aims at fulfilling this goal without revealing the identity of~$C$ to~$V$ or to anyone else.
One of its goal is also to avoid that any other smartcard $C'$ replays those attributes later on.
The protocol should also ensure unlinkability of~$C$. To the best of our knowledge, there was no prior formal analysis
of that security property for this protocol.

The key ingredient of this protocol is {\em attribute-based credential} (ABC). It is a cryptographic container for attributes.
In ABC, attributes are signed by some issuers and allow for {\em selective 
disclosure} (SD): it is possible to produce
a zero-knowledge (ZK) proof revealing a subset of attributes signed by the issuer along with a proof that the selected
disclosed attributes are actually in the credential. This non-interactive proof protocol can be bound to some fresh data
to avoid replay attacks. We shall use the notation $\SD(\vect{a_i};n)$ to 
denote the selective disclosure of attributes $\vect{a_i}$ bound to $n$.
Note that $\SD(\emptyset;n)$ (no attribute is disclosed) still proves the existence of a credential.
There are two majors ABC schemes: Microsoft U-Prove~\cite{u-prove} and IBM's 
Idemix~\cite{camenisch2012electronic}. We decided to model IBM's scheme
(since it is the one that is used in IRMA) following the formal model
given in~\cite{camenisch2010formal}. We may note that we
  consider here some privacy issues whereas the security analysis
  presented  in~\cite{camenisch2010formal} is dedicated to the
  analysis of some reachability properties.
It involves complex cryptographic primitives (\eg commitments, blind signature, ZK proofs) but \proverif can deal with them all.
In this scheme, each user has a master secret never revealed to other parties.
Issuers issue credentials bound to the master secret of users (note that users are known to issuers under pseudonyms).
A SD consists in a ZK proof bound to $n$ proving some knowledge:
knowledge of the master secret,
knowledge of a credential bound to the master secret,
knowledge that the credential has been signed by the given organisation,
knowledge that the credential contains some given attributes.

We analyse the ABCDH~\cite{alpar2013secure} using the model of SD from~\cite{camenisch2010formal} used in the following scenario:
\begin{itemize}
\item an organisation $\Oage$ issues credentials about the age of majority;
\item an organisation $\Ocheck$ issues credentials giving the right to check the age of majority;
\item a user $C$ wants to watch a movie rated adult-only due to its violent contents; his has a credential from $\Oage$ with the attribute \texttt{adult};
\item a movie theatre $V$ wants to verify whether the user has the right to watch this movie; it has a credential from $\Ocheck$ with the attribute \texttt{canCheckAdult}.
\end{itemize}
The scheme is informally given in Figure~\ref{fig:abcdh}.
% Identity parameters: userB, userA,
% --> $\textit{user}_V, \textit{user}_C$
% Session names: sess_I, sess_R, xv, NEWu_88, n, xc, NEWu, NEWn,
% --> $n_V,n_C,n$

$n_V,n_C$ and $n$ are fresh nonces.
Functions $\mathsf{f}_1\slash 1$,$\mathsf{f}_2\slash 1$ and $\mathsf{f}_3\slash 2$ are independent
hash functions; we thus model them
as free constructor symbols.
The construction $\SD(\cdot;\cdot)$ is not modelled atomically
and follows~\cite{camenisch2010formal} but we do not describe here its details.
We note however that when $V$ (respectively $C$) sends a $\SD(\cdot;\cdot)$,
  the corresponding message we do not detail here contains an identity-parameter
$\textit{user}_V$ (respectively $\textit{user}_C$).

\newcommand{\g}{g}

\begin{figure}[h]
  \centering
    $
  \begin{array}{lrcll}
    1.&V& \to &C : & \dihe(g,n_V),\SD( \texttt{canCheckAdult};\mathsf{f}_1(\dihe(g,n_V)))\\
    2.&C& \to &V : & \dihe(g,n_C),\SD(\emptyset;\mathsf{f}_1(\dihe(g,n_V)),\dihe(g,n_C))\\
    3.&V& \to &C : & \senc(\langle\mathtt{0x00},\ok\rangle,k) \\
    4.&C& \to &V : & \senc(\langle\mathtt{0x01},\ok\rangle,k) \\
    5.&V& \to &C: & \senc(\langle n,\mathtt{requestAdult}\rangle,k) \\
    6.&C& \to &V : & \senc(\langle \mathtt{adult},
                     \SD(\mathtt{adult};\mathsf{f}_3(n,\mathtt{seed}))\rangle,
                     k)
%%%%%%  OLD VERSION in the THESIS:
    % 1.&\verifier& \to &\client : & \dihe(\g,n_V),\SD(\emptyset;f_1(\dihe(\g,n_V)))\\
    % 2.&\client& \to &\verifier : & \dihe(\g,n_C),\SD(\emptyset;f_1(\dihe(\g,n_V)),\dihe(g,n_C))\\
    % 3.&\verifier& \to &\client : & \senc(\langle\mathtt{0x00},\ok\rangle,k) \\
    % 4.&\client& \to &\verifier : & \senc(\langle\mathtt{0x01},\ok\rangle,k) \\
    % 5.&\verifier& \to &\client : & \senc(\langle n;\mathtt{requestAdult}\rangle,k) \\
    % 6.&\client& \to &\verifier : & \senc(\langle \mathtt{adult};
    %                  \SD(\mathtt{adult};f_3(n,\mathtt{seed}))\rangle,
    %                  k)
\end{array}$
  \caption{ABCDH (where $seed=\dihe(\dihe(\g,n_C),n_V)$ and $k=\mathsf{f}_2(seed)$)}
  \label{fig:abcdh}
\end{figure}

This is a 2-party protocol that falls into our class. Actually, we
have that $\textit{user}_V\in\fn(\ini)\cap\vect{k} \neq \emptyset$ and
$\textit{user}_C\in\fn(\res)\cap \vect{k}\neq\emptyset$, but  $\fn(\ini) \cap \fn(\res) = \emptyset$
(non-shared case).
The complete model of this protocol is quite complex and can be found
in~\cite{depotANO}.
  Frame Opacity can be automatically established
  using the syntaxic heuristic (see \Cref{sec:mecha-ukano}) in less than 40 seconds.
  The other heuristics were not enough precise to conclude (yielding negative results) showing the importance of having the choice between
  heuristics that are precise but less efficient and ones that are more efficient but less precise. 
Regarding well-authentication, due to the length of the protocol,
the queries are also quite long. Because of the latter
and the high complexity of the underlying term algebra,
it required too much time for \proverif to terminate.
We addressed this performance issue by soundly splitting up big queries into smaller ones.
This way, we successfully established well-authentitcation
for this protocol within 3 hours.

% To sum up, \ukano was able to verify unlinkability of the ABCDH
% protocol (after some minor manual modifications\footnote{Note that those modifications are systematic and could be automated in the future.} on the files provided by \ukano) in 
% in around 3 hours.
% Tout WA en 2h30 (avec simplification de WA comme expliqué en conclu/limitation)

\subsection{DAA join \& DAA sign}
\label{sec:DAA-join-sign}
A Trusted Platform Module (TPM) is a hardware device aiming at protecting cryptographic keys and at performing
some cryptographic operations. Typically, a user may authenticate himself to a service provider relying on
such a TPM.
The main advantage is to physically separate the very sensitive data from the rest of the system. 
On the downside however, such devices may be used by malicious agents to breach users' privacy
by exploiting their TPMs.
Direct Anonymous Attestation (DAA) protocols have been designed to let TPMs authenticate themselves whilst
providing accountability and privacy.

In a nutshell, some issuers issue credentials representing membership to a group
to the TPM using group signatures via the {\em DAA join} protocol.
Those credentials are bound to the internal secret of the TPM that must remain unknown to the service provider.
Then, when a TPM is willing to prove to a verifier its membership to a group, it uses the {\em DAA sign} protocol.
We analysed the RSA-based DAA join and sign protocols as described in~\cite{smyth2015formal}.
Both protocols rely on complex cryptographic primitives (\eg blind signatures, commitments, and Zero Knowledge proofs) but \proverif can
deal with them all.
Note that the authors of~\cite{smyth2015formal} have automatically established a game-based version of unlinkability
of the combination of DAA Join and DAA Sign using \proverif.
We only provide an analysis of each protocol in isolation since the combination of the two protocols is a 3-party protocol.

\subsubsection{DAA join}
\label{sec:DAA-join-sign:join}
In the RSA-based DAA join protocol, the TPM starts by sending a credential request in the form
of a commitment containing its internal secret, some session nonce and the public key of the issuer.
The issuer then challenges the TPM with some fresh nonces encrypted asymmetrically with
the public key of the TPM. After having received the expected TPM's answer, the issuer sends a new nonce
as second challenge.
To this second challenge, the TPM needs to provide a ZK proof bound to this challenge
proving that he knows the internal secret on which the previous commitment was bound.
Finally, after verifying this proof, the issuer blindly signs the commitment allowing the TPM
to extract the required credential.

\newcommand{\TPM}{\mathsf{TPM}}
\newcommand{\pubI}{pub_I}
\begin{figure}[h]
\centering
$\begin{array}{lrcll}
1. & \tpm & \to & \issuer: & N_I,U \\
2. & \issuer & \to & \tpm: & \mathsf{penc}(n_e,n,\pk(sk_{\TPM})) \\
3. & \tpm & \to & \issuer: & \h((U,n_e))\\
4. & \issuer & \to & \tpm: & n_i \\
5. & \tpm & \to & \issuer: & n_t,\mathsf{ZK}_{\mathsf{join}}(% \mathtt{join},
                                (tsk,n_v),
                                (zeta_I,N_I,U,(n_t,n_i)))\\
6. & \issuer & \to & \tpm: & \mathsf{clsign}((U,r),sk_I)
\end{array}$
 \caption{DAA Join}
  \label{fig:daa-join}
\end{figure}

We give in \Cref{fig:daa-join} an Alice \& Bob description of the protocol between the TPM and the issuer.
The message $zeta_I=\h((\mathsf{0},bsnI))$ relies on $bsnI$: using
a fresh $bsnI$ allows to ensure that the session of DAA Join will be unlinkable from previous ones.
The message $tsk=\h((\h((\textrm{DAAseed},\h(KI))),cnt,0))$ combines the internal
secret of the $\TPM$ (\ie $\textrm{DAAseed}$) with the public long-term key of the issuer (\ie $KI$).
The commit message $N_I= \mathsf{commit}(zeta_I,tsk)$ binds $zeta_I$ with the internal secret
while the commit message $U= \mathsf{clcommit}(\pk(sk_I),n_v,tsk)$ expresses a credential request for a signature key $sk_I$.
The goal of the $\TPM$ will be to get the message $U$ signed by the issuer. More precisely, the issuer
will blindly sign the message $U$ with the signature key $sk_I$ after making sure that the $\TPM$ can decrypt challenges encrypted with its public key (step 2.)
and that he can provide a fresh ZK proof showing he knows its internal secret binds in $U$ and $N_I$ (step 5.).
Finally, if all checks are successful, the issuer will blindly sign the credential request $U$ (step 6.).
We note $\mathsf{clsign}((U,r),sk_I)$ the blind signature of a commitment
$U$ with signature key $sk_I$ and some random $r$. The function $\mathsf{penc}$ denotes a randomized
asymmetric encryption scheme.
%An Alice \& Bob description is given in \Cref{fig:daa-join}.
Note that $\mathsf{ZK}(\ar,\ar)$ has two arguments: 
the first one should contain private data and the second
one should contain public data. One can always extract public data from ZK proofs and one can check
if both public and private data match as expected.

This protocol falls in our class and lies in the shared case\footnote{Both roles share the identity name $sk_\TPM$ (but note that the Issuer only uses $\pk(sk_\TPM)$). Indeed, before executing the join protocol, the TPM and the issuer should establish a one-way authenticated channel that is not specified by the DAA scheme. Therefore, an Issuer session is associated to a single TPM's identity it is expected to communicate with.}
(\ie $\fn(\ini)\cap\fn(\res)=\{sk_\TPM\}$).
 % {OK, que
 %  partage-t-il exactement? Est-ce uniquement ce fameux canal
 %  fraichement établi?}\marginpar{?}
%
% generated idealized messages (any heuristic implemented in \ukano but
% the {greedy heuristic} allows one to obtain simpler idealised messages which yield better performance
% (2 seconds instead of 28 seconds for the fully syntactical heuristic).
\ukano automatically established frame opacity in less than 30
  seconds using the syntaxic idealisation, and in less that 3 seconds
  when using the quasi-syntaxic heuristic. Note that the semantic
  one is not precise enough to allow one to conclude.
Regarding well-authentication, we had to leverage the same splitting technique explained in \Cref{sec:irma}
so that \ukano could conclude in a reasonable amount of time (around 30 seconds).

\subsubsection{DAA sign}
\label{sec:caseStudies:DAA}
\label{sec:DAA-join-sign:sign}
Once a TPM has obtained such a credential, it may prove its membership using the DAA sign protocol.
This protocol is played by a TPM and a verifier: the verifier starts by challenging the TPM with a fresh nonce (step 1.),
the latter then sends a complex ZK proof bound to this nonce (step 2.).
The latter ZK proof also proves that the TPM knows a credential from the expected issuer bound to a secret he knows
(essentially a message $\mathsf{clsign}((U,r),sk_I)$ received in a previous session of DAA join).
The verifier accepts only if the ZK proof can be successfully checked (step 3.).

We give in Figure~\ref{fig:daa-sign} an Alice \& Bob description of the protocol between a verifier
and the TPM willing to sign a message
$m$ using its credential $cred= \mathsf{clsign}((U,r),sk_I)$
he received from a past DAA join session.
From its credential $cred$,
the TPM will compute a new credential dedicated to the current sign session:
$cred'=\mathsf{clcommit}((\pk(sk_I),cred),n_c)$. Indeed, if the TPM
had directly used $cred$ then two sessions of DAA sign would have been trivially linkable.
The TPM also computes a commit of $tsk$ that was used to obtain the credential: $N_V=\mathsf{commit}(tsk,zeta_V)$ where $zeta_V$ is a fresh nonce\footnote{The protocol also specifies a mode that makes different signatures linkable by construction using $zeta_V=\h((0,{bsnV}))$. We focus on the other mode for which unlinkability is expected to hold.}.

\begin{figure*}[th]
  \centering$
  \begin{array}{lrcll}
    1. & \verifier & \to & \TPM: & n_v\\
    2. & \tpm  & \to & \verifier: & (zeta_V,\pk(sk_I),N_V,cred',n_t,\\
       &       &     &           &\phantom{(}
                              \mathsf{ZK_{sign}}((tsk,n_c),(zeta_V,\pk(sk_I),N_V,cred',(n_t,n_v,m)))) \\ 
    3. & \verifier & \to & \tpm: & \mathsf{accept}\slash\mathsf{reject}
  \end{array}$
  
  \caption{DAA Sign}
  \label{fig:daa-sign}
\end{figure*}
%$

  Similarly to Examples~\ref{ex:DAA},\ref{ex:DAA-s}, we distinguish two cases whether $sk_I$ is
  considered as a private constant or as an identity parameter.
  We recall that this choice critically impacts the privacy property that
  is modeled. Indeed, the privacy set~\cite{pfitzmann2001anonymity} is considered to be
  (a) the set of users who obtained a credential from a given issuer in the former case, or,
  (b) the set of all users in the latter case.
\smallskip{}

\noindent{\bf (a) $sk_I$ as a private constant.}
This 2-party protocol falls in our class and lies in the non-shared case.
Indeed, we model infinitely many different TPMs that may take part to the DAA sign protocol
with any verifier whose role is always the same (he has no proper identity).
We automatically analysed this protocol with \ukano
and established both frame opacity and well-authentication in less than 4 seconds.
Frame opacity has been established using a well-chosen idealisation adapted from the syntaxic heuristic.
\smallskip{}

\noindent{\bf (b) $sk_I$ as an identity parameter.}
  This 2-party protocol falls in our class and lies in the shared case.
  Indeed, we model infinitely many different TPMs with credentials signed by pairwise different issuers that may take part to the
  DAA sign protocol with a verifier who is checking credential from the corresponding
  issuer\footnote{We discuss a more precise modelling and why it cannot be analyzed in our framework in~\ref{subsec:class}.}.
  We automatically analysed this protocol with \ukano
  and we found that frame opacity is violated for any of \ukano heuristic
  (note that the idealisation for the case (a) is not conform for (b)).
  By inspecting the attack trace returned by \ukano, one can quickly
  rebuild
  an attack against unlinkability and anonymity.
  Indeed, the attack on frame opacity shows that an attacker can exploit the fact that
  the ZK proof contains in its public part the public
  key of the issuer (\ie $\pk(sk_I))$). A passive eavesdropper is thus able
  to learn the issuer that has signed the credential used in a ZK proof sent by a prover,
  hence breaking anonymity and unlinkability.
  This is not surprising as the privacy mechanism of DAA sign was intended to
  protect users' privacy inside a certain group (associated with an issuer),
  which is a property we have checked, and which holds, with the variant (a).

%\lum{(TABLE 1): pas de ligne pour LAK (stateless) car pas de fichier ProVerif qui capture l'attaque (demande de modéliser XOR, pe %une facon d'approximer qui permet de trouver l'attaque ?)}

\subsubsection{Summary}
We now summarise our results in \Cref{tab:case:summary}.
We only summarize results obtained regarding unlinkability and
considering concurrent sessions.
For each protocol, we mention the identity parameters of each role.
Most of our case studies fall into the shared case with $\fn(\ini)
\cap \fn(\res) \neq \emptyset$. 
We indicate the verification time in seconds to verify both conditions.
When there is an attack, we give the time \proverif takes to show that one of the condition fails to hold.
We note  \verif~for a condition automatically checked using our tool \ukano~and \nope~when the condition does not hold.
Note that all positive results were established automatically using
our tool \ukano~(which is based on \proverif)
without any manual effort (except for the cases indicated by \verif$^*$
for which little manual efforts were needed).

% \begin{table}[ht]
%   \centering
%   \begin{tabular}{l|ccc|c}
%  \multirow{2}{*}{Protocol \rule[-5mm]{0pt}{1cm} }     & Frame% Frame Opacity
%                   & Well- % Well-Authentication &
%                   &   \multirow{2}{*}{Unlinkability \rule[-5mm]{0pt}{1cm} }
%                   &   Verification\\
% & opacity & \; auth. \;&& time\\
%  \hline
%         Hash-Lock & \verif & \verif & \holds & < 1s \\
%      %  LAK (stateless) & $-$  & \nope & \attaque & {$-$ ??}\\
%         Fixed LAK & \verif & \verif & \holds & < 1s\\ \hline
%         BAC       & \verif & \verif & \holds & < 1s \\
%         BAC\slash PA\slash AA & \verif & \verif & \holds& < 1s \\ 
%         BAC\slash AA\slash PA & \verif & \verif & \holds& < 1s \\ \hline
%         PACE (faillible dec) &  $-$ & \nope & \attaque & < 30s \\
%         PACE (as in~\cite{bender2009security})
%              &  $-$ & \nope & \attaque & < 1m \\
%         PACE & $-$ & \nope & \attaque & < 2m\\
%         PACE with tags & \verif & \verif & \holds & < 2m\\
%         \hline
%         DAA sign &  \verif & \verif & \holds & < 5s\\
%         DAA join & \verif & \ \ \verif$^*$ & \holds & < 5s\\
%         ABCDH (irma) & \verif & \ \ \verif$^*$ & \holds & < 3h\\
%   \end{tabular}
%   \caption{Summary of our case studies (Section~\ref{sec:casestudies})}
%   \label{tab:case:summary}
% \end{table}

\newcommand{\ra}[1]{\renewcommand{\arraystretch}{#1}}
\renewcommand{\sp}{\phantom{a}}
\newdimen\POV
\newdimen\PART
\newdimen\DIR
\newdimen\CELL
\POV=0.15em
\PART=0.1em
\DIR=0.03em
\CELL=30pt
% \heavyrulewidth=.08em   % \lightrulewidth=0.05em   % \cmidrulewidth=0.03em   % \belowrulesep=.65ex   % \aboverulesep=.4ex

\begin{table}[ht]
  \centering
  \ra{1}
  \begin{tabular}{@{}l%|
    cc%|
    ccc%|
    c@{}}
    \toprule \midrule
    \multirow{2}{*}{Protocol \rule[-5mm]{0pt}{1cm}} &
    \multicolumn{2}{c}{Identity parameters} & \multirow{2}{*}{\parbox{\CELL}{Frame opacity} \rule[-5mm]{0pt}{1cm}} % Frame Opacity
    & \multirow{2}{*}{\parbox{\CELL}{\; Well-auth.} \rule[-5mm]{0pt}{0.92cm}} % Well-Authentication &
    & \multirow{2}{*}{Unlink. \rule[-5mm]{0pt}{1cm} }
    & \multirow{2}{*}{\parbox{1.5\CELL}{Verification \ \quad\null\hfill time\hfill\null} \rule[-5mm]{0pt}{1cm}}\\
    \cmidrule[\DIR](r){2-3}&&&&\\[-10pt]
    & in role $\ini$ & in role $\res$& &&& \\    
    \cmidrule[\POV]{1-7}
    Hash-Lock & $k$ & $k$&\verif & \verif & \holds & < 1s \\[1mm]
    Fixed LAK & $k$ &$k$ & \verif & \verif & \holds & < 1s\\
    \cmidrule[\DIR](r){1-7}
    BAC     & $k_E, k_M$ &  $k_E, k_M$ & \verif & \verif & \holds & < 1s \\[1mm]
    BAC\slash PA\slash AA   & $k_E, k_M$ &  $k_E, k_M$ &\verif & \verif & \holds& < 1s \\[1mm]
    BAC\slash AA\slash PA   & $k_E, k_M$ &  $k_E, k_M$ &\verif &
                                                                 \verif & \holds& < 1s \\
    \cmidrule[\DIR](r){1-7}
    PACE (faillible dec) &  $k$ & $k$ & $-$ & \nope & \attaque & < 30s \\[1mm]
        PACE (as in~\cite{bender2009security})
             &     $k$ & $k$ &$-$ & \nope & \attaque & < 1m \\[1mm]
        PACE   &  $k$ & $k$& $-$ & \nope & \attaque & < 2m\\[1mm]
        PACE with tags  &  $k$ & $k$ &\verif & \verif & \holds & <
                                                                 2m\\
    \cmidrule[\DIR](r){1-7}

    ABCDH (irma) & $\textit{user}_V$ & $\textit{user}_C$ &\verif & \ \ \verif$^*$ & \holds & < 3h\\
    \cmidrule[\DIR](r){1-7}
        DAA join & $\textrm{DAAseed}, sk_\TPM$ & $sk_\TPM$&\verif & \ \ \verif$^*$ & \holds & < 5s\\[1mm]
    DAA sign (a) &  $\emptyset$& $\textrm{DAAseed}, cnt, r$ &\verif$^*$ & \verif & \holds & < 5s\\[1mm]
    DAA sign (b)&  $sk_I$& $sk_I, \textrm{DAAseed}, cnt, r$ &\nope & \verif & \attaque & < 1s\\
    
        % \hline
    \midrule \bottomrule\\
  \end{tabular}
  \caption{Summary of our case studies regarding unlinkability with
    concurrent sessions}
  \label{tab:case:summary}
\end{table}

\newcommand{\sid}{\mathit{sid}}

\section{Limitations of our approach}
\label{sec:limitations}

In this section, we would like to discuss some further limitations of our approach. We first explain some 
limitations that come from the approach itself %\lucca{and our
                                %theoretical result} % S.D. pour moi
                                %approach itself = theoretical result
in Section~\ref{subsec:theory}.
In Section~\ref{subsec:limitations-tool}, we then discuss 
some limitations of our tool \ukano which inherits some of the limitations of 
the \proverif tool on which it is based.

%%
%% LIMITATION THEORIQUE
%%

\subsection{Limitations of our Theorem~\ref{theo:main}}
\label{subsec:theory}

Our approach consists of providing two sufficient conditions under
which anonymity (see Definition~\ref{def:anonymity}) 
and unlinkability (see Definition~\ref{def:un}) 
are satisfied.  
These condtions, even if they are satisfied by many concrete examples,
may not be fullfilled by some protocols that are nevertheless
anonymous and unlinkable (see examples described in Section~\ref{subsec:theory:tightness}).
We then discuss in Section~\ref{subsec:class} some limitations that
come from the class of protocols we consider. 
%In particular, we reconsider the PACE protocol 
%to highlight a rather subtle limitation of our approach, and, more specifically, of our notion of honest execution.

\subsubsection{Tightness of our conditions}
\label{subsec:theory:tightness}

As illustrated  by the toy
protocols given in \Cref{ex:too-strong-1,ex:too-strong-2},
our conditions are sufficient but not necessary
to ensure unlinkability or anonymity.

\begin{example}
\label{ex:too-strong-1}
We suppose that, initially, 
each tag  has its own key~$k$ and the reader maintains a database containing
those keys. The protocol relies on symmetric encryption, and 
can be informally
described as follows.
$$
\begin{array}{rcll}
\montag & \to & \reader: & \{\id\}_k\\
\end{array}$$

Once the reader receives the encryption, it opens it and checks  the identity of the tag before accepting (or not) 
to  grant access to the tag.
This protocol falls into our generic class of $2$-party protocols (shared case).
Anonymity w.r.t.~$\id$ is satisfied but unlinkability is not: a given tag always sends the same message.

Regarding our conditions, frame opacity does \emph{not} hold.
Consider
$$
\phi = % \new, \id_1, \id_2, k_1, k_2.
\{w_1 \mapsto \{id_1\}_{k_1}, \; w_2 \mapsto \{id_1\}_{k_1}, \; w_3 \mapsto \{id_2\}_{k_2}\}.
$$ 
Such a frame can be obtained when executing the $\mathcal{M}_\Pi$ process. The syntactical idealisation will rename  both occurrences 
of $id_1$ (resp.~$k_1$) using different names whereas 
the semantical  idealisation will idealise each output using a fresh names. In both cases,  the resulting  
idealised frame is not statically equivalent to $\phi$. Thus, frame opacity cannot be established using these idealisations.
Actually, no idealisation will be able to idealise the two first  outputs in the same way  and the two last outputs  in different way at the same time.
This illustrates that frame opacity is a too strong condition when considering anonymity.

Regarding well-authentication, we can establish that such a condition does \emph{not} hold as well.
Still considering the execution leading to the frame $\phi$ above, we can then consider a reader that starts two sessions accepting twice $w_1$ as an input.
It will then continue by executing its conditionals positively. The annotations of these two 
conditionals will be respectively $R(\{k_1, \id_1\}, \sid)$ and $R(\{k_1,\id_1\}, \sid')$.
Therefore, condition (ii) of Definition~\ref{condi:auth} is not satisfied. These two conditionals are not safe and they are both associated to the same annotation (the one carried out by the output $w_1$).
This does not break anonymity but simply shows that replaying messages is a scenario that allows an attacker to fool one party (here the reader) up to some point (here until the end).
\end{example}

\begin{example}
\label{ex:too-strong-2}
In order to ensure unlinkability, we now suppose that the tag sends its identity accompanied with a freshly generated random number $r$. Therefore, we have that:
$$
\begin{array}{rcll}
\montag & \to & \reader: & \{\langle r, \id\rangle\}_k.\\
\end{array}$$

This protocol falls into our generic class of $2$-party protocols (shared case).  The identity parameters of both roles are $\id$ and $k$ whereas $r$ is the session parameter of role $\ini$.
As in the previous example, anonymity w.r.t. $\id$ holds. Unlinkability should hold, assuming that the reader does not output any message indicating whether the test has been passed with success or not.

Actually, frame opacity can be established relying on either the syntaxical idealisation or the semantical one. The fresh random number inside 
each encryption allows one to ensure that all the ciphertexts are different. However, for the same reason as the one explained in the previous example, well-authentication does not hold: condition (ii)
is not satisfied. 

{We recall that this can be considered as a false attack only if the protocol and the use case both enforce that
  the continuation of the protocol in case the test passes is always indistinguishable from
  the continuation in the other case; this is a strong assumption.}
\end{example}

\subsubsection{Class of protocols}
\label{subsec:class}

Among the limitations coming from our definition of protocols, we
first reconsider the DAA sign and PACE protocols 
to highlight some limitations of our approach.
%a rather subtle limitation of our approach, and, more specifically, of our notion of honest execution.}

%\noindent
\paragraph*{Two parties only.}
Our notion of protocols only covers 2-party protocols. This obviously excludes important protocols with more than
2 parties such as secure group communication
protocols~\cite{meadows2003formal}, e-voting
protocols~\cite{DRS-ifiptm08},
make
mobile communication protocols~\cite{basin2018formal},
the combination of DAA join and DAA sign~\cite{DRS-ifiptm08} that features 3 parties, \etc
This also excludes scenarios where privacy is considered between {\em group of entities}.
For instance for DAA sign (see Example~\ref{ex:DAA} or Section~\ref{sec:caseStudies:DAA}),
verifiers and clients may be associated to different issuers.
Using our framework, one can analyze privacy between users in a single group (as in Example~\ref{ex:DAA}),
or between groups but where each goup has only one user (as in Example~\ref{ex:DAA-s}).
In the latter case, one would rather want to model privacy between groups where each group contains an unbounded number of users,
but this is out of the scope of our approach. This is not surprinsing since such a scenario actually features
three parties: clients, verifiers, issuers (forming groups); all with unbounded number of entities.
% Note however, that such a scenario contains 
\smallskip{}

% \noindent
\paragraph*{Honest trace.}
Now, we want to report on a potential limitation we discovered when analysing the PACE protocol using \tamarin.
Our initial aim was to investigate the scenario where sessions can be executed
only sequentially. We have turned to \tamarin since \proverif is not able
to faithfully model such scenarios. We wrote a \tamarin model encoding
well-authentication and found surprisingly that this condition does not hold, even
with the tagged version.
This contrasts with the positive result obtained with \proverif.
Actually, this comes from the fact that \tamarin models Diffie-Hellman exponentiation
in a more faithful way than \proverif. Some behaviours that were not possible in the \proverif model become possible,
and it happens that well-authentication is not satisfied in such a model.

Indeed, the attacker can alter the Diffie-Hellman shares, as informally
depicted in \Cref{fig:attack:pace:dh}, without impacting the successive 
conditionals.
\begin{figure}[ht]
  \centering
  \newcommand{\attacker}{\mathsf{Attacker}}
$ \begin{array}{lrcll}
    1.&\montag& \to &\reader : & \{s_T\}_{k}\\
    2. &\reader & \to & \attacker: & g^{n_R}\\
    2'. &\attacker & \to & \montag: & (g^{n_R})^X\\
    3. &\montag & \to & \attacker: & g^{n_T}\\
    3'. &\attacker & \to & \reader: & (g^{n_T})^X\\
\end{array}$
  \caption{Example of successful but dishonest interaction ($X$ can be any message)}
  \label{fig:attack:pace:dh}
\end{figure}

This is problematic because successful tests will pass (independently of the message $X$)
while such interactions are not honest according to our current
definition of honest trace (see \Cref{def:honest}). 
This problematic interaction is however not detected in \proverif,
due to the lack of equations in the underlying
  equational theory: the final keys computed by both parties will
  be different, $((g^{n_R})^X)^{n_T}$ for the tag and
  $((g^{n_T})^X)^{n_R}$ for the reader.
Therefore such an interaction cannot be completed sucessfully, and well-authentication will
  be established.
%Obviously, this
%interaction is not an attack on unlinkability at all.

  The failure of well-authentication described above does not yield
  a failure of unlinkability. It is thus a case where one might want to
  make well-authentication less restrictive. One direction for
  weakening it is to
extend the notion of  ``honest trace associated to a protocol''
(\Cref{def:proto}):
instead of a single honest trace
we would associate to a protocol a set
of symbolic traces that are, roughly,
traces with (possibly) variables in recipes.
For PACE, one may for instance use
$\tr_h = \Out(c_I,w_1).\In(c_R,\dihe(w_1,X)).\Out(c_R,w_2).\In(c_I,\dihe(w_2,X)).\ldots$
in addition to the standard
trace.
However, in order to adapt our proof technique, we need to make sure that
whatever the recipes chosen to fill in the variables (\eg $X$ in $\tr_h$),
the resulting concrete trace can be executed by the protocol
and the produced frame always has the same idealisation.
Remark that this is the case for $\tr_h$ in the case of the PACE protocol.
% We \david{conjecture} that our theorem and our method could be adapted in this way in order
% to deal with more protocols.
% PACE and other examples based on such Diffie-Hellman
% exchanges, but this is not covered by our current result.

In practice, this limitation on the well-authentication condition
does not seem very important, as the issue is tied to the
peculiar use of two Diffie-Hellman rounds in PACE and,
expanding the discussion beyond privacy,
the attack on well-authentication shows a potential weakness in that protocol.
    Indeed, Tag and Reader fail to establish an agreement on each other's Diffie-Hellman
    shares from the first round (\ie $g^{n_R}$ and  $g^{n_T}$).
    Therefore, an attacker is able to manipulate those shares without being detected,
    which goes against best practices in protocol design.
    In contrast, the MAC messages 8 and 9 (see \Cref{fig:pace:AB}) allow the 
    Tag and Reader to agree on each other's
    Diffie-Hellman shares from the second round (\ie $G^{n_T'}$ and $G^{n_R'}$)
    and on the shared key resulting from the first round (\ie $g^{n_R n_T}$ in $G$).
    Adding $g^{n_T}$ (respectively $g^{n_R}$) to the first (respectively second) MAC message would fix this lack of agreement.
    We have formally verified with \tamarin that PACE with this modification indeed satisfies well-authentication, thus providing an agreement on the full protocol transcript.

\smallskip{}

\paragraph*{Stateless only.}
Our framework and our theorem only applies to stateless protocols (\ie no mutable states persistent across sessions).
This immediately excludes numerous real-world protocols such as
secure messaging protocols~\cite{cohn2017formal},
mobile communication protocols~\cite{basin2018formal}, \etc

% \discu{TODO:
%   Pour Example~\ref{DAA-s}, on ne modèle qu'un user ``par issuer''. Si on voulait être plus précis ici, on pourrait modéliser:}
% $$\discu{
%   \mathcal{M}_{\Pi_\DAA^{\skIssuer}}= {!}\;\new (\skIssuer). \; \big(
%   {!}\; \new n_V. P_V
%   \;\mid \;
%   {!}\; \new k_C,\monid_C \; {!}\;\new (n_C,m). P_C\big)}$$
% \discu{
%   et vérifier UK entre users ($\neq k_C$)
%   ce qui n'est pas couvert par nos définitions et notre approche.
%   A voir si la def ci-dessus est vraiment intéressante. Dans ce cas, peut-être
%   que l'on veut en parler comme une limitation à la fin ?}

%%   
%% UKANO
%%
\subsection{Limitations of our tool \ukano}
\label{subsec:limitations-tool}
Our tool \ukano  also suffers from some limitations. 
In particular, the heuristics we propose to build idealisation could be improved. 
For instance, in case of DAA sign, we were unable to establish frame opacity fully automatically: we had to propose 
a well-chosen idealisation adapted from the syntaxic heuristic. Note that, in general the syntaxic heuristic is a good choice but yields large messages that may cause some efficiency issues to \proverif. The semantical heuristic is less precise but much more efficient. In case of DAA sign, we make a trade-off betweeen these two choices to establish frame opacity.
Regarding well-authentication, the resulting queries happen to be quite big and may also cause some troubles to \proverif.
This can be adressed by soundly splitting the query (see Section~\ref{sec:irma}) but this feature has not been implemented in \ukano.

Besides the limitations of our tool \ukano, we also inherit some of the limitations of the backend tool, i.e. \proverif.
For instance, in terms of cryptographic primitives, we have seen that \proverif only consider a very abstract model 
for modular exponentiation, and does not allow one to model all the algebraic properties of the exclusive-or operator.
We are also unable to faithfully model scenarios involving sequential compositions. This feature could be modeled in \proverif 
relying on private chanels but abstractions performed by \proverif when modeling private channels will not allow us to benefit from the extra information
of that encoding.
Of course, this limits the scope of our approach but progress made on existing 
verification tools will directly benefit to our approach as well. In 
particular, a natural extension for our tool \ukano would be to consider 
\tamarin as a backend. This would bring precise support for sequential composition, 
a more faithful model for Diffie-Hellman exchange, and also the
recent extension to deal with the exclusive-or operator~\cite{dreier2018automated}.

\section{Conclusion}
\label{sec:conclusion}

We have identified two conditions, namely well-authentication and
frame opacity, which imply anonymity and unlinkability
for a wide class of protocols.
Additionally, we have shown that these two conditions
can be checked automatically using the tool \proverif, and we have
mechanised their verification in our tool \ukano.
This yields a new verification technique to check anonymity and unlinkability for an
unbounded number of sessions. It has proved quite effective
on various case studies. In particular, it has brought first-time
unlinkability proofs for the BAC protocol (e-passport)
and ABCDH protocol.
Our case studies also illustrated that our methodology is
useful to discover attacks against unlinkability and anonymity
as illustrated by the new attacks we found on PACE and LAK.

%\smallskip{}

In the future, we plan to improve the way our sufficient conditions are checked.
For instance, we would like to let \ukano build idealisations in a cleverer way; notably in the choice
of heuristics to adopt, and  we are interested in other tools we may
leverage as back-ends (\eg \tamarin).
We are also interested in simplifying further the verification of our 
conditions towards completely reducing the verification of equivalence-based 
properties to pure reachability verifications, which are known to be
much simpler.
Specifically,
we believe that frame opacity could be verified via reachability and 
syntactical checks only.
Obtaining such a result 
would certainly be useful, as it would allow
us to use a richer toolset to verify case studies.

%\smallskip{}
%%% en théorie, extension
Based on limitations discussed in Section~\ref{subsec:class},
we also identify a number of research problems aimed at generalizing
the impact of our technique.
We would like to investigate the extension of our main theorem
to the case of protocols with states. This is certainly technically 
challenging, but would make it possible to model more protocols, or
at least model them more faithfully. We are also interested in extending
our method to protocols with more than 2 parties which are commonplace
(\eg the combination of DAA join and DAA sign is essentially a 3-party protocol).

%\smallskip{}

Finally, we believe that the overall methodology developed in this paper
(\ie privacy via sufficient conditions) could be applied in other contexts where privacy is
critical: \eg e-voting, attribute-based credentials, blockchain technologies,
 transparent certificate authorities. 
In our opinion, the privacy via sufficient conditions approach
also sheds light on the privacy notions themselves. Indeed, each
sufficient condition helps to get a better grasp of necessary ingredients
for preserving privacy.
It might thus be interesting to translate such conditions into more comprehensive guidelines
helping the design of new privacy-enhancing protocols.

\paragraph*{Acknowledgement.} We would like to thank Bruno Blanchet
for his valuable help regarding the mechanisation in ProVerif of our frame opacity
condition.
The extension of bi-processes in Section~\ref{sec:mecha-fo} is due to him.
We also thank Solène Moreau for her useful feedback on earlier versions
of this paper.

\bibliographystyle{abbrv}
\bibliography{biblio.bib}

\appendix
\label{app:proofs}
%\section{Proofs}
%\label{sec:proofs}
We provide in this appendix the proof of our main result (Theorem~\ref{theo:main}).
Our main argument consists in showing that,
for any execution of $\pMa$, there is an indistinguishable execution
of $\pS$ (the other direction being easy).
This indistinguishable execution will be obtained by a modification of the involved
agents. We will proceed via a renaming of agents
applied to an abstraction of the given execution of~$\pMa$.
We then prove that the renamed executions can still be executed and produce
an indistinguishable frame.

\smallskip{}

We fix a protocol $\Pi = (\vect{k}, \nI, \nR, \dag_I,
\dag_R, \ini, \res)$, some identity names $\vect\id$ and some fresh constants~$\vect\idzero$ yielding a process $\pMa$ as defined in
\Cref{sec:properties}.
We denote ${\idzero}_i$ (resp $\id_i$, $k_i$, and $k_i'$) the  
$i$\textsuperscript{th} element of the sequence $\vect\idzero$
(resp. $\vect\id$, $\vect k$, and $\vect{k'}$).
The construction of the proof will slightly differ depending
on $\dagI,\dagR$ (sequential vs concurrent sessions)
and whether $\fn(\ini)\cap\fn(\res)=\emptyset$ or not (non-shared case
vs shared case).

%%%%%%%%%%%%%%%%%%%
%
% Ground configurations
%
%%%%%%%%%%%%%%%%%%%%
\section{Abstraction of configurations} %: sequences of annotations, ground configurations and renamings}

Instead of working with $\pMa$, $\pM$, and $\pS$, it will be more
convenient to work with {\em ground configurations}.
Intuitively, we will associate to each execution of 
$\pMa$, $\pM$, and $\pS$, a ground configuration
that contains all agents involved in that execution,
already correctly instantiated.
By doing so, we are able to get rid of technical details such as
unfolding replications and repetitions a necessary number of times,
create necessary identity and session parameters, \etc
These ground configurations are generated from sequences of annotations 
satisfying some requirements.

\subsection{Sequences of annotations}

The sequence of annotations from which we will build ground configurations
shall satisfy some requirements that we list below. Essentially, the goal
is to make sure that no freshness condition over session and identity names
is violated.

\begin{definition}
\label{def:well-formed}
A sequence $S$ of annotations is \emph{well-formed} if the following
conditions hold.
\begin{itemize}
\item In all annotations $A(\vect {k'}, \vect {n'})$ with $A \in \{I, R\}$, the session parameters
  $\vect {n'}$ are names,
  and the identity parameters $\vect {k'}$ are made
  of names or constants $\vect \idzero$.
  We also have that $|\vect {n'}|=|\vect {n_A}|$ and
  $|\vect {k'}|=|\vect k|$. 
%(remind that $\vect {n_I}, \vect {n_R}$ and $\vect k$ are session and identity names
%  from $\Pi$).  
  Moreover, when $\vect \idzero\cap\vect{k'}\neq\emptyset$, 
  $k'_i={\idzero}_j$ if, and only if, $k_i=\id_j$.
\item No name appears both as identity and session parameter in any two annotations.
\item Two different annotations never share a session parameter.
\item Two annotations either have the same identity parameters, or do
  not share any identity parameter at all.
\end{itemize}
We say that a sequence of annotations $S$ is {\em single-session} when
two different annotations of the same role never share an identity parameter
and 
no annotation contains a constant in $\vect{\idzero}$.
\end{definition}

We straightforwardly lift those definitions to annotated traces by 
only keeping the underlying sequence of annotations and dropping actions.
A well-formed (resp. well-formed, single-session) sequence of annotations contains
annotations of agents that can be instantiated
from~$\pMa$ (resp.~$\pS$). Conversely,
we may note that given an annotated trace~$\ta$ such that $\pMa \lrstep{\ta} K'$, we have
that  $\ta$ is well-formed.

\subsection{Ground configurations}
After the introduction of some notations,
we explain how {\em ground configurations} are obtained from well-formed sequences of annotations.
Given a well-formed sequence of annotations $S$,
and some role $A \in \{I,R\}$, we introduce the following notations:
\begin{itemize}
\item $\mathrm{id}_A(S)$ is the set of identity parameters of agents of role~$A$
  occurring in $S$, \ie
$$\mathrm{id}_A(S) = \{\vect{k}\ |\ A(\vect k,\vect n)\in S\}.$$
\item for $\vect{l}\in\mathrm{id}_A(S)$, we note 
$\mathrm{sess}_A(S,\vect{l})$ the set of session parameters of agents of role $A$ and identity
  parameters $\vect{l}$ occurring in $S$, \ie
$\mathrm{sess}_A(S,\vect{l}) = \{\vect{n}\ |\ A(\vect l,\vect n)\in S\}.$
\item for $\vect{l}\in\mathrm{id}_A(S)$,
  we note 
$\mathrm{sess}^\mathsf{seq}_A(S,\vect{l})$ the sequence made
  of elements from $\mathrm{sess}_A(S, \vect{l})$,
  without repetition, in order of first occurrence in $S$.
\end{itemize}

Finally, for some sequence of elements $L=e_1,e_2,e_3,\ldots$ and a
process $P(e)$ parametrized by $e$, we denote by
$\coprod_{e\in L} P(e)$ the process
$P(e_1);(P(e_2);(P(e_3);\ldots))$.

\begin{definition}
Let $S$ be a well-formed sequence of annotations.
The {\em ground configuration associated to~$S$}, denoted by $\K(S)$, is
the multiset $\p_I \sqcup \p_R$ where $\p_I$
is defined as follows depending on~$\dagI$:
\begin{itemize}
\item if $\dagI = \,\rep$
then 
$
\p_I = \Big\{\;
        \bigl(\ini\{\vect{k}\mapsto\vect{l},\vect{n}_I\mapsto\vect{m}\}\bigr)\annot{I(\vect{l},\vect{m})}%
        \;\;|\;\;
    I(\vect{l},\vect{m})\in S 
            \;\Big\};
$
\item if $\dagI=\rec$ then
$
\p_I = \left\{\;
    \coprod_{\vect{m}\in S_{\vect{l}}}
\bigl(\ini\{\vect{k}\mapsto\vect{l},\vect{n}_I\mapsto\vect{m}\}\bigr)\annot{I(\vect{l},\vect{m})}
    \;\;|\;\;
    \vect{l}\in\mathrm{id}_I(S) \text{ and }
    S_{\vect{l}}=\mathrm{sess}_I^\mathsf{seq}(S,\vect{l})
    \;\right\}.
$
\end{itemize}
The multiset $\p_R$ is computed in a similar way,
  replacing $\ini$, $\vect n_I$ and $I$ by
  $\res$, $\vect n_R$ and $R$ respectively.
\end{definition}

\newcommand{\toy}{\mathsf{toy}}

\begin{example}
Consider the following toy protocol
$\Pi_\toy \defe (k, n_I, n_R, \rec, \rec, \ini, \res)$
where $\ini=\Out(c_I,\senc(n_I,k))$
and $\res=\In(c_R,x)$.
We have that
$$
\mathcal{M}_{\Pi_\toy} = !\;\new k.(\rec\;\new n_I.\ini \ |\ \rec\;\new n_R.\res)
\lrstep{\ta}
(\q;\phi)
$$
for
$\ta=\tau.\tau.\tau.\tau.\tau.\;\ell:\Out(c_I,w_0)\annot{I(k,n_I)}.
                         \tau.\tau.\tau.\;\ell:\Out(c_I,w_1)\annot{I(k,n_I')}.$
The ground configuration associated to $\ta$ is the following multiset
with one element:
$$\K(\ta)= \{
(\Out(c_I,\senc(n_I,k))\annot{I(k,n_I)});
(\Out(c_I,\senc(n_I',k))\annot{I(k,n_I')})\}.$$
Note that $\K(\ta)$ is also able to produce the annotated trace $\ta$
up to some $\tau$ actions.
\end{example}

We lift those definitions to annotated traces as before.
A ground configuration associated to a well-formed sequence of annotations is essentially
an ``unfolding'' of $\pMa$.
Therefore, there is a strong relationship between the original process and the one
obtained through $\K(\cdot)$ as established in the following
proposition.

\begin{proposition}
\label{prop:gc-pm}
Let $\ta$ be a well-formed annotated trace. We have that:
\begin{enumerate}
\item[(1)] If $\pMa \LRstep{\ta} K$ (resp. $\pS \LRstep{\ta} K$), then 
$\K(\ta) \LRstep{\ta} K'$ for some $K'$ such that $\phi(K) = \phi(K')$.
\item[(2)] If $\K(\ta) \LRstep{\ta} K$, then $\pMa \LRstep{\ta} K'$ for
  some $K'$ such that $\phi(K) =\phi(K')$.
\item[(3)] If $\K(\ta) \LRstep{\ta} K$ and $\ta$ is single-session,
then $\pS \LRstep{\ta} K'$ for some $K'$ such that $\phi(K) =\phi(K')$.
\item[(4)] If $\ta=\ta_1.\ta_2$ and $\K(\ta_1.\ta_2)\lrstep{\ta_1} K$ then
  $\K(\ta_1)\lrstep{\ta_1}K'$ for some $K'$ such that $\phi(K) =\phi(K')$.
\end{enumerate}
\end{proposition}
\begin{proof}
  Item (1) holds by construction of the operator $\K(\cdot)$, which has been built
  by closely mimicking how $\pMa$ and $\pS$ create agents. We thus have that when an agent
  is at top-level in a configuration in the execution of $\pMa$ (or $\pS$) then it is also
  available in the execution of $\K(\ta)$.
  In general, less $\tau$ actions are necessary for the execution starting
  with $\K(\ta)$ than for the executions starting with $\pMa$ or $\pS$.
  Indeed, there is no need, in ground configurations,
  to spawn agents by unfolding replications or repetitions or
  creating fresh names.
  This is because agents are (more) immediately available in $\K(\ta)$.
  
  Item (2) heavily relies on the well-formedness of $\ta$. One can thus prove that
  all agents in $\K(\ta)$ can be created along the execution by choosing appropriate names
  when triggering rules \textsc{New}. For instance, the first item of \Cref{def:well-formed} makes sure that the arity
  of parameters in agents matches the number of names to be created.
  The second and third items of \Cref{def:well-formed} ensure that the freshness guard conditions of the rule \textsc{New}
  holds for names to be created.
  Finally, the fourth item of \Cref{def:well-formed} implies that when an agent $a=A(\vect k,\vect n)$ must be
  created then either (i) names in $\vect{k}$ are completely fresh and this identity~$\vect{k}$ can be created from $\pMa$ by unfolding~$!$ and create names $\vect{k}$
  or (ii) names in $\vect{k}$ have already been created and thus the agent $a$ can be created
  from the last replicated process used to create identity $\vect{k}$ in the first place.

  Item (3) is similar to (2). The single-session hypothesis provides exactly
  what is needed to mimic the execution using $\pS$ rather than $\pMa$.
    
  Finally, item (4) stems from a simple observation.
  Compared to $\K(\ta_1)$, the multiset of processes $\K(\ta_1.\ta_2)$
  adds processes in parallel and in sequence after some processes
  of $\K(\ta_1)$. However, these extra processes are unused
  when executing $\ta_1$, thus $\K(\ta_1)$ can perform the same
  execution.
\end{proof}

%%
%%
%% Renamings of Annotations
%%

\subsection{Renamings of annotations}
As mentioned before, we shall prove that for any execution of $\pMa$,
there is an indistinguishable execution of $\pS$.
This indistinguishable execution that $\pS$ can perform
will be obtained by a renaming of annotations.
We define next a generic notion of such renamings of annotations.
However, the crux of the final proof is to find a {\em good} renaming
that implies: (i) the executability by $\pS$ of the renamed trace, and
(ii) the static indistinguishability of the resulting frames (before and after
the renaming).

\begin{definition}
\label{def:renaming}
%\com{Précisions sur totalité?}
A {\em renaming of annotations} (denoted by $\rho$) is an injective mapping from
annotations to annotations such that:
\begin{itemize}
\item for any well-formed sequence of annotations $S$, $S\rho$ is well-formed;
\item $\rho$ is role-preserving:
\ie initiator (resp.\ responder) annotations are mapped to initiator
(resp.\ responder) annotations;
\item for any two annotations $a_1 = A_1(\vect{k_1}, \vect{n_1})$, $a_2 = 
  A_{2}(\vect{k_{2}}, \vect{n_{2}})$, if
    $\rho(a_1)$ and
    $\rho(a_2)$ have the same identity parameters, then  $\vect{k_1} =  \vect{k_2}$.
\end{itemize}
\end{definition}
The two first conditions are expected: renaming of annotations shall only modify session and identity parameters
whilst preserving well-formedness. The final condition ensures that renamings do not create more ``sequential dependencies''
between agents (\ie agents sharing the same identity and whose role
can execute sessions only sequentially): after renaming, less pairs of agents have same identity.

Next, we define $\ta\rho$ as the annotated trace obtained
from $\ta$ by applying $\rho$ to annotations only.
Note that, by definition of renamings, the resulting $\ta\rho$ is well-formed as well.

One can also define the effect of renamings on ground configurations.
If $\rho(\aagent(\vect{k\vphantom{'}},\vect{n\vphantom{'}}))
=\aagent(\vect{k'},\vect{n'})$,
the renaming $\sigma$ induced by $\rho$ on 
$\aagent(\vect{k\vphantom{'}},\vect{n\vphantom{'}})$
is the (injective) mapping such that
$\vect {k\vphantom{'}} \sigma = \vect{k'}$ and $\vect{n\vphantom{'}}\sigma = \vect{n'}$.
Given a ground configuration
$\p = \{
 \coprod_{j}
    P^i_j\annot{a^i_j}
\}_i$,
we define
$\p\rho = \{
 \coprod_{j}
    P^i_j\sigma^i_j\annot{\rho(a^i_j)}
 \}_i$
where $\sigma^i_j$ is the renaming induced by $\rho$ on $a^i_j$.
Note that the renaming on parameters induced by a renaming of annotations
may conflict: this happens, for example, when
$\rho(\aagent(\vect{k},\vect{n}))=\aagent(\vect{k_1},\vect{n})$
and
$\rho(\aagent(\vect{k},\vect{m}))=\aagent(\vect{k_2},\vect{m})$.

A renaming of annotations can break executability. 
Even when executability is
preserved, it is not obvious to relate processes before and after the 
renaming, as messages can be affected in complex ways and conditionals may not
evaluate to the same outcome. Fortunately, frame opacity and 
well-authentication will provide us with strong properties to reason
over executions, as seen in the next subsections.

\begin{example}
\label{ex:renaming-breaks-executability}
Consider the annotated trace $\ta$ from \Cref{ex:annotation}
that $\pMa$ can execute.
The ground configuration $\K(\ta)$ can execute it as well, using $\LRstep{}$.
We now define $\rho$ as follows: 
$\rho(\aini(k',n'_I))= \aini(k_1,n'_I)$ and
$\rho(\ares(k',n'_R))= \ares(k_2,n'_R)$ for some
fresh names $k_1,k_2$.
The trace $\ta\rho$ can no longer be executed by $\pMa$ nor by $\K(\ta)\rho$
  (even using $\LRstep{}$) because the first output sent
by $\ares(k_2,n'_R)$ (\ie $\senc(\langle {n'_I}, {n'_R}\rangle,k_2)$)
  will not be accepted by  $\aini(k_1,n'_I)$ since $k_1\neq k_2$.
\end{example}

% ====================================================================

\section{Control is determined by associations}
We show in that section that the outcome of tests is entirely determined by
associations. This will be useful to show that, if we modify
an execution (by renaming agents) while preserving enough
associations, then the control flow is left unchanged.

\begin{proposition} \label{prop:then-iff-assoc}
  We assume that $\Pi$ satisfies item (i) of the well-authentication
  condition.
  Let $\ta=\ta_0.\tau_x\annot{a_1}$ with $\tau_x\in\{\taut,\taue\}$
  be a well-formed annotated trace such that
  $$\K(\ta)\lrstep{\ta_0.\tau_x\annot{a_1}}(\p;\phi)$$
  and the last action (\ie $\tau_x\annot{a_1}$)
  is performed by an unsafe conditional.
  We have that $\tau_x=\taut$ if, and only if, 
  there exists $a_2\in\agents$ such that $a_1$ and $a_2$
  are associated in $(\ta_0,\phi)$.
\end{proposition}
\begin{proof}
%  For both directions, 
  ($\Rightarrow$) 
  We start by
  applying Proposition~\ref{prop:gc-pm} to obtain
  an execution $$\pMa\LRstep{\ta_0.\taut\annot{a_1}}(\p';\phi)
  \text{ and thus }
  \pMa\lrstep{\ta_0^*.\taut\annot{a_1}}(\p'';\phi).$$
  for some $\ta_0^*\upto\ta_0$.
  As a consequence of well-authentication, item (i) applied on the above execution,
  we obtain that for some $a_2\in\agents$, 
  $a_1$ and $a_2$ are associated in $(\ta_0^*,\phi)$.
  Since $\ta_0\upto\ta_0^*$, they are also associated in $(\ta_0,\phi)$.
  % C'est la même chose sans la condition de réciprocité.

  ($\Leftarrow$) For this other direction, we observe that (up to changes
  of recipes that do not affect the resulting messages) if two agents are 
  associated in the above execution (starting with $\K(\ta)$),
  then they are executing \emph{the} honest trace of $\Pi$
  modulo a renaming of parameters, thus the considered test must be successful.
  We thus assume that $a_1 = \aagent_1(\vect{k_1},\vect{n_1})$
  and $a_2 = \aagent_2(\vect{k_2},\vect{n_2})$ are associated in $(\ta_0,\phi)$
  we shall prove that $\tau_x = \taut$.
  By association, $\restrict{\ta_0}{a_1,a_2}$ is honest:
  its observable actions are of the form
  $\Out(c_1,w_1).\In(c'_1,M_1)
  \ldots \Out(c_n,w_n).\In(c'_n,M_n)$
  with possibly an extra output at the end, and are
  such that $M_i\phi\redc \theo w_i\phi$ for all $1\leq i \leq n$.
  Consider~$\ta'$ obtained from~$\ta_0$ by replacing
  each recipe $M_i$ by $w_i$.
  Since this change of recipes does not affect the resulting messages,
  the modified trace can still be executed by $\K(\ta)$ and
  yields the same configuration $(\p;\phi)$.
  But now $\restrict{\ta'}{a_1,a_2}$ is a self-contained execution,
  \ie if $P$ and $Q$ are the processes (possibly sub-processes) respectively annotated
  $a_1$ and $a_2$ in $\K(\ta)$, we have:
  $$(\{P\annot{a_1},Q\annot{a_2}\};\emptyset)
  \lrstep{\restrict{\ta'}{a_1,a_2}}
  (\{P'\annot{a_1},Q'\annot{a_2}\};\phi')
  \lrstep{\tau_x\annot{a_1}}
  (\{P''\annot{a_1},Q'\annot{a_2}\};\phi').$$
  In the shared case (\ie $\fn(\ini)\cap\fn(\res)\neq\emptyset$), by definition of association, the identity parameters of $a_1$ are equal to those of
  $a_2$. Otherwise, it holds that $\fn(\ini)\cap\fn(\res)=\emptyset$. In both cases, we thus have:
  $$
  \begin{array}{rcl}
  (\{\new\;\vect{k}.(\new\;\vect{n}_I.\ini\ |\ \new\;\vect{n}_R.\res)\};\emptyset) &\lrstep{\tau^*}&
  (\{P\annot{a_1},Q\annot{a_2}\};\emptyset)\\
  &\lrstep{\restrict{\ta'}{a_1,a_2}}&
  (\{P'\annot{a_1},Q'\annot{a_2}\};\phi')\\
  &\lrstep{\tau_x\annot{a_1}}&
  (\{P''\annot{a_1},Q'\annot{a_2}\};\phi').\\
  \end{array}
  $$
  In that execution, everything is deterministic (up to the equational
  theory) and thus
  the execution is actually a prefix of \emph{the} honest execution
  of $\Pi$ (from the process $P_\Pi$ defined in \Cref{def:proto}), up to a bijective renaming of parameters
  (note that $P$ and $Q$ do not share session parameters).
  Remind that all tests must be positive in the honest execution (\ie $\taue$
  does not occur in the honest execution).
  Therefore, $\tau_x=\taut$ concluding the proof.
\end{proof}

%%
%% INVARIANCE OF FRAME IDEALISATION
%%
\section{Invariance of frame idealisations}

In general, a renaming of annotations can break executability: as
illustrated in Example~\ref{ex:renaming-breaks-executability},
mapping two dual annotations to annotations with different 
identities breaks the ability of the two underlying agents to communicate successfully.
Moreover, even when executability is preserved, parameters change (so do names) and thus
frames are modified.
However, as stated next in Proposition~\ref{prop:invariance-id-op}, such renamings do not change idealised frames.
We obtain the latter
since we made sure that idealised frames only depend on what is already observable and not on
specific identity or session parameters.
In combination with frame opacity, this will
imply (\Cref{prop:invariant-idealisation}) that a renaming of annotations has no observable effect on the resulting 
\emph{real} frames.

\begin{proposition}
  \label{prop:invariance-id-op}
Let $\ta$ be an annotated trace such that $\ideaf(\ta)$ is
  well-defined. Let $\rho$ be a renaming of
annotations. We have that  $\ideaf(\ta)
\sim \ideaf(\ta\rho)$.
\end{proposition}

\begin{proof}
Let $\ta$ be an annotated trace such that $\ideaf(\ta)$ is
  well-defined. Let $\rho$ be a renaming of annotations.
Let $\fr_1$ be an arbitrary name assignment,  and $\fr_2$ be an injective function satisfying
   $\fr_2(a\rho,x) = \fr_1(a,x)$. First, we show, by induction on $\ta$, that $\ideaf^{\fr_1}(\ta)=\ideaf^{\fr_2}(\ta\rho)$.
The only interesting case is when 
  $\ta=\ta_0.(\ell:\Out(c,w)\annot{a})$. In such a case, we have that:
  $$\ideaf^{\fr_1}(\ta_0.(\ell:\Out(c,w)\annot{a}))=\ideaf^{\fr_1}(\ta_0)
  \cup\{w\mapsto \ideam{\ell}\sigma^\mathsf{i}_1\sigma^\mathsf{n}_1
  \redv\}$$
  with
  $\sigma^\mathsf{n}_1(x^\mathsf{n}_j)={\fr_1}(a,x^\mathsf{n}_j)$ and
  $\sigma^\mathsf{i}_1(x^\mathsf{i}_j)=R_j\ideaf^{\fr_1}(\ta_0)$ where
  $R_j$ is the $j$-th input of~$a$ in~$\ta_0$.
  The idealised frame
  $\ideaf^{\fr_2}(\ta\rho)$ is defined similarly, \ie
$$\ideaf^{\fr_2}(\ta_0\rho.(\ell:\Out(c,w)\annot{a\rho}))=\ideaf^{\fr_2}(\ta_0\rho)
  \cup\{w\mapsto \ideam{\ell}\sigma^\mathsf{i}_2\sigma^\mathsf{n}_2
  \redv\}$$
  with
  $\sigma^\mathsf{n}_2(x^\mathsf{n}_j)={\fr_2}(a\rho,x^\mathsf{n}_j)$ and
  $\sigma^\mathsf{i}_2(x^\mathsf{i}_j)=R_j^\rho\ideaf^{\fr_2}(\ta_0\rho)$ where
  $R_j^\rho$ is the $j$-th input of~$a\rho$ in~$\ta_0\rho$.
  By induction hypothesis we know that
  $\ideaf^{\fr_1}(\ta_0)=\ideaf^{\fr_2}(\ta_0\rho)$. Therefore, to
  conclude, it remains to show that
  $\ideam{\ell}\sigma^\mathsf{i}_1\sigma^\mathsf{n}_1=\ideam{\ell}\sigma^\mathsf{i}_2\sigma^\mathsf{n}_2$.
  Actually, we have that $R_j^\rho = R_j$, thus
  $\sigma^\mathsf{i}_1(x^\mathsf{i}_j)=\sigma^\mathsf{i}_2(x^\mathsf{i}_j)$,
  and
  $\sigma^\mathsf{n}_2(x^\mathsf{n}_j)={\fr_2}(a\rho,x^\mathsf{n}_j)=
  {\fr_1}(a,x^\mathsf{n}_j)=\sigma^\mathsf{n}_1(x^\mathsf{n}_j)$. 

We have shown that $\ideaf^{\fr_1}(\ta)=\ideaf^{\fr_2}(\ta\rho)$ and
relying on Proposition~\ref{prop:id-represent}, we easily deduce that $\ideaf(\ta)
\sim \ideaf(\ta\rho)$.
\end{proof}

\begin{proposition} \label{prop:invariant-idealisation}
  We assume that $\Pi$ satisfies frame opacity.
  Let $\rho$ be a renaming of annotations and $\ta$ be a well-formed
  annotated trace.
  If $\K(\ta) \LRstep{\ta} (\p_1;\phi_1)$
  and $\K(\ta\rho) \LRstep{\ta\rho} (\p_2;\phi_2)$,
  then we have that $\phi_1 \sim \phi_2$.
\end{proposition}

\begin{proof}
  We start by applying Proposition~\ref{prop:gc-pm} on the two given executions
  to obtain two executions starting from~$\pMa$:
\begin{itemize}
\item   $\pMa\lrstep{\ta^*}(\p_1';\phi_1)$  with $\ta^*\upto\ta$; and
\item  $\pMa\lrstep{\ta_\rho^*}(\p_2';\phi_2)$ with
$\ta_\rho^*\upto\ta\rho$.
\end{itemize}
Relying on frame opacity, we know that $\ideaf(\ta^*)$ (resp.
$\ideaf(\ta_\rho^*)$) is well-defined and $\ideaf(\ta^*) \sim \phi_1$
(resp. $\ideaf(\ta_\rho^*) \sim \phi_2$).

 Note that, if $\ta_1$ and $\ta_2$ are two annotated traces such that $\ta_1\upto\ta_2$ and
  $\fr_1$ is a name assignment, then
  $\ideaf^{\fr_1}(\ta_1)=\ideaf^{\fr_1}(\ta_2)$. 
Thanks to this
  remark, we
  easily deduce that $\ideaf(\ta)  \sim \ideaf(\ta^*)$
  and   $\ideaf(\ta\rho) = \ideaf(\ta_\rho^*)$.
Thanks to Proposition~\ref{prop:invariance-id-op}, we know that 
$\ideaf(\ta) \sim \ideaf(\ta\rho)$ and by transitivity of $\sim$, we
conclude that $\phi_1 \sim \phi_2$.
\end{proof}

\section{A sufficient condition for preserving executability}

We can now state a key lemma (\Cref{lem:executable-renaming}),
identifying a class of renamings which
yield indistinguishable executions.
More precisely, this lemma shows that for renamings satisfying some requirements,
if $\K(\ta)$ can execute an annotated trace $\ta$ then
$\K(\ta)\rho$ has an indistinguishable execution following the annotated trace
$\ta\rho$. Remark that, in the conclusion, the renaming is applied after building
the ground configuration ($\K(\ta)\rho$) instead of building the ground configuration
of the renamed trace ($\K(\ta\rho)$).
Both variants are a priori different.
However, in the final proof and in order to leverage previous propositions, we will need to
relate executions of $\K(\ta)$ with executions
of $\K(\ta\rho)$.
The following easy proposition bridges this gap.
We also state and prove a variant of \Cref{prop:gc-pm}, item (4)
when $\rho$ is applied after building the ground configuration.

\begin{proposition}
Let $\ta$ be a well-formed annotated trace and $\rho$ be  a renaming
of annotations.

If 
$\K(\ta)\rho\lrstep{\ta\rho}(\p';\phi)$
then
$\K(\ta\rho)\LRstep{\ta\rho}(\p'';\phi)$.
\label{prop:relate-rho}
\end{proposition}
\begin{proof}
Essentially, the proposition follows from the fact that there are less agents in sequence in
$\K(\ta\rho)$ than in $\K(\ta)\rho$, thanks to the third item of \Cref{def:renaming}.

More formally, by considering $\K(\ta\rho)$ and $\K(\ta)\rho$ as multiset of processes without sequence
(by removing all sequences and taking the union of processes), we obtain the same multisets.
Next, it suffices to prove that no execution is blocked by a sequence in $\K(\ta\rho)$.
By definition of the renaming $\rho$ (third requirement in \Cref{def:renaming}), if an agent $P\annot{\rho(a)}$ occurring
in $\K(\ta\rho)$ is in sequence with an agent $\rho(a')$ before him, then
$P\annot{a\rho}$ occurring in $(\K(\ta)\rho)$ must be in sequence with $a$ before him as well.
Hence, when a process $P\annot{\rho(a)};Q$ is available (\ie at top-level) at some point
in the execution from $\K(\ta)\rho$, then a similar process
$P\annot{\rho(a)};Q'$ is also available at the same point in the execution from $\K(\ta\rho)$.
However, a process may become available in the execution from 
$\K(\ta)\rho$ only after having performed rule \textsc{Seq},
while the same process may be
immediately available in the multiset in the execution from $\K(\ta\rho)$.
This is why we only obtain a weak execution
$\K(\ta\rho)\LRstep{\ta\rho}(\p'';\phi)$.
\end{proof}

\begin{proposition}
If $\ta=\ta_1.\ta_2$ is a well-formed annotated trace and $\K(\ta_1.\ta_2)\rho\lrstep{\ta_1\rho} K$ then
$\K(\ta_1)\rho\lrstep{\ta_1\rho}K'$ with $\phi(K) =\phi(K')$.
\label{prop:gc-pm:var}
\end{proposition}
\begin{proof}
  The argument is the same as for \Cref{prop:gc-pm}, item (4):
  the processes that are added to $\K(\ta_1)\rho$ when considering 
  $\K(\ta_1.\ta_2)\rho$ are unused in the execution of $\ta_1\rho$;
  moreover, the effect of $\rho$ on the processes of $\K(\ta_1)$
  is obviously the same as in $\K(\ta_1.\ta_2)$.
\end{proof}

Finally, after having defined the notion of {\em connection} between agents, we can state our
key lemma.

\begin{definition}
  Annotations $a$ and $a'$ are \emph{connected} in $(\ta,\phi)$ if
  they are associated in $(\ta_0,\phi)$ for some prefix $\ta_0$ of $\ta$
  that contains at least one $\taut$ action of an unsafe conditional
  annotated with either~$a$ or~$a'$.
\end{definition}

\begin{lemma}
\label{lem:executable-renaming}
We assume that $\Pi$ satisfies frame opacity and
item (i) of well-authentication.
Let $\ta$ be a well-formed annotated trace such that
$\K(\ta) \lrstep{\ta} (\p;\phi)$.
Let $\rho$ be a renaming of annotations.
Moreover, when $\fn(\ini)\cap\fn(\res)\neq\emptyset$ (shared case),
we assume that for any annotations $a$, $a'$, it holds that~$a$ and~$a'$ are connected in $(\ta,\phi)$, if, and only if,
$\rho(a)$ and $\rho(a')$ are dual.

We have that $\K(\ta)\rho \lrstep{\ta\rho} (\q;\psi)$
for some $\psi$ such that $\phi \sim \psi$.
\end{lemma}

%\stef{\bf Stef: Je tente ma version !}

\newcommand{\stefdiff}{\mathsf{choice}}

% \begin{lemma}
% \label{lem:executable-renaming}
% We assume that $\Pi$ satisfies frame opacity and
% item (i) of well-authentication.
% Let $\ta$ be a well-formed annotated trace such that
% $\K(\ta) \lrstep{\ta} (\p;\phi)$.
% Let $\rho$ be a renaming of annotations.
% Moreover, when $\fn(\ini)\cap\fn(\res)\neq\emptyset$ (shared case),
% we assume that for any annotations $a$, $a'$, it holds that~$a$ and~$a'$ are connected in $(\ta,\phi)$, if, and only if,
% $\rho(a)$ and $\rho(a')$ are dual.

% We have that $\K(\ta)\rho \lrstep{\ta\rho} (\q;\psi)$
% for some $\psi$ such that $\phi \sim \psi$.
% \end{lemma}

\begin{proof}
The ground configurations $\K(\ta)$ and $\K(\ta)\rho$ have the
  same shape: these processes only differ by their terms. Thus, we
  can put them together to form a bi-process\footnote{
    The bi-process considered here does not make use of the
    extension of diff-equivalence presented before:
    its inputs are of the form $\In(c,x)$ and not $\In(c,\stefdiff[x,y])$.
  }, \ie a process in which terms
  are bi-terms of the form $\stefdiff(t_1,t_2)$. Given a bi-process $B$, we denote $\fst{B}$
  (resp. $\snd{B}$) the process obtained from $B$ by replacing any
  bi-term $\stefdiff[t_1,t_2]$ by $t_1$ (resp. $t_2$). Moreover, we
  write $B \lrstep{\alpha[\stefdiff[a,a']]} B'$ to indicate
  that $\fst{B}$ executes $\alpha$ with annotation $a$ and $\snd{B}$
  performs $\alpha$ in the same way but with annotation $a'$.
  More generally, we
write $B \lrstep{\stefdiff[\ta,\ta']} B'$ to indicate that the
bi-process $B$ executes the trace $\ta$ on the left and $\ta'$ on the
  right, where the two traces differ only in their annotations.

    For the sake of simplicity, we decorate outputs of the biprocess obtained 
    from $\K(\ta)$ and $\K(\ta)\rho$ with pairwise distinct handles from 
    $\mathcal{W}$, and we assume that
    the handle that decorates an output will be used to store the output
    message when it will be executed.
  Lastly, we associate a vector of terms in $\T(\Sigma_\pub, \W \cup \X)$
  to each
safe conditional of the protocol:
recall that $\Let \; \vect z = \vect t \;\In \; P \; \Else \; Q$ is
identified as a safe conditional only if there exists a sequence
$\overline{R}$ of terms in $\T(\Sigma_\pub,
  \W \cup \X)$ such that $\overline{R}\{w_1 \mapsto u_1, \ldots, w_m \mapsto
  u_m\} =  \overline{t}$ where $u_i$ are  the messages used in outputs
  occurring before the conditional and $w_i$ is the handle
  associated to $u_i$; $\overline{R}$ is the vector
  of terms associated to the safe conditional. Note that
  $\overline{R}$ may contain variables from $\X$ corresponding to
  inputs performed before the
  conditional, which  will be instantiated by ground terms before
  the execution of the conditional.
When executing a process with labels on conditionals, we assume
  that the execution of an input $\In(c,x)$  with recipe $R_{in}$ will
  instantiate the 
  variable $x$ occurring in the label of the conditional with $R_{in}$.

 For any prefix
  $\K(\ta) \lrstep{\ta_0} (\p_0;\phi_0)$
  of the execution $\K(\ta) \lrstep{\ta} (\p;\phi)$,
  we prove that there exists an execution
  $\K(\ta)\rho \lrstep{\ta_0\rho} (\q_0;\psi_0)$ such that:
%\marginpar{\stef{petit souci: appliquer $\redc$ sur un terme avec
 %   variable n'est pas defini. On pourrait dire que c'est l'identité.}}
  \begin{itemize}
  \item[(a)] $B \lrstep{\stefdiff[\ta_0,\ta_0\rho]} B_0$ with $\fst{B_0} =
      (\p_0;\phi_0)$ and $\snd{B_0} = (\q_0;\psi_0)$. 
\item[(a')] Any
      bi-conditional $\Let \; \vect z = \stefdiff[\vect t_l, \vect t_r] \;\In \; B_P \; \Else \; B_Q$
      labeled with $\overline{R}$ is such that $\overline{R}=
      \overline{C}[R_1, \ldots, R_k]$ with $\overline{C}$ a sequence of
      contexts built on $\Sigma_\pub$, and $R_i \in \T(\Sigma_\pub, \W
      \cup \X)$ is either a variable in $\X$ or a
      $w \not\in \dom(\phi_0)$ or a recipe i.e. a term in
      $\T(\Sigma_\pub,\dom(\phi_0))$. Moreover, assuming that $u \redc
      u$ for any constructor term (even if it contains some variables), we have that $\overline{C}[R_1\phi^+_0\redc, \ldots,
      R_k\phi^+_0\redc] = \vect t_l$ and $\overline{C}[R_1\psi^+_0\redc, \ldots,
      R_k\psi^+_0\redc] = \vect t_r$ where $\phi^+_0$ (resp.
      $\psi^+_0$) is $\phi_0$ (resp. $\psi_0$)
      extended with $w \mapsto u$ for each output $\Out(c,u)$ decorated with
      $w$ preceding the
      conditional in $\fst{B_0}$ (resp. $\snd{B_0}$).
\item[(b)] $\phi_0\sim\psi_0$;
\item[(c$_1$)]
  when $\fn(\ini)\cap\fn(\res)=\emptyset$ (non-shared case),
  $\rho(a)$ and $\rho(a')$ are associated in $(\ta_0\rho,\psi_0)$
  if, and only if, $a$ and $a'$ are associated in $(\ta_0,\phi_0)$;
    \item[(c$_2$)]
      when $\fn(\ini)\cap\fn(\res)\neq\emptyset$ (shared case),
      $\rho(a)$ and $\rho(a')$ are associated in $(\ta_0\rho,\psi_0)$
      if, and only if, $a$ and $a'$ are associated in $(\ta_0,\phi_0)$ and
      connected in $(\ta,\phi)$.
  \end{itemize}

 \noindent We proceed by induction on the prefix~$\ta_0$ of $\ta$.

\smallskip{}

\noindent\emph{Case $\ta_0$ is empty.} In such a case, 
$\ta_0\rho$
can obviously be executed.
  Condition (b) is trivial since both frames are empty.
  In order to check conditions (c$_1$) and (c$_2$), note that
  association coincides with duality for empty traces. We start by
  establishing condition (c$_1$).
  Being dual simply means being distinct roles,
  hence one obviously has that $\rho(a)$ and $\rho(a')$ are dual
  if, and only if, $a$ and~$a'$ are. This allows us to conclude for
  condition (c$_1$).  Now, we establish 
 condition (c$_2$). 
  By hypothesis, we have that $a$ and $a'$ are connected in $(\ta,\phi)$
  if, and only if, $\rho(a)$ and $\rho(a')$ are dual, this allows us
  to conclude regarding one direction. Now, if $\rho(a)$ and
  $\rho(a')$ are dual, then $a$ and $a'$ are dual too by definition of
  an agent renaming. Hence, we have the other direction.
    Condition (a) holds  and condition (a') holds since by definition
    of being safe, we have the expected $\overline{R}$.

\smallskip{}

\noindent\emph{Case the prefix of $\ta$ is of the form $\ta_0 .
  \alpha$.}
  The action $\alpha$ may be an input, an output, a conditional (\ie $\taut$ or $\taue$).
  By (a), we know that there is a process in $\q_0$
  which is able to perform an action of the same nature.
We consider separately the case where $\alpha$ is an input, an output or 
  a conditional. In those cases
  $\alpha$ is necessarily annotated, say by $a$,
  and has been produced by a process annotated $a$ in~$\p_0$.
  By induction hypothesis we have $(\p_0;\phi_0)$ and $(\q_0;\psi_0)$
  and the following executions
  $$
  \K(\ta)\lrstep{\ta_0}(\p_0;\phi_0) \lrstep{\alpha\annot{a}}(\p_0';\phi_0')
  \text{ and }
  \K(\ta)\rho\lrstep{\ta_0\rho}(\q_0;\psi_0)
  $$
  satisfying all our invariants.
  Note that one has  $(\ta_0 .
  \alpha\annot{a})\rho=(\ta_0\rho).\alpha\annot{\rho(a)}$.
  Moreover,
  we have that $B \lrstep{\stefdiff[\ta_0,\ta_0\rho]} B_0$ with $\fst{B_0} =
  (\p_0;\phi_0)$ and $\snd{B_0} = (\q_0;\psi_0)$.
  Now, we have to prove that there exists $B'_0$ such that $B_0
  \lrstep{\stefdiff[\alpha,\alpha\rho]} B'_0$ with $\fst{B'_0} =
  (\p'_0;\phi'_0)$.

  \paragraph{\underline{Case where $\alpha$ is an output.}}
  We immediately
  have that $\q_0$ can perform $\alpha\annot{\rho(a)}$,
  on the same channel and with the same handle.
  We now have to check our invariants
  for $\ta_0 . \alpha\annot{a}$. Let $\Out(c,w)\annot{a}$ be $\alpha$.
    Conditions (a) is obviously preserved. Regarding condition (a'),
    it is easy to see that the result holds. The term that was added
    in $\phi_0^+$ (resp. $\psi^+_0$) is now present in $\phi'_0$
    (resp. $\psi'_0$).  

  Conditions (c$_1$) and (c$_2$) follow from the fact that association
  is not affected by the execution of an output:
  $\rho(a)$ and $\rho(a')$ are associated
  in $(\ta_0 . \alpha\annot{a})\rho$ % where $\alpha$ is an output,
  if, and only if,
  they are associated in $\ta_0\rho$,
  and similarly without $\rho$.
  Finally, we shall prove (b): $\phi_0'\sim\psi_0'$
  where $\phi_0'$ (resp. $\psi_0'$) is the resulting frame
  after the action $\alpha\annot{a}$ (resp. $\alpha\annot{a\rho}$).
  Applying \Cref{prop:gc-pm} item (4) on the execution before renaming
  and \Cref{prop:gc-pm:var} on the execution after renaming, one obtains
  $$
  \K(\ta_0.\alpha\annot{a})\lrstep{\ta_0.\alpha\annot{a}}(\p_0'';\phi_0')
  \text{ and }
  \K(\ta_0.\alpha\annot{a})\rho\lrstep{(\ta_0\rho).\alpha\annot{a\rho}}(\q_0';\psi_0').
  $$
  We finally conclude $\phi_0'\sim\psi_0'$
  using \Cref{prop:relate-rho} on the execution on the right
  and then \Cref{prop:invariant-idealisation}.

  \paragraph{\underline{Case where $\alpha$ is a conditional.}}
  We first need to make sure
  that the outcome of the test is the same for $a$ and $a\rho$.
  We let $\tau_x$ (resp. $\tau_y$) be the action produced
  by evaluating the conditional of~$a$ (resp. $a\rho$) and shall
  prove that $\tau_x=\tau_y$.
    We distinguish two cases depending on whether the conditional has
    a label (\ie has been identified as safe) or not.
  \begin{itemize}
  \item
If the conditional has a label $\overline{R}$, since this
  conditional is now at toplevel,  we know that $\phi^+_0 = \phi_0$
  and $\psi^+_0 = \psi_0$, and $\overline{R}=\overline{C}[R_1, \ldots,
  R_k]$ only contains variables from $\dom(\phi_0) = \dom(\psi_0)$. On the left, we have that the
  conditional will be evaluated to true iff $\overline{t}_l \redc$ is
  a message, i.e. iff $\overline{C}[R_1 \phi^+_0\redc, \ldots,
  R_k \phi^+_0\redc]\redc$ is a message, i.e. iff $\overline{C}[R_1, \ldots,
  R_k ]\phi_0\redc$ is a message and similarly on the right. Since
  $\phi_0 \sim \psi_0$, this allows us to conclude that the two
  conditionals have the same outcome.

  \item
    If the conditional is unsafe, we make use of
    Proposition~\ref{prop:then-iff-assoc} to show that
    the outcome of the conditional is the same on both sides.
    First, we deduce the following executions from
    \Cref{prop:gc-pm} item (4) applied on the execution before renaming
    and \Cref{prop:gc-pm:var} applied on the execution after renaming:
    $$
    \K(\ta_0.\tau_x\annot{a})\lrstep{\ta_0.\tau_x\annot{a}}(\p_0'';\phi_0)
    \text{ and }
    \K(\ta_0.\tau_y\annot{a})\rho\lrstep{(\ta_0\rho).\tau_y\annot{a\rho}}(\q_0';\psi_0).
    $$
    To infer $\tau_x=\tau_y$ from \Cref{prop:then-iff-assoc},
    it remains to prove that 
    $a$ and $a'$ are associated in $(\ta_0,\phi_0)$ if, and only if,
    $\rho(a)$ and $\rho(a')$ are associated in $(\ta_0\rho,\psi_0)$.
 When $\fn(\ini)\cap\fn(\res)=\emptyset$ (non-shared case),
    this is given by the invariant (c$_1$). Otherwise, when
    $\fn(\ini)\cap\fn(\res)\neq\emptyset$ 
(shared case), 
(c$_2$) gives us that $\rho(a)$ and $\rho(a')$ are associated in
$(\ta_0\rho,\psi_0)$ if, and only if, $a$ and $a'$ are associated in
$(\ta_0,\phi_0)$ and connected in $(\ta,\phi)$. Therefore, to
conclude, it is actually sufficient to show that if $a$ and $a'$ are associated in
$(\ta_0,\phi_0)$ then $\rho(a)$ and $\rho(a')$ are associated in
$(\ta_0\rho,\psi_0)$.
Since $a$
and $a'$ are associated in $(\ta_0,\phi_0)$, 
thanks to    Proposition~\ref{prop:then-iff-assoc},  we know that the outcome of the test
    will be positive (\ie $\tau_x=\taut$),
    and thus $a$ and $a'$ are connected in $(\ta_0.\tau_{\taut},\phi_0)$, and
    therefore $a$ and $a'$ are also connected in $(\ta,\phi)$. Thanks
    to (c$_2$) we have that $\rho(a)$ and $\rho(a')$ are associated in
      $(\ta_0\rho,\psi_0)$, and we are done.
 \end{itemize}
  After the execution of this conditional producing $\tau_x=\tau_y$,
  condition (a) and (a') obviously still hold since $\tau_x=\tau_y$
  implying that both agents go to the same branch of the conditional.
  Invariant (b) is trivial since frames have not changed.
  Conditions (c$_1$) and (c$_2$) are preserved
  because the association between~$a$ and~$a'$ is preserved if, and only if,
  the outcome of the test is positive,
  which is the same before and after the renaming.

  \paragraph{\underline{Case where $\alpha = \In(c,R_{in})$ is an input.}}
  We immediately have that $\q_0$ can perform $\alpha\annot{\rho(a)}$
  on the same channel and with the same recipe $R_{in}$ (since $\dom(\phi_0)=\dom(\psi_0)$
  follows from $\phi_0\sim\psi_0$).
   Conditions (a) and (b) are obviously preserved.
    Let us establish Condition (a').
We consider a bi-conditional $\Let \; \vect z = \stefdiff[\vect t_l, \vect t_r] \;\In \; B_P \; \Else \; B_Q$
      labeled with $\overline{R}$ before executing $\alpha$. Thus, we
      know that $(a')$ holds, i.e. $\overline{R} = \overline{C}[R_1,
      \ldots, R_k]$ such that $\overline{C}[R_1\phi^+_0\redc, \ldots,
      R_k\phi^+_0\redc] = \vect t_l$ and $\overline{C}[R_1\psi^+_0\redc, \ldots,
      R_k\psi^+_0\redc] = \vect t_r$.
After executing $\alpha$, either this conditional is kept unchanged
and the result trivially holds, or $\vect t_l$ becomes $\vect
t_l\{x\mapsto R_{in}\phi_0\redc\}$ and we have that the label of this
conditional is now $\overline{R}' = \overline{R} \{x \mapsto R_{in}\}$,
and we have that $R_{in}\phi_0\redc$ (and thus $R_{in}\psi_0\redc$) is
a message.
To conclude, it remains to show that $\overline{C}[R_1\{x \mapsto
R_{in}\}\phi^+_0\redc, \ldots, R_k\{x \mapsto R_{in}\}\phi^+_0\redc] =
\vect t_l\{x \mapsto R_{in}\phi_0\redc\}$ (and similarly for $\psi_0$).
Either $R_i = x$ and we have that $R_i\{x\mapsto R_{in}\}\phi^+_0\redc
= R_{in}\phi_0\redc = (R_i\phi^+_0\redc) \{x \mapsto R_{in}\phi_0\redc\}$. 
Otherwise, $R_i$ does not contain $x$, and we
have that $R_i\{x\mapsto R_{in}\}\phi^+_0\redc = R_i\phi^+_0\redc =
(R_i\phi^+_0\redc) \{x \mapsto R_{in}\phi_0\redc\}$.  

Thus, we have that
$\overline{C}[R_1\{x \mapsto
R_{in}\}\phi^+_0\redc, \ldots, R_k\{x \mapsto R_{in}\}\phi^+_0\redc]  = 
\overline{C}[(R_1\phi^+_0\redc)\{x \mapsto R_{in}\phi_0\}, \ldots,
(R_k\phi^+_0\redc)\{x \mapsto R_{in}\phi_0\}]
= 
\overline{C}[R_1\phi^+_0\redc, \ldots,
R_1\phi^+_0\redc]\{x \mapsto R_{in}\phi_0\redc\}
= \vect t_l \{x \mapsto R_{in}\phi_0\redc\}$.
This allows us to conclude.

  Conditions (c$_1$) and (c$_2$) are preserved because honest interactions
  are preserved by the renaming, since $\phi_0\sim\psi_0$ by 
  invariant (b).
  We only detail one direction of (c$_1$), the other cases being similar.
  Assume that $\rho(a)$ and $\rho(a')$ are associated
  in $((\ta_0.\alpha\annot{a})\rho,\psi_0)$.   
  The renamed agents $\rho(a)$ and $\rho(a')$ are also associated in $(\ta_0\rho,\phi')$,
  thus $a$ and $a'$ are associated in $(\ta_0,\phi')$.
  Now, because $\alpha$ did not break the association of
  $\rho(a)$ and $\rho(a')$ in $(\ta_0\rho,\psi_0)$,
  it must be that the input message in $\alpha = \In(c,R_{in})$
  corresponds to the last output of $\rho(a')$ in $\ta_0\rho$.
  Formally, if that last output corresponds to the handle $w$ in
  $\psi_0$, we have $R_{in}\psi_0\redc \theo w\psi_0$.
  But, because $\phi_0 \sim \psi_0$ by 
  invariant (b), % Proposition~\ref{prop:invariant-idealisation},
  we then also have $M\phi_0\redc \theo w\phi_0$ and
  the association of $a$ and $a'$ in $(\ta_0,\phi_0)$ carries
  over to $(\ta_0.\alpha\annot{a},\phi_0)$.
\end{proof}

%%%%%%%%%%%%%%%%%%%%%%%%%%%%%% FINAL PROOFS %%%%%%%%%%%%%%%%%%%%%%%%%%%%%%%%%%%%

\section{Final proof}

Thanks to \Cref{lem:executable-renaming}, we can transform any execution of
$\pMa$ into an indistinguishable execution of~$\pS$,
provided that an appropriate renaming of annotations exists.
In order to prove that such a renaming exists in \Cref{prop:single-session-renaming},
we show below that in the shared case, agents cannot be connected
multiple times.

\begin{proposition}
\label{prop:wa-doubleCo}
  Assume that $\Pi$ satisfies item (ii) of well-authentication
  and that ${\fn(\ini)\cap\fn(\res)}\neq\emptyset$ (shared case).
  Consider a well-formed annotated trace $\ta$ and an execution
  $ \K(\ta) \lrstep{\ta} (\p;\phi)$.
  Let $a$, $a_1$, and $a_2$ be three annotations such that
  $a$ and $a_1$ (resp. $a$ and $a_2$) are connected in $(\ta,\phi)$, 
then we have that
  $a_1=a_2$.
\end{proposition}

\begin{proof}
  Consider the first unsafe conditional performed in $\ta$ by one of
  the three agents $a$, $a_1$ or $a_2$.
  This conditional exists and must be successful, otherwise the agents would
  not be connected in~$\ta$.
  In other words, we have $\ta = \ta'.\taut[x].\ta''$
  where $x \in \{a,a_1,a_2\}$ and
  $\ta'$ does not contain any unsafe conditional performed by one of
  our three agents.
  Since $a$ is connected with both~$a_1$ and~$a_2$ in $\ta$,
  it is associated with both~$a_1$ and~$a_2$ in $\ta'$.
  This contradicts condition (ii) of well-authentication
  applied to $\ta'.\taut[x]$.
\end{proof}

\begin{proposition} \label{prop:single-session-renaming}
 Let $\Pi$ be a protocol satisfying item (ii) of
 well-authentication, and $\ta$ be a well-formed annotated trace such 
that
  $ \K(\ta) \lrstep{\ta} K$. There exists a renaming of annotations
  $\rho$ such that:
\begin{itemize}
\item $\ta\rho$ is single-session;
\item Moreover, when $\fn(\ini) \cap \fn(\res) \neq \emptyset$ (shared
  case), for any annotations $a$ and $a'$, we have that $a$ and $a'$
  are connected in $(\ta, \phi)$, if, and only if, $\rho(a)$ and
  $\rho(a')$ are dual.
\end{itemize}
\end{proposition}

\begin{proof}
  For $\vect k \in \mathrm{id}_I(\ta)\cup\mathrm{id}_R(\ta)$,
  we define $\con(\vect k)$ as follows:
  \begin{itemize}
  \item when $\fn(\ini)\cap\fn(\res)=\emptyset$ (non-shared case),
    we let $\con(\vect k)$ be the empty set.
    \item when $\fn(\ini)\cap\fn(\res)\neq\emptyset$ (shared case),
    we let $\con(\vect k)$ be the set of all
    $(\vect{n_1},\vect{n_2})$ such that $\aini(\vect k,\vect{n_1})$ and
    $\ares(\vect k,\vect{n_2})$ are connected in $(\ta,\phi(K))$;
  \end{itemize}
  Essentially, $\con(\vect k)$ denotes the set of pairs of (dual) sessions
  that the renaming to be defined should keep on the same identity.
  Applying \Cref{prop:wa-doubleCo}, we deduce that
  for any $\vect k \in \mathrm{id}_A(\ta)$
  and $(\vect{n_1},\vect{n_2}),(\vect{n_3},\vect{n_4})\in\con(\vect k)$,
  then either (i) $\vect{n_1}=\vect{n_3}$ and $\vect{n_2}=\vect{n_4}$
  or (ii) $\vect{n_1}\neq \vect{n_3}$ and $\vect{n_2}\neq\vect{n_4}$.

  Next, we assume the existence of a function
  $k^c : \N^*\times\N^*\times^*\N^*\mapsto \N^*$
  that associates to any sequence of names $(\vect{k},\vect{n_1},\vect{n_2})$
  a vector of % fresh [david: pas formel, et subsumé par disjointness]
  names of the length of identity parameters of $\Pi$:
  $\vect{k'} = k^c(\vect{k},\vect{n_1},\vect{n_2})$.
  These name vectors are assumed to be all disjoint and not containing
  any name already occurring in the annotations of $\ta$.
  This gives us a mean to pick fresh identity parameters for each combination 
  of $\vect{k},\vect{n_1},\vect{n_2}$ taken from the annotations of~$\ta$.
  We also assume a function
  $k^1$ % (\vect{k},\vect{n_1})$
  such that the vectors $k^1(\vect{k},\vect{n_1})$
  are again disjoint and not overlapping with annotations of~$\ta$
  and any $k^c(\vect{k'},\vect{n'_1},\vect{n'_2})$,
  and similarly for $k^2(\vect{k},\vect{n_2})$ which should also not
  overlap with $k^1$ vectors.
  These last two collections of identity parameters will be used to give
  fresh identities to initiator and responder agents, independently.
  We then define $\rho$ as follows:
  $$
  \begin{array}[]{rcll}
    \aini(\vect k,\vect n_1) & \mapsto & 
    \aini(k^c(\vect k,\vect{n_1},\vect{n_2}), \vect{n_1})
    &\text{ if }(\vect{n_1},\vect{n_2})\in\con(\vect k)
    \\
    & \mapsto& 
    \aini(k^1(\vect{k},\vect{n_1}), \vect{n_1})
    &\text{otherwise}
    \\[2mm]
    \ares(\vect k,\vect n_2) & \mapsto & 
    \ares(k^c(\vect k,\vect{n_1},\vect{n_2}), \vect{n_2})
    &\text{ if }(\vect{n_1},\vect{n_2})\in\con(\vect k)
    \\
    & \mapsto& 
    \ares(k^2(\vect{k},\vect{n_2}), \vect{n_2})
    &\text{otherwise}
    \\
  \end{array}
  $$
  We now prove that the renaming $\rho$ defined above satisfies all requirements.

  \paragraph*{The mapping $\rho$ is a renaming.}
  First, for any well-formed $\ta'$, the fact that $\ta'\rho$ is well-formed follows from the following:
  (i) session names are not modified and
  (ii) identity names are all pairwise distinct and never intersect except
  for agents $\rho(\aini(\vect{k},\vect{n_1}))$ and
  $\rho(\ares(\vect{k},\vect{n_2}))$ such that
  $(\vect{n_1},\vect{n_2})\in\con(\vect k)$ but there cannot be a third agent sharing the
  same identity names according to the result obtained above.
  The mapping $\rho$  is obviously role-preserving.
  Finally, if $\rho(a)$ and $\rho(a')$ share the same identity parameters then
  (a) $\fn(\ini)\cap\fn(\res)\neq\emptyset$ and
  (b) $a$ and $a'$
  are connected in $(\ta,\phi)$ and are thus dual implying that $a$ and $a'$ share the same
  identity parameters as well.
  
  \paragraph*{The renaming $\rho$ is single-session.} First, $\idzero$
  never occurs in the image of $\rho$. Second, 
all agents are mapped to agents having fresh, distinct identity parameters
  except for agents $a,a'$ that were connected in $(\ta,\phi$).
  However, as already discussed,
  in such a case, there is no third agent sharing those identity parameters
  and $a$ and $a'$ are necessarily dual.

\smallskip{}

  %\paragraph*{Hypothesis of \Cref{lem:executable-renaming}.}
  To conclude, we shall prove that, in the shared case,
  for any annotations $a$, $a'$, it holds that
  $a$ and $a'$ are connected in $(\ta,\phi)$, if, and only if,
  $\rho(a)$ and $\rho(a')$ are dual. This is actually a direct
  consequence of the definition of the renaming~$\rho$.
\end{proof}

\bigskip{}

We are now able to prove our main theorem.
\begin{proof}[{\bf Proof of \Cref{theo:main}}]
We have that $\pS \sqsubseteq \pM \sqsubseteq \pMa$, and we have to
establish that 
${\pMa \sqsubseteq \pS}$.
Consider an execution $\pMa \lrstep{\ta} (\p;\phi)$.
First, thanks to Proposition~\ref{prop:gc-pm}, there is an annotated
trace $\ta'\upto\ta$ such that 
$\K(\ta)\lrstep{\ta'}(\p';\phi)$.
Let $\rho$ be the renaming obtained in
Proposition~\ref{prop:single-session-renaming} for~$\ta'$.
By Lemma~\ref{lem:executable-renaming}, we have that 
$\K(\ta)\rho \lrstep{\ta'\rho} (\q;\phi_\rho)$ for some frame
$\phi_\rho$ such that $\phi_\rho\sim\phi$.
We then deduce from \Cref{prop:relate-rho} that 
$\K(\ta\rho) \LRstep{\ta'\rho} (\q';\phi_\rho)$, and thus
$\K(\ta\rho) \LRstep{\ta\rho} (Q';\phi_\rho)$.
Since $\ta\rho$ is single-session,
Proposition~\ref{prop:gc-pm} implies that 
$(\pS;\emptyset) \LRstep{\ta\rho} (\q'';\phi_\rho)$, and this allows
us to conclude.
\end{proof}

\end{document}